\documentclass[a4paper,12pt]{article}

\usepackage[ngerman, american]{babel} 
\usepackage[utf8]{inputenc}
\usepackage[a4paper, left=2.5cm, right=2.5cm, top=3cm, bottom=3cm]{geometry}
\usepackage{graphicx}
\usepackage{float}
\usepackage{color,psfrag}
\usepackage{amssymb}
\usepackage{amsmath} 
\usepackage{amsthm}
\usepackage{dsfont}
\usepackage{mathrsfs}
\usepackage{color}
\usepackage[automark,headsepline]{scrpage2}
\usepackage{setspace}
\usepackage{hyperref}

\newcommand{\cF}{\mathcal{F}}

\newcommand{\cH}{\mathcal{H}}

\newcommand{\cN}{\mathcal{N}}

\newcommand{\tH}{\widetilde{H}}
\newcommand{\tP}{\widetilde{P}}

\newcommand{\tW}{\widetilde{W}}

\newcommand{\C}{\mathds{C}}
\newcommand{\R}{\mathds{R}}
\newcommand{\N}{\mathds{N}}

\newcommand{\dist}{\mathop{\mathrm{dist}}}

\newcommand{\supp}{\mathop{\mathrm{supp}}}
\newcommand{\dom}{\mathop{\mathrm{dom}}}

\newcommand{\dx}{{\mathrm d}}
\newcommand{\I}{\mathds{1}}
\newcommand{\1}{\mathbf{1}}
\renewcommand{\L}{\mathrm{L}}

\renewcommand{\P}{\mathcal{P}}	

\newtheorem{theorem}{Theorem}[section]         
\newtheorem{lemma}[theorem]{Lemma}             
\newtheorem{corollary}[theorem]{Corollary}

\newtheorem{remark}[theorem]{Remark}

\numberwithin{equation}{section} 
\providecommand{\keywords}[1]{\textbf{\textit{Keywords: }}{\it #1}}
\begin{document}
\setcounter{page}{0} \thispagestyle{empty}
\title{Existence of Resonances for the Spin-Boson-Model with Critical Coupling Function}
\author{Jana Reker\\
{\normalsize Institut für Analysis und Algebra}\\
{\normalsize Technische Universität Braunschweig, Germany}\\
{\normalsize J.Reker@RekerNet.de}}
\date{April 20, 2018}
\maketitle

\begin{abstract}
A two-level atom coupled to the quantized radiation field is studied. In the physical relevant situation, the coupling function modeling the interaction between the two component behaves like $|k|^{-1/2}$, as the photon momentum tends to zero. This behavior is referred to as critical, as it constitutes a borderline case. Previous results on non-existence state that, in the general case, neither a ground state nor a resonance exists. Hasler and Herbst have shown \cite{HaslerHerbst2011}, however, that a ground state does exist if the absence of self-interactions is assumed. Bach, Ballesteros, K{\"o}nenberg, and Menrath have then explicitly constructed the ground state this specific case~\cite{BachBallesterosKoenenbergMenrath2017} using the multiscale analysis known as Pizzo's Method \cite{Pizzo2003}. Building on this result, the existence of resonances is considered. In the present paper, using multiscale analysis, a resonance eigenvalue of the complex deformed Hamiltonian is constructed. Neumann series expansions are used in the analysis and a suitable Feshbach-Schur map controls the exponential decay of the terms in the series.
\end{abstract}

\thispagestyle{empty}
\newpage

\section{Introduction}\label{introduction}
The spin-boson model describes the system of a two-level atom coupled to the quantized radiation field. It is modeled by a Hamiltonian $H_g$ acting on the Hilbert space
\begin{equation}\label{def-space}
\cH=\C^2\otimes\cF
\end{equation}
where $\C^2$ corresponds to the two-level atom and $\cF$ denotes the bosonic Fock space corresponding to the photon field. When studied separately, the components can be described by the Hamiltonians $H_{at}$ and $H_{ph}$ acting on $\C^2$ and $\cF$, respectively. The model additionally includes the interaction between these two parts. It is represented by a symmetric operator that acts on the tensor product of the two Hilbert spaces and, therefore, mixes both components. The Hamiltonian $H_g$ is self-adjoint and of the perturbative form
\begin{equation}\label{structure-hamiltonian}
H_g=H_{at}\otimes\I_{\cF}+\I_{\C^2}\otimes H_{ph}+gW=H_0+gW,
\end{equation}
with $\I$ denoting identity and $gW$ being the coupling term, which is interpreted as a perturbation of the operator $H_0$. The (small) constant $g>0$ scales the influence of this perturbation. There are several variations of this model including different assumptions on the coupling function involved in $W$, in particular, its behavior in the infrared limit~${|k|\rightarrow0}$. It may also be further generalized to an $N$-level atom for some $N>2$. In the specific model studied here, the coupling function involved in the operator $W$ is characterized by a dependence on the momentum $k$ that behaves like~$C|k|^{-\frac{1}{2}}$, as $k$ tends to zero. Although prescribed by physical principles, this behavior is often referred to as critical as it constitues a borderline case. Previous results on non-existence~(see \cite{AraiHirokawaHiroshima1999}) state that the generalized model without any kind of infrared cutoff or regularization does not possess a ground state or resonances. However, specifying the part of the coupling function acting on $\C^2$ to have no diagonal entries has been shown in \cite{HaslerHerbst2011} to be the key assumption for the existence of a ground state for the spin-boson model in the critical case. Thus, we assume the part of $W$ acting on $\C^2$ to be off-diagonal. Several results for the spin-boson model have been established using a modified, slightly more regular function that behaves as $C|k|^{-\frac{1}{2}+\mu}$, as~$k\rightarrow0$, for some small~${\mu>0}$~(see e.g.~\cite{BallesterosDeckertHaenle2018}). Note that, in this case, additional assumptions on the coupling function are not needed. Naturally, the study of the mathematical aspects of this and similar models describing non-relativistic matter such as atoms and molecules coupled to the quantized radiation field, known as non-relativistic quantum electrodynamics, is strongly motivated by physics. One important application of the spin-boson model in particular is found in quantum computing, where the two-level atom in the spin-boson model is interpreted as a qubit (see e.g.~\cite[Section~7.1]{BermanMerkliSigal2008}).

\bigskip
Generally speaking, the study of the spin-boson model is an analysis of the spectrum of the operator $H_g$ given in \eqref{structure-hamiltonian} with the eigenvalues of the Hamiltonian being interpreted as the energy levels of the system while the corresponding eigenvectors are interpreted as the states of the system with this energy. If the coupling constant is equal to zero, the unperturbed system described by $H_0$ is an atom and the radiation field which are not interacting with each other. The spectrum of $H_0$ is then given by sums $a+b$ with~$a$ being a spectral value of $H_{at}$ and~$b$ being a spectral value of $H_{ph}$. The resulting spectrum is sketched in the following image.

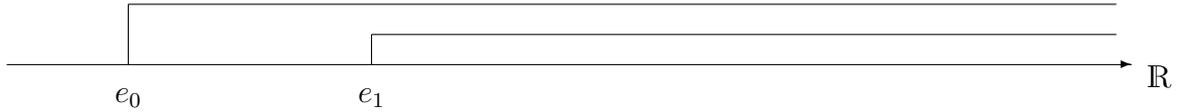
\begin{figure}[htb]	
\centering
\setlength{\unitlength}{4144sp}%
\begingroup\makeatletter\ifx\SetFigFont\undefined%
\gdef\SetFigFont#1#2#3#4#5{%
  \reset@font\fontsize{#1}{#2pt}%
  \fontfamily{#3}\fontseries{#4}\fontshape{#5}%
  \selectfont}%
\fi\endgroup%
\begin{picture}(6777,657)(1969,-8896)
\thinlines
{\color[rgb]{0,0,0}\put(1981,-8611){\vector( 1, 0){6660}}
}%
{\color[rgb]{0,0,0}\put(4141,-8611){\line( 0, 1){180}}
\put(4141,-8431){\line( 1, 0){4410}}
}%
{\color[rgb]{0,0,0}\put(2701,-8611){\line( 0, 1){360}}
\put(2701,-8251){\line( 1, 0){5850}}
}%
\put(2701,-8836){\makebox(0,0)[b]{\smash{{\SetFigFont{12}{14.4}{\familydefault}{\mddefault}{\updefault}{\color[rgb]{0,0,0}$e_0$}%
}}}}
\put(4141,-8836){\makebox(0,0)[b]{\smash{{\SetFigFont{12}{14.4}{\familydefault}{\mddefault}{\updefault}{\color[rgb]{0,0,0}$e_1$}%
}}}}
\put(8731,-8746){\makebox(0,0)[lb]{\smash{{\SetFigFont{12}{14.4}{\familydefault}{\mddefault}{\updefault}{\color[rgb]{0,0,0}$\R$}%
}}}}
\end{picture}%
\caption{The spectrum of $H_0$}\label{fig-spectrum-unperturbed}
\end{figure}

Both $e_0$ and $e_1$ are eigenvalues of $H_0$ with the corresponding eigenvectors given by the tensor product of the eigenvector corresponding to either the ground state or the excited state of the atom and the Fock vacuum, which is the eigenvector corresponding to the ground state of the photon field. When studying $H_g$, for some $g>0$, the fate of these two eigenvalues is of particular interest. From the physics point of view, it corresponds to the questions whether the coupled system itself possesses a ground state and whether the excited states of the unperturbed system turn into so-called metastable states when the atom and the radiation field are coupled to one another. In the case that an affirmative answer is given to the latter question, the corresponding eigenvalue is called a resonance. One may then further study the life-time of these states making a connection to Fermi's golden rule (see e.g.~\cite{Westrich2011}). As the two-level atom in the spin-boson model can be interpreted as a qubit, comparing the expected life-time of the metastable state to the minimal time necessary to perform computations is a question of interest in this context. Naturally, there are also other physical factors to be taken into account. Here, the spin-boson model is used, too, to approach certain questions from a theoretical point of view~(see e.g.~\cite{BermanMerkliSigal2008}).

\bigskip
For the model studied here, the existence of a unique ground state has already been established by Bach, Ballesteros, Könenberg, and Menrath in \cite{BachBallesterosKoenenbergMenrath2017} using a multiscale analysis approach known as Pizzo's method (see \cite{Pizzo2003}). The existence of a resonance is the topic of this paper. There are several different definitions fo a resonance (see e.g. \cite{Simon1973}). In this paper, a resonance is defined for the following analysis to be an eigenvalue of a complex deformation $H_g(\theta)$ of the Hamiltonian $H_g$. Note that, with the eigenvalue of the unperturbed operator immersed in the continuum, the classical perturbation theory developed for isolated eigenvalues is not applicable. Additionally, when the perturbation is included, the value in the midst of the spectrum can no longer be identified. Here, applying complex deformation to the unperturbed Hamiltonian allows us to tilt the two rays by a small angle $\vartheta=\mathrm{Im}(\theta)$ into the lower half of the complex plane, thus exposing the eigenvalue at the tip of a continuum~(see Fig.~\ref{fig-spectrum-tilted}).\\

\begin{figure}[h]	
\centering
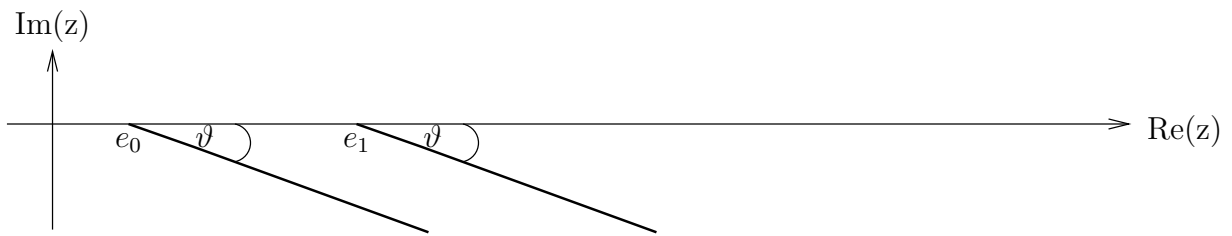
\caption{The spectrum of $H_0(\theta)$}\label{fig-spectrum-tilted}
\end{figure}

Note that the spectrum now also includes complex values and that the self-adjointness of the Hamiltonian is lost in the process. A method to cope with the absence of the spectral theorem and functional calculus is Neumann series expansion, which is used at many points throughout the analysis. There are different ways to study the spectrum of $H_g(\theta)$. One possible approach is renormalization group analysis~(see e.g.~\cite{BachFroehlichSigal1998a},~\cite{BachFroehlichSigal1998b},~\cite{GriesemerHasler2009}) based on the isospectral Feshbach-Schur map, which has also been used to establish the results in~\cite{HaslerHerbst2011}. The approach chosen here, however, is based on the proof for the existence of a ground state given in \cite{BachBallesterosKoenenbergMenrath2017}, i.e., multiscale analysis. It is achieved by introducing a sequence of so-called infrared-regularized Hamiltonians~${H_g^{(n)}\negmedspace(\theta)}$ including an infrared cutoff which becomes smaller, as~$n$ increases. This tool allows to artificially introduce gaps in the spectrum isolating the eigenvalue in question. The effect on the spectrum of the unperturbed Hamiltonian is sketched in the following image.

\begin{figure}[h]		
\centering
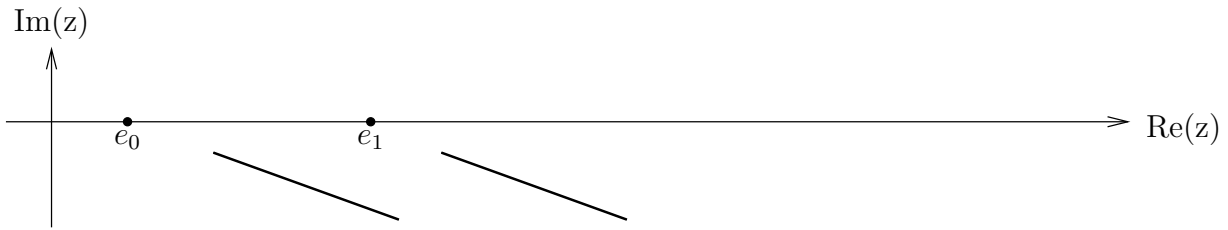
\caption{The spectrum of $H_0^{(0)}\negmedspace(\theta)$}\label{fig-spectrum-tilted-cutoff}
\end{figure}

Note that the infrared cutoff generates a small neighborhood around $e_1$ that does not include any other spectral values. The gaps, therefore, make it possible to define projections $P^{(n)}$ to the corresponding eigenspace using operator-valued integrals which can inductively be shown to be rank-one. The existence of the eigenvalue is then derived from the convergence of the sequence of projections $(P^{(n)}\otimes P_{\Omega^{(n,\infty)}})_n$, where $P_{\Omega^{(n,\infty)}}$ is the projection to the vacuum vector of the appropriate Fock space such that the resulting projection acts on $\cH$, as well as the convergence of the sequence of eigenvalues~$(E_g^{(n)})_n$ of the infrared-regularized Hamiltonians. In the ground state case, the construction of both the energies and the projections can be done directly, i.e., non-inductively~(see~\cite{BachBallesterosKoenenbergMenrath2017}) with the main tools being standard results from functional analysis like the spectral theorem. The convergence of the projections is then established using that the model possesses a symmetry which can explicitly be employed in the form of the equation
\begin{equation}\label{useful-symmetry-tool}
P^{(n)}\sigma_1P^{(n)}=0,
\end{equation}
originally due to \cite{HirokawaHiroshimaLorinczi2014}. This is a direct cosequence of the off-diagonal coupling and allows eliminating certain singular terms arising during the analysis. Note, however, that this direct approach can not be applied when studying resonances. As the infrared-regularized Hamiltonians are not self-adjoint as a result of the complex dilation, the resolvents need to be constructed inductively using repeated Neumann series expansion. Due to the critical coupling function involved, this results in an exponential growth in the bound for the norm of the resolvents that has to be balanced out in order to allow the construction. To approach this difficulty, an inductive argument similar to~\cite{BachBallesterosPizzo2017} is used. Here, the main tool in establishing the existence of the projections and their convergence is a suitable Feshbach-Schur map which provides stronger bounds in the Neumann series expansions. This allows to construct a resonance eigenvalue and, therefore, show its existence, which is the main result of this paper.

\begin{figure}[h]	
\centering
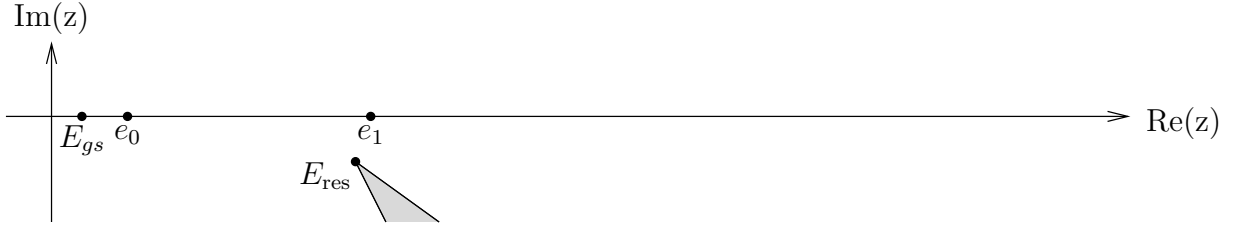
\caption{The ground state and the resonance eigenvalue in the complex plane}\label{fig-spectrum-aim}
\end{figure}

Note that $E_{\mathrm{res}}=E_{\mathrm{res}}(g)$ and that the distance between the unperturbed eigenvalue $e_1$ and the resonance $E_{\mathrm{res}}$ can be estimated in terms of the coupling constant.
%
\subsection{The Model}
A two-level atom coupled to the quantized radiation field is considered. It is assumed that the ground state energy of the atom is given by $0$ and that the energy corresponding to the excited state of the atom equals $2$ such that the Hamiltonian of the two-level atom acting on~$\C^2$ is given by
\begin{equation}\label{def-hamiltonian-atom}	
H_{at}:=\begin{pmatrix} 2&0\\0&0 \end{pmatrix}.
\end{equation}
Note that this choice of eigenvalues of $H_{at}$ is somewhat arbitrary, as one may choose any other two real numbers to represent the states of the atom by rescaling the energy. The Hilbert space corresponding to the radiation field is the boson Fock space associated with~$\L^2(\R^3)$, which is given by
\begin{equation}		
\cF=\cF\bigl[\L^2(\R^3)\bigr]=\bigoplus_{n=0}^{\infty}\cF_n,\ \mathrm{with}\ \cF_0=\C,\ \cF_n=\otimes_s^n\L^2(\R^3)
\end{equation}
Here, $\L^2(\R^3)$ denotes the space of square-integrable functions on $\R^3$ and $\otimes_s^n$ the n-fold symmetrized tensor product. Further, let $\Omega\in\cF$ denote the vacuum vector. Note that the property of polarization is not reflected in the choice of the Hilbert space and that, therefore, referring to the quanta of the field as photons is not completely accurate. The Hamiltonian of the radiation field is given by
\begin{equation}\label{def-hamiltonian-field}	
H_{ph}(\omega)=\dx \Gamma(\omega)=\int_{\R^3}\omega(k)\ a^*(k)\ a(k)\ \dx^3k,
\end{equation}
with $\omega(k)=|k|$, and $a^*(k)$, $a(k)$ denoting the creation and annihilation operators representing the canonical commutation relations (''CCR'')
\begin{equation}\label{ccr}		
\bigl[a(f),a^*(g)\bigr]=\langle f,g\rangle_{\L^2}\I_{\cF}, \quad \bigl[a(f),a(g)\bigr]=\bigl[a^*(f),a^*(g)\bigr]=0,\quad a(k)\Omega=0.
\end{equation}
on $\cF$. Note that \eqref{ccr} is to be understood in the sense of operator-valued distributions. In \eqref{def-hamiltonian-field}, we used Nelson's notation for the second quantization. The Hilbert space of the atom-photon system is now given by
\begin{equation}\label{def-hilbert-space}		
\cH=\C^2\otimes\cF.
\end{equation}
Finally, the coupling term is defined as
\begin{equation}\label{def-coupling-term}			
gW=\Phi(G)=\int_{\R^3}[G(k)\otimes a^*(k)+G^*(k)\otimes a(k)]\ \dx^3k,
\end{equation}
where the function $G$ is given by
\begin{equation}\label{coupling-function-general}	
G(k)=g\frac{f(k)}{\sqrt{\omega(k)}}\begin{pmatrix} 0&b\\\bar{b}&0 \end{pmatrix}.
\end{equation}
Here, $g$ denotes the coupling constant, $b\neq0$ is a complex number and $f$ is chosen to be an entire function with rapid decay on the real axis that satisfies~${|f(e^{ia}k)|=|f(e^{-ia}k)|}$ for all~$a\in\R$, e.g. $f(k)=\exp(-\frac{k^2}{\Lambda^2})$, with some fixed $\Lambda>0$ denoting the ultraviolet-cutoff. Note that the matrix involved in the coupling function is off-diagonal and, in the case that $b=re^{i\beta}$, for some $\beta\in[0,2\pi)$ and $r>0$, similar to the first Pauli matrix
\begin{equation}\label{def-sigma-1}	
\sigma_1=\begin{pmatrix} 0&1\\1&0 \end{pmatrix}
\end{equation}
with the similarity transformation given by a rotation. Since the matrix defining the atom Hamiltonian is diagonal, this transformation does not affect the rest of the model. In the case that $|b|=r\neq1$, the additional factor $r$ may be absorbed into $g$, thus allowing to write \eqref{coupling-function-general} as
\begin{equation}\label{coupling-function-simplified}	
G(k)=g\frac{f(k)}{\sqrt{\omega(k)}}\ \sigma_1.
\end{equation}
With the matrix occurring in the coupling term identified, we omit $\sigma_1$ and refer to the remaining real-valued function as $G$. Taking small values of $k$, observe that
\begin{equation}\label{critical-infrared-divergence}	
G(k)=g\frac{f(k)}{\sqrt{\omega(k)}}\sim c|k|^{-1/2}, \mathrm{\ as\ }k\rightarrow0,
\end{equation}
i.e., the model includes a critical coupling function. Putting the parts for the atom, the radiation field, and the interaction together yields the Hamiltonian for the model which is given by
\begin{equation}\label{def-hamiltonian-complete}	
H_g=\begin{pmatrix} 2&0\\0&0 \end{pmatrix}\otimes \I_{\cF}+\I_{\C^2}\otimes H_{ph}(\omega)+\Phi(G).
\end{equation}
Note that both $a(G)$ and $a^*(G)$ are $H_{ph}$-bounded with relative bound zero and that the interaction, therefore, is an infinitisimal perturbation in the sense of Kato. This implies that $H_g$ is a closed operator and, since $gW$ is symmetric, the Kato-Rellich theorem also implies that $H_g$ is also self-adjoint on its domain~$\dom(\I_{\C^2}\otimes H_{ph})$ (see \cite[p.~191]{Kato1980} and~\cite[Theorem~X.12]{ReedSimonII1975}), where
\begin{equation}\label{dom-h-ph}		
\dom(H_{ph})=\Big\{\Phi\in\cF\Big|\sum_{j=1}^{\infty}\displaystyle{\int}_{\R^{3j}}\Bigr(\sum_{m=1}^j|k_m|\Bigl)^2|\Phi_j(k_1,\dots,k_j)|^2\dx k_1\dots\dx k_j<\infty\Bigr\}.
\end{equation}
When studying resonances, a complex deformation of this operator is considered which is introduced in Section~\ref{sect-dilation-analyticity}. It is given by
\begin{equation}\label{def-deformed-h-g}
H_g(\theta):=H_{at}\otimes\I_{\cF}+\I_{\C^2}\otimes\bigl(e^{-\theta}H_{ph}\bigr)+gW_{\theta},
\end{equation}
where $\theta=i\vartheta$ for some sufficiently small $\vartheta\in\R$. The coupling function occuring in the transformed coupling term $W_{\theta}$ is denoted as $G_{\theta}$.
%
\subsection{Main Result and Outine of its Proof}\label{sect-sketch}

The main result of this paper is the existence of a resonance in the model presented previously. It can be formulated as follows: 

\begin{theorem}[Existence of Resonances]\label{the-main-result}		
Let $\vartheta$ be small enough to satisfy the assumptions of Lemma~\ref{analytic-continuation}, i.e.,~\eqref{final-restriction-theta} holds, and choose the parameters $\rho_0,\gamma,g>0$ sufficiently small such that~\eqref{final-restriction-rho-0},~\eqref{final-restriction-gamma} and~\eqref{final-restriction-g} are satisfied. Then there exists an eigenvalue~${E_{\mathrm{res}}\in D(2,Cg)}$ of $H_g(\theta)$.
\end{theorem}

The main part of the proof consists of an inductive argument, which is carried out in Sections~\ref{sect-ind-basis},~\ref{sect-ind-hypothesis}, and~\ref{sect-ind-step}. Theorem~\ref{the-main-result} is a then a direct result of Theorem~\ref{limit-is-eigenvalue}. The location of the eigenvalue in the lower half of the complex plane is ensured by an argument similar to~\cite{Westrich2011} and~\cite{BachFroehlichSigal1999}), thus implying that the eigenvalue is indeed a resonance. Note that non-degenracy is not discussed here. We plan, however, to investigate this question further in the future. When considering resonances, note that the initial situation seems fairly similar to proving the existence of a ground state once the Hamiltonian has been transformed using the complex dilatation as one has to identify an eigenvalue at the tip of a con\-ti\-nu\-um. Recalling~\eqref{def-deformed-h-g}, however, the transformed Hamiltonian~$H_g(\theta)$ considered here is no longer self-adjoint due to it possessing complex spectral values. For this reason, the numbers used to approximate the resonance are not necessarily real and, although the general approach can be transferred from the ground state case (see~\cite{BachBallesterosKoenenbergMenrath2017}), the infrared-regularized operators, as well as their respective resonance energies and the corresponding projections need to be constructed inductively in the multiscale analysis.

\bigskip
We start by  considering the infrared cutoff scale and the resulting sequence of infrared-regularized Hamiltonians. For $n\in\N_0$ set
\begin{equation}\label{def-rhos}		
\rho_n=\rho_0\cdot\gamma^n,
\end{equation}
where the parameters $\rho_0,\gamma\in(0,1)$ are to be chosen appropriately later on. The sequence~$(\rho_ n)_{n\in\N_0}$ is now used to define a sequence of balls $B_n=B(0,\rho_n)$ centered about the origin of $\R^3$ with respective radius $\rho_n$. Cutting the energies below~$\rho_n$ from the dispersion and the coupling function by
\begin{equation}\label{def-cutoff-functions}	
\omega^{(n)}(k):=\1_{\R^3\setminus B_n}(k)\omega(k),\quad G_{\theta}^{(n)}(k)=\1_{\R^3\setminus B_n}(k)G_{\theta}(k),
\end{equation}
and inserting these functions into the transformed Hamiltonian now yields
\begin{align}		
H_{ph}^{(n)}&:=H_{ph}\bigl(\omega^{(n)}\bigr),\label{def-h-ph-n}\\
W_{\theta}^{(n)}&:=\sigma_1\otimes\int_{\R^3}\Bigl(G_{\theta}^{(n)}a^*(k)+\overline{G_{\bar{\theta}}^{(n)}}a(k)\Bigr)\ \dx^3k,\label{def-w-theta-n}\\
H_g^{(n)}\negmedspace(\theta)&:=H_{at}\otimes \I_{\cF}+\I_{\C^2}\otimes\bigl(e^{-\theta}H_{ph}^{(n)}\bigr)+gW_{\theta}^{(n)},\label{def-h-g-theta-n}
\end{align}
where the infrared cutoff scale $(n)$ is explicitly displayed in the superscript. For the following considerations, the $\theta$-dependence in~\eqref{def-h-g-theta-n} is omitted in the notation, unless it is explicitly used. The operators $H_g^{(n)}$ act on the Hilbert spaces
\begin{equation}\label{def-cutoff-spaces}	
\cH^{(n)}:=\C^2\otimes\cF^{(n)},\quad \cF^{(n)}:=\cF[L^2(\R^3\setminus B_n)].
\end{equation}
Note that including the infrared cutoff neither changes the general properties of the operators established at the end of Section~\ref{sect-dilation-analyticity} nor requires a different argument to obtain them. In particular, for every $n\in\N_0$, $H_{ph}^{(n)}$ is self-adjoint, $H_0^{(n)}$ is normal and $H_g^{(n)}$ is closed and an infinitisimal perturbation of $H_0^{(n)}$. As~$H_g^{(n)}$ and~$H_g^{(n+1)}$ do not act on the same Hilbert space, we introduce
\begin{equation}\label{def-h-n-tilde}		
\tH_g^{(n)}:=H_g^{(n)}\otimes\1_{\cF^{(n)}}+\1_{\cH^{(n)}}\otimes H_{ph}^{(n,n+1)}.
\end{equation}
Here, the operator
\begin{equation}\label{def-h-ph-bridge}	
H_{ph}^{(n,n+1)}:=H_{ph}(\1_{B_n\setminus B_{n+1}}\omega)
\end{equation}
acts on the Hilbert space
\begin{equation}\label{def-spaces-bridge}	
\cH^{(n)}\otimes\cF^{(n,n+1)},\quad \cF^{(n,n+1)}=\cF[L^2(B_n\setminus B_{n+1})].
\end{equation}
Using that there are isomorphisms such that
\begin{equation}\label{isomorphisms-tilde-to-n+1}	
\cF^{(n+1)}\cong\cF^{(n)}\otimes\cF^{(n,n+1)},\quad \cH^{(n+1)}\cong\cH^{(n)}\otimes\cF^{(n,n+1)},
\end{equation}
one can conclude that $\tH_g^{(n)}$ acts on $\cH^{(n+1)}$ (see \cite[Section~6.2.1]{BachBallesterosPizzo2017}). This allows to compare the operators at successive energy scales. We define the operators $W_{\theta}^{(n,n+1)}$ analoguosly.

\bigskip
Next, the desired eigenvalue is established for the infrared-regularized Hamiltonians. We use an argument similar to \cite{BachBallesterosPizzo2017}. First, the operator $H_g^{(0)}$ including a fixed initial infrared cutoff $\rho_0$ is considered, which serves as the induction basis. Its resolvent is constructed from~$(H_0^{(0)}-z)^{-1}$ using a Neumann series expansion~(Lemma \ref{invertibility-0}). This allows to define a projection~$P^{(0)}$ by
\begin{equation}\label{def-p-0}		
P^{(0)}:=\frac{-1}{2\pi i}\int_{\gamma^{(0)}_{-1}}\frac{\dx z}{H_g^{(0)}-z}
\end{equation}
with $\gamma^{(0)}_{-1}$ being the curve along a circle of suitable radius chosen on the scale of $\rho_0$ centered at~${E^{(-1)}:=2}$. Here,~$E^{(-1)}$ denotes the eigenvalue of the unperturbed Hamiltonian~$H_0^{(0)}$ which is taken as an initial value for the construction. Considering
\begin{equation}\label{def-psi-0}		
\Psi^{(0)}:=P^{(0)}\Psi^{(-1)}, \quad \Psi^{(-1)}=\Bigl(\begin{pmatrix}1\\0\end{pmatrix}\otimes\Omega_0\Bigr)
\end{equation}
with $\Psi^{(-1)}$ being the eigenvector corresponding to $E^{(-1)}$, we calculate $\|\Psi^{(-1)}-\Psi^{(0)}\|$ by comparing the projections $P^{(0)}$ and $P^{(-1)}$. Here, $P^{(-1)}$ denotes the projection to the eigenspace of $H_0^{(0)}$ corresponding to~$E^{(-1)}$. As
\begin{equation}		
\big\|P^{(-1)}-P^{(0)}\big\|<1,
\end{equation}
the existence of an eigenvalue $E_g^{(0)}$ of $H_g^{(0)}$ in the interior of $\gamma^{(0)}_{-1}$ is implied and it follows that $P^{(0)}$ is a rank-one projection. By the choice of the radius of $\gamma^{(0)}_{-1}$ as well as $g$ and $\gamma$, it is possible to shift the center of the contour from $E^{(-1)}$ to $E^{(0)}$ and to consider the next operator. Note, however, that this construction can not be inductively continued without adjustments due to the critical coupling function involved. As a consequence of the inductive approach, neither the resolvent nor the eigenvalue can be stated without the result of the previous step. With Neumann series expansion needed in every single step to show invertibility, the factors obtained from the convergent geometric series accumulate in the norm bound, implying exponential growth of one of the terms involved in the estimate. When considering the resolvents directly, i.e., without additional tools, this either forces an adjustment of the coupling constant in every single step finally leading to the assumption $g=0$, or the iteration is stopped after a finite number of steps due to the divergence of the Neumann series.

\bigskip
The key part of the induction step is showing the existence of the resolvent of~$H_g^{(n+1)}$. The above difficulty is addressed by considering a suitable Feshbach-Schur map and using Neumann series expansion to establish the existence of the inverse of the Feshbach-Schur map of $H_g^{(n+1)}$ rather than the original operator. The existence of the resolvent is then implied by isospectrality (see e.g.~\cite[Theorem~IV.1]{BachFroehlichSigal1998a}). By considering the Feshbach-Schur map, a stronger bound is obtained that allows to compensate the exponential growth and establish the invertibility for the induction step. Rewriting both~${(H_g^{(n+1)}-z)^{-1}}$ and~${(\tH_g^{(n)}-z)^{-1}}$ using the formulae for the Feshbach-Schur map, the difference of the resolvents is estimated~(Theorem~\ref{estimate-difference-resolvents}). As the term is equal to the integrand when considering the difference of the corresponding projections, one can directly derive a norm bound~(Corollary~\ref{estimate-difference-projections}), which yields the desired exponential decay and allows to identify~$P^{(n+1)}$ as rank-one.
Thus, the existence of an eigenvalue~$E_g^{(n+1)}$ of~$H_g^{(n+1)}$ in the complex vicinity of $E_g^{(n)}$ is implied. In the next step, the distance between~$E_g^{(n+1)}$ and~$E_g^ {(n)}$ is estimated, yielding another exponential bound~(Theorem~\ref{estimate-difference-energies}). The induction step is completed by showing a norm bound for~$(H_g^{(n+1)}-z)^{-1}\overline{P^{(n+1)}}$ (Theorem~\ref{resolvent-norm-2}). Here, the projection to the complement of the eigenspace involved in the term allows to also extend the estimate into the safety zone around the eigenvalue considered previously.

\bigskip
With these estimates at hand, the infrared-limit can be constructed~(Corollary~\ref{convergences}). Making use of the exponential bounds shown earlier, the convergence of the sequence of resonance energies~$(E_g^{(n)})_n$ as well as the convergence of~$(P^{(n)}_{\infty})_n$ may be deduced. Here, the projections~$P^{(n)}_{\infty}$ are defined for $n\in\N_0$ by
\begin{equation}\label{def-p-n-infty}		
P^{(n)}_{\infty}:=P^{(n)}\otimes\Omega^{(n,\infty)}
\end{equation}
with $\Omega^{(n,\infty)}$ being the Fock vacuum in $\cF[L^2(B_n)]$. Applying a similar argument as for the ground state in~\cite{BachBallesterosKoenenbergMenrath2017}, one finds that~$\lim_{n\rightarrow\infty}P^{(n)}_{\infty}$ projects to an eigenspace of~$H_g(\theta)$ and that the corresponding eigenvalue~$E_{\mathrm{res}}$ is given by~$\lim_{n\rightarrow\infty}E_g^{(n)}$~(Theorem~\ref{limit-is-eigenvalue}).
\newpage
\section{Properties of the Model}\label{sect-model}
In this section, the main tools for the analysis are presented. The complex dilation used is introduced and a suitable Feshbach-Schur map is considered. We start by noting the following lemma (see e.g. \cite[Lemma~2.1]{BachBallesterosKoenenbergMenrath2017}), which is frequently used throughout the analysis and constitutes the basis of many estimates in the following sections.

\begin{lemma}\label{standard-estimate}		
Let $\rho>0$ and $G\in L^2(\R^3)$ such that $\frac{G}{\sqrt{\omega}}\in L^2(\R^3)$, then
\begin{equation}\label{formula-standard-estimate}	
\Big\|\Phi(G)\bigl(H_{ph}(\1_{\supp(G)}\omega)+\rho\bigr)^{-\frac{1}{2}}\Big\|\leq2\Bigl(\Big\|\frac{G}{\sqrt{\omega}}\Big\|_{L^2}+\rho^{-\frac{1}{2}}\|G\|_{L^2}\Bigr)=:2\|G\|_{\rho},
\end{equation}
with $\Phi(G)$ being the coupling term defined in \eqref{def-coupling-term} and $\1_\mathcal{A}$ denoting the characteristic function of a set $\mathcal{A}\subset\R^3$.
\end{lemma}

\begin{proof}		
Start by taking $\psi\in\dom\bigl(H_{ph}(\1_{\supp(G)}\omega)\bigr)$ and noting
\begin{equation}
\|\Phi(G)\psi\|\leq\|a(G)\psi\|+\|a^*(G)\psi\|.
\end{equation}
Using the Cauchy-Schwarz inequality, the first term can be estimated by
\begin{align}	
\Big\|\int G^*(k)a(k)\psi\dx^3k\Big\|^2
&\leq\Biggl(\int\1_{\supp(G)}(k) |G(k)|\cdot \|a(k)\psi\|\ \dx^3k\Biggr)^2\nonumber \\
&\leq\Biggl(\int\1_{\supp(G)}(k) |G(k)| \frac{\sqrt{\omega(k)}}{\sqrt{\omega(k)}}\|a(k)\psi\|\ \dx^3k\Biggr)^2\nonumber \\
&\leq\Biggl(\int\frac{|G(k)|^2}{\omega(k)}\dx^3k\Biggr)\Biggl(\int\bigl(\1_{\supp(G)}(k)\bigr)^2\omega(k)\|a(k)\psi\|^2\ \dx^3k\Biggr)\nonumber\\
&=\Big\|\frac{G}{\sqrt{\omega}}\Big\|_{L^2}^2\Biggl(\int\bigl(\1_{\supp(G)}(k)\bigr)^2\omega(k)\|a(k)\psi\|^2\ \dx^3k\Biggr),\label{standard-estimate-calculation-1}
\end{align}
where the last term on the right side can be identified as $\|(H_{ph}(\1_{\supp(G)}\omega))^{\frac{1}{2}}\psi\|^2$ as
\begin{equation}\label{identifying-h-ph}	
\|a(k)\psi\|^2=\langle a(k)\psi,a(k)\psi\rangle=\langle\psi,a(k)^*a(k)\psi\rangle,
\end{equation}
by the definition of the creation and annihilation operators, which, together with the dispersion and the continuity of the inner product, allows to make the connection to the photon field Hamiltonian. The estimate for $\|a^*(G)\psi\|$ now follows using the canonical commutation relations and the inequality already in place
\begin{align}		
\|a(G)^*\psi\|^2&=\big\langle\psi,a(G)a(G)^*\psi\big\rangle=\big\langle\psi,\bigl(\bigl[a(G),a(G)^*\bigr]+a(G)^*a(G)\bigr)\psi\big\rangle\nonumber\\
&\leq\|G\|_{\L^2}^2\cdot\|\psi\|^2+\Big\|\frac{G}{\sqrt{\omega}}\Big\|_{\L^2}^2\cdot\big\|H_{ph}^{1/2}\psi\big\|.\label{standard-estimate-calculation-2}
\end{align}
Taking $\psi=(H_{ph}+\rho)^{-\frac{1}{2}}\phi$ now yields the claim since $-\rho<0$ is an element of the resolvent set of $H_{ph}$.
\end{proof}

\subsection{Dilation Analyticity}\label{sect-dilation-analyticity}
In this section, the complex deformation is introduced and applied to the Hamiltonian, leading to a new operator $H_g(\theta)$ which may be generalized to complex $\theta$ with sufficiently small imaginary part using analytic continuation. The properties of the new operator are then briefly discussed, establishing all results needed for considering the resonances in the following chapter. To obtain a unitary operator acting on~$\cH$, start by taking~$\theta\in\R$ and define a unitary linear operator by
%
\begin{equation}\label{def-small-u-theta}	
u_{\theta}: \L^2(\R^3)\rightarrow \L^2(\R^3),\ (u_{\theta}\phi)(k)=e^{-\frac{3}{2}\theta}\phi(e^{-\theta}k).
\end{equation}
Next, $u_{\theta}$ is lifted to an operator $U_{\theta}:\cF\rightarrow\cF$ acting on the Fock space by second quantization
\begin{equation}\label{def-u-theta}	
U_{\theta}=\Gamma(u_{\theta})=\I_{\C}\oplus\bigoplus_{n=1}^{\infty}\bigl(u_{\theta}\otimes\dots\otimes u_{\theta}\bigr).
\end{equation}
Now identify $U_{\theta}=\I_{C^2}\otimes U_{\theta}$ to obtain an operator acting on the Hilbert space of the system. It acts on an element of $\cH$ by the following relations (see e.g. \cite[Section~I.5]{BachFroehlichSigal1998a})
\begin{align}		
U_{\theta}\Biggl(\begin{pmatrix}a\\b\end{pmatrix}\otimes\Omega\Biggr)&= \begin{pmatrix}a\\b\end{pmatrix}\otimes\Omega,\label{u-theta-relation-1}\\
U_{\theta}(M\otimes a^{\#}(\phi))U_{\theta}^*&=M\otimes a^{\#}(u_{\theta}\phi),\label{u-theta-relation-2}
\end{align}
where $\Omega\in\cF$ denotes the Fock vacuum, $a,b\in\C$, $M\in\C^{2\times2}$ and $a^{\#}\in\{a,a^*\}$. Note that, since $u_{\theta}$ is unitary, the transformation $U_{\theta}$ conserves the canonical commutation relations on $\cF$. It also satisfies $U_{\theta}\Omega=\Omega$ which implies that $U_{\theta}$ is unitary (see e.g.~\cite{ThirringIV1980}).
%
We now use $U_{\theta}$ to transform $H_g$. Considering the unperturbed operator first, using linearity yields
\begin{equation}\label{transformation-linearity}	
U_{\theta}H_0U_{\theta}^*=H_{at}+e^{-\theta}H_{ph}
\end{equation}
The key to the transformation of the last operator lies in the structure of the coupling term. It allows a direct computation using \eqref{u-theta-relation-2}, leading to
\begin{align}	
gW_{\theta}:&=gU_{\theta}WU_{\theta}^*=g\sigma_1\otimes a^*\Biggl(\frac{e^{-\frac{3}{2}\theta}f(e^{-\theta}k)}{\sqrt{|e^{-\theta}k|}}\Biggr)+g\sigma_1\otimes a\Biggl(\frac{e^{-\frac{3}{2}\theta}\overline{f(e^{-\theta}k)}}{\sqrt{|e^{-\theta}k|}}\Biggr)\nonumber\\
&=g\sigma_1\otimes\int_{\R^3}\Biggl(\frac{e^{-\frac{3}{2}\theta}f(e^{-\theta}k)}{\sqrt{|e^{-\theta}k|}}a^*(k)+\frac{e^{-\frac{3}{2}\theta}\overline{f(e^{-\theta}k)}}{\sqrt{|e^{-\theta}k|}}a(k)\Biggr)\ \dx^3k\nonumber\\
&=g\sigma_1\otimes\int_{\R^3}\Biggl(\frac{e^{-\theta}f(e^{-\theta}k)}{\sqrt{|k|}}a^*(k)+\frac{\overline{e^{-\bar{\theta}}f(e^{-\bar{\theta}}k)}}{\sqrt{|k|}}a(k)\Biggr)\ \dx^3k.\label{transformation-result-interaction}
\end{align}
Note that the transformation done in the last step is allowed since $\theta\in\R$ and  thus~$\theta=\bar{\theta}$. This seemingly small modification ensures that the involved functions are analytic and, therefore, is the key to continuing the operator to complex numbers with a sufficiently small imaginary part.

\begin{lemma}[Analytic Continuation]\label{analytic-continuation} 
Consider the operators
\begin{align}		
W_{\theta}&:=\sigma_1\otimes\int_{\R^3}\Biggl(\frac{e^{-\theta}f(e^{-\theta}k)}{\sqrt{|k|}}a^*(k)+\frac{\overline{e^{-\bar{\theta}}f(e^{-\bar{\theta}}k)}}{\sqrt{|k|}}a(k)\Biggr)\ \dx^3k,\label{def-w-theta}\\
H_g(\theta)&:=H_{at}\otimes \I_{\cF}+\I_{\C^2}\otimes\bigl(e^{-\theta}H_{ph}\bigr)+gW_{\theta},\label{def-h-theta}
\end{align}
and let~$\theta\in\C$ with $|\theta|<\frac{\pi}{6}$. Then $(W_{\theta})_{\theta\in D(0,\frac{\pi}{6})}$ is an analytic family of type A and $H_g(\theta)$ posesses an analytic continuation into the complex plane.
\end{lemma}

The proof requires a tool from complex analysis. It is noted in the following lemma~(see~\cite[Lemma~A.4]{BachBallesterosPizzo2017}).

\begin{lemma}\label{tool-analytic-continuation}	
Let $F:\mathcal{M}\rightarrow\C$ be an analytic function defined in a complex neighborhood~$\mathcal{M}$ of the closed disk~$\overline{D(0,r)}$ for some~$r>0$ and~$\gamma(0,r)$ denote the curve describing the circle with radius~$r$ around~$0$. Then, for~$s<r$,~$\theta\in D(0,s)$ and $h$ such that~${\theta+h\in D(0,s)}$, it is
\begin{equation}\label{formula-tool-analytic-continuation}	
\Big|\frac{F(\theta+h)-F(\theta)}{h}-F'(\theta)\Big|\leq\frac{|h|r}{(r-s)^3}\max_{z\in\gamma(0,r)}|F(z)|.
\end{equation}
\end{lemma}

\begin{proof}		
Since $F$ is analytic, one may use Cauchy's integral formula to write
\begin{align}	
\frac{F(\theta+h)-F(\theta)}{h}-F'(\theta)&=\frac{1}{2\pi i}\int_{\gamma(0,r)}\Biggl(\frac{F(z)}{h(z-\theta)}-\frac{F(z)}{h(z-(\theta+h))}-\frac{F(z)}{(z-\theta)^2}\Biggr)\ \dx z\nonumber\\
&=\frac{1}{2\pi i}\int_{\gamma(0,r)}\frac{hF(z)}{(z-\theta)^2(z-(\theta+h))}\ \dx z,\label{calc-tool-analytic-continuation-1}
\end{align}
which yields the desired result by estimating the absolute value as
\begin{align}	
\Big|\frac{1}{2\pi i}\int_{\gamma(0,r)}\frac{hF(z)}{(z-\theta)^2(z-(\theta+h))}\ \dx z\Big|&\leq\frac{1}{2\pi}\cdot2\pi r\cdot|h|\max_{z\in\gamma(0,r)}\Big|\frac{F(z)}{(z-\theta)^2(z-(\theta+h))}\Big|\nonumber\\
&=\frac{|h|r}{(r-s)^3}\max_{z\in\gamma(0,r)}|F(z)|.\label{calc-tool-analytic-continuation-2}
\end{align}
\end{proof}

With this tool in place, the analytic continuation can be established.

\begin{proof}[Proof of Lemma \ref{analytic-continuation}]		
Following the proof of Theorem 4.3 and 4.4 in \cite{BachBallesterosPizzo2017}, the first step is to show that $\theta\mapsto W_{\theta}(H_{ph}+\rho)^{-\frac{1}{2}}$ is analytic for $\rho>0$ and sufficiently small $\theta$. The fact that $\dom(W_{\theta})$, by definition, does not depend on the parameter $\theta$ then implies the claim that $(W_{\theta})_{\theta}$ is an analytic family of type A. Setting
\begin{equation}\label{def-g-theta}		
G_{\theta}(k):=\frac{e^{-\theta}f(e^{-\theta}k)}{\sqrt{|k|}},
\end{equation}
note that both $G_{\theta}$ and $\overline{G_{\bar{\theta}}}$ are analytic functions of $\theta$ for fixed $k$. Let $G'_{\theta}$ and $\overline{G_{\bar{\theta}}}'$, respectively, denote the derivative with respect to $\theta$ and consider the operator
\begin{equation}\label{def-formal-derivative}
W'_{\theta}:=\sigma_1\otimes\int_{\R^3}\Bigl(G'_{\theta}a^*(k)+\overline{G_{\bar{\theta}}}'a(k)\Bigr)\ \dx^3k,
\end{equation}
which results from inserting $G'_{\theta}$ instead of $G_{\theta}$ in~\eqref{def-w-theta}. Using $W'_{\theta}$ as formal derivative, Lemma~\ref{standard-estimate} allows reducing the term to a statement about the functions involved
\begin{equation}\label{calc-analytic-continuation-1}	
\Big\|\Bigl(\frac{W_{\theta+h}-W_{\theta}}{h}-W'_{\theta}\Bigr)\psi\Big\|\leq\Big\|\frac{G_{\theta+h}-G_{\theta}}{h}-G'_{\theta}\Big\|_{\rho}\cdot\big\|(H_{ph}+\rho)^{-\frac{1}{2}}\psi\big\|.
\end{equation}
Lemma~\ref{tool-analytic-continuation} now implies the existence of a positive constant~${\widetilde{C}\neq\widetilde{C}(k,\theta)}$ such that
\begin{equation}\label{calc-analytic-continuation-2}	
\Big\|\Bigl(\frac{W_{\theta+h}-W_{\theta}}{h}-W'_{\theta}\Bigr)\psi\Big\|\leq\widetilde{C}\cdot|h|\cdot\big\|(H_{ph}+\rho)^{-\frac{1}{2}}\psi\big\|.
\end{equation}
This establishes the analyticiy of~$W_{\theta}(H_{ph}+\rho)^{-\frac{1}{2}}$ as a Banach space-valued map. The second part of the Lemma can be shown analogously. Here we use the formal derivative~${H_g'(\theta)=-e^{-\theta}H_{ph}+gW'(\theta)}$, yielding
\begin{align}	
&\Big\|\Bigl(\frac{H_g(\theta+h)-H_g(\theta)}{h}-H_g'(\theta)\Bigr)\psi\Big\|\label{calc-analytic-continuation-3}\\
=\Big\|\Bigl(&\frac{H_{at}-H_{at}}{h}\Bigr)\psi\Big\|+\Big|\frac{e^{-(\theta+h)}-e^{\theta}}{h}+e^{-\theta}\Big|\cdot\big\|H_{ph}\psi\big\|+g\cdot\Big\|\Bigl(\frac{W_{\theta+h}-W_{\theta}}{h}-W'_{\theta}\Bigr)\psi\Big\|.\nonumber
\end{align}
Note that the first term vanishes and the third term has already been treated in~\eqref{calc-analytic-continuation-2}. For the second term, observe that 
\begin{equation}\label{calc-analytic-continuation-4}		
\Big|\frac{e^{-(\theta+h)}-e^{\theta}}{h}+e^{-\theta}\Big|\cdot\big\|H_{ph}\psi\big\|\leq|h|\cdot \eta(h)\cdot\big\|H_{ph}\psi\big\|,
\end{equation}
where $\eta(h)\rightarrow0$ as $h\rightarrow0$. Since $\rho>0$, it is $\|H_{ph}\psi\|\leq\|(H_{ph}+\rho)\psi\|$ and we estimate
\begin{equation}\label{calc-analytic-continuation-5}		
\big\|(H_{ph}+\rho)^{-\frac{1}{2}}\psi\big\|\leq\frac{1}{2}\Bigl(\big\|(H_{ph}+\rho)\psi\big\|+\|\psi\|\Bigr).
\end{equation}
These inequalities can be used to obtain a relative bound in terms of $(H_{ph}+\rho)$. Also including $H_{at}$ on the right-hand side now yields
\begin{equation}\label{calc-analytic-continuation-6}		
\Big\|\Bigl(\frac{H_g(\theta+h)-H_g(\theta)}{h}-H_g'(\theta)\Bigr)\psi\Big\|\leq|h|\cdot\Bigl(C_1\big\|(H_0(0)+\rho)\psi\big\|+C_2\|\psi\|\Bigr)
\end{equation}
which is the desired bound in terms of $H_0(0)$. The constants $C_1,C_2>0$ result from combining the bounds from the previous estimates. This implies that $H_g(\theta)(H_0(0)+\rho)^{-1}$ is analytic as a Banach space-valued map $D(0,s)\rightarrow\mathcal{B}(\cH)$ as well as the desired analytic continuation into the complex plane.
\end{proof}
%
From now on, set $\theta=i\vartheta$ with $\vartheta\in\R$ sufficiently small to satisfy the assumptions of the previous lemma and define the operators
\begin{align} 	
H_0(\theta)&:=H_{at}\otimes \I_{\cF}+\I_{\C^2}\otimes(e^{-\theta}H_{ph}),\label{def-h-null-theta}\\
H_g(\theta)&:=H_{at}\otimes \I_{\cF}+\I_{\C^2}\otimes(e^{-\theta}H_{ph})+gW_{\theta}\label{def-h-g-theta}.
\end{align}
The section is concluded with a discussion of the properties of $H_0(\theta)$ and $H_g(\theta)$.

\pagebreak
\begin{lemma}\label{properties}		
Let $H_0(\theta),H_g(\theta)$ as defined in \eqref{def-h-null-theta} and \eqref{def-h-g-theta} respectively.
\begin{itemize}
\item[(i)] $H_0(\theta)$ is normal, i.e., $\bigl[H_0(\theta),H_0(\theta)^*\bigr]=0$.
\item[(ii)] Let ${z\in\sigma\bigl(H_g(\theta)\bigr)}$. Then there exists some $R>0$ and a constant $C>0$ such that~${z\in \overline{D(2,R)}}$ implies that $z\in\big\{a+b|a\in\sigma\bigl(H_0(\theta)\bigr),b\in\overline{D(0,Cg)}\big\}$.
\item[(iii)] $H_g(\theta)$ is closed.
\end{itemize}
\end{lemma}

\begin{proof}		
(i)		
Note that $H_0(\theta)$ is a closed operator on $\C^2\otimes\dom(H_{ph})$. Since both $H_{at}$ and~$H_{ph}$ are self-adjoint, the adjoint~$H_0(\theta)^*$ is also a closed operator acting on $\C^2\otimes\dom(H_{ph})$ and it is given by
\begin{equation}\label{adjoint-h-0}		
H_0(\theta)^*=H_{at}+e^{\theta}H_{ph}
\end{equation}
This implies that
\begin{equation}		
\mathnormal{}\bigl[H_0(\theta),H_0(\theta)^*\bigr]=0.\label{commutator-normal-h-0}
\end{equation}
(ii)		
Recalling that $H_g(\theta)=H_0(\theta)+gW_{\theta}$, consider $z\notin\sigma\bigl(H_0(\theta)\bigr)$ and write
\begin{align}		
\bigl(H_g(\theta)-z\bigr)^{-1}=\bigl(H_0(\theta)-z\bigr)^{-1}\sum_{j=0}^{\infty}\bigl[-gW_{\theta}(H_0(\theta)-z)^{-1}\bigr]^j\label{calc-properties-1}
\end{align}
using a formal Neumann series expansion, which converges if ${g\|W_{\theta}\bigl(H_0(\theta)-z\bigr)^{-1}\|<1}$. Expanding the term yields
\begin{align}		
g\cdot\big\|W_{\theta}\bigl(H_0(\theta)-z\bigr)^{-1}\big\|
&=g\cdot\Big\|W_{\theta}(H_{ph}+1)^{-\frac{1}{2}}\Big\|\cdot\Big\|\frac{(H_{ph}+1)^{\frac{1}{2}}}{H_0(\theta)-z}\Big\|,\label{calc-properties-2}
\end{align}
where we use fractions as an alternative notation to write the commuting operators. By Lemma~\ref{standard-estimate}, it is
\begin{equation}\label{calc-properties-3}	
\big\|W_{\theta}(H_{ph}+1)^{-\frac{1}{2}}\big\|\leq2\Bigl(\Big\|\frac{G_{\theta}}{\sqrt{\omega}}\Big\|_{L^2}+\|G_{\theta}\|_{L^2}\Bigr).
\end{equation}
For the second term note that $H_0(\theta)$ is normal following the analysis in (i), and one may, therefore, use the functional calculus to write
\begin{align}		
\Big\|\frac{(H_{ph}+1)^{\frac{1}{2}}}{H_0(\theta)-z}\Big\|&\leq\sup_{r\geq0}\big|(r+1)^{\frac{1}{2}}(H_{at}+e^{-\theta}r-z)^{-1}\big|\nonumber\\
&\leq\max\Big\{\sup_{r\geq0}\Big|\frac{r+1}{2+e^{-\theta}r-z}\Big|,\ \sup_{r\geq0}\Big|\frac{r+1}{e^{-\theta}r-z}\Big|\Big\}\label{calc-properties-4}
\end{align}
Here, the maximum is attained for the first term. Expanding the numerator by adding and subtracting $(2-z)$ implies for the supremum that
\begin{equation}\label{calc-properties-5}	
\sup_{r\geq0}\Big|\frac{r+1}{2+e^{-\theta}r-z}\Big|\leq\frac{1+|2-z|}{\dist\bigl(z,\sigma(H_0(\theta))\bigr)}+1 
\end{equation}
The term $|2-z|$ is bounded since $z$ is taken from a circle around 2 with radius $R>0$. This implies the existence of a constant $C'>0$ such that
\begin{equation}\label{calc-properties-6}		
\Big\|\frac{(H_{ph}+1)^{\frac{1}{2}}}{H_0(\theta)-z}\Big\|\leq C'\ \frac{1}{\dist\bigl(z,\sigma(H_0(\theta))\bigr)}.
\end{equation}
Both estimates together yield the existence of a constant $C>0$ such that
\begin{equation}\label{calc-properties-7}	
g\cdot\big\|W_{\theta}(H_{ph}+1)^{-\frac{1}{2}}\big\|\cdot\Big\|\frac{(H_{ph}+1)^{\frac{1}{2}}}{H_0(\theta)-z}\Big\|\leq Cg\ \frac{1}{\dist\bigl(z,\sigma(H_0(\theta))\bigr)},
\end{equation}
which is smaller than one if $g$ is chosen small enough. This ensures the convergence of the Neumann series and yields the claim of (ii) when the explicit bounds are inserted.

(iii)		
In (ii) it was established that $\varepsilon=\|gW_{\theta}\bigl(H_0(\theta)-z\bigr)^{-1}\|<1$. This implies that~$gW_{\theta}$ is relatively bounded by $H_0(\theta)$ with relative bound $\varepsilon<1$. Having already established that $H_0(\theta)$ is closed in (i), the same must hold for~$H_g(\theta)=H_0(\theta)+gW_{\theta}$~(see \cite[p.~191]{Kato1980}).
\end{proof}

\subsection{Application of the Feshbach-Schur Map}\label{sect-application-feshbach}
In this section, a suitable Feshbach-Schur map is applied to the transformed Hamiltonian~\eqref{def-h-g-theta} to obtain an effective Hamiltonian acting only on the Fock space. Recall that the complex dilated and analytically continued Hamiltonian is given by
\begin{equation}\label{def-h-g-theta-with-matrices}		
H_g(\theta)=\begin{pmatrix}2&0\\0&0\end{pmatrix}\otimes \I_{\cF}+\begin{pmatrix}1&0\\0&1\end{pmatrix}\otimes(e^{-\theta}H_{ph})+g\begin{pmatrix}0&1\\1&0\end{pmatrix}\otimes\bigl(a^*(G_{\theta})+a(G_{\overline{\theta}})\bigr)
\end{equation}
on $\C^2\otimes\cF$ and denote $\Phi_{\theta}=a^*(G_{\theta})+a(G_{\overline{\theta}})$. Next, consider the following projections acting on $\C^2$ by
\begin{equation}\label{def-p-up-p-down}		
P_{\uparrow}:=\begin{pmatrix}1&0\\0&0\end{pmatrix},\quad P_{\downarrow}:=\begin{pmatrix}0&0\\0&1\end{pmatrix},
\end{equation}
and define a projection acting on $\cH$ by
\begin{equation}\label{def-projection-p}		
\P :=P_{\uparrow}\otimes\I_{\cF}.
\end{equation}
By definition, $\P$ maps to the excited state of the atom. We show that the operator
\begin{equation}\label{feshbach-applied-to-h}		
F_{\P}(H_g(\theta)-z)=\P(H_g(\theta)-z)\P-\P(H_g(\theta)-z)\overline{\P}(\overline{\P}(H_g(\theta)-z)\overline{\P})^{-1}\overline{\P}(H_g(\theta)-z)\P,
\end{equation}
is well-defined for $z$ in a small complex vicinity of $2$. First, the terms appearing in \eqref{feshbach-applied-to-h} are explicitly calculated. 
%
As $\P(gW_{\theta})\P=0$, we have that
\begin{equation}\label{calc-application-feshbach-2}		
\P\bigl(H_g-z\bigr)\P=P_{\uparrow}\otimes(2+e^{-\theta}H_{ph}-z).
\end{equation}
Next, consider the second term in \eqref{feshbach-applied-to-h}. The expressions including both $\P$ and $\overline{\P}$ yield
\begin{equation}\label{calc-application-feshbach-4}		
\P\bigl(H_g-z\bigr)\overline{\P}=gP_{\uparrow}\sigma_1\otimes\Phi_{\theta},\quad \overline{\P}\bigl(H_g-z\bigr)\P=g\sigma_1P_{\uparrow}\otimes\Phi_{\theta}.
\end{equation}
The last expression to be considered is given by
\begin{equation}\label{calc-application-feshbach-7}		
\overline{\P}\bigl(H_g-z\bigr)\overline{\P}=P_{\downarrow}\otimes(e^{-\theta}H_{ph}-z).
\end{equation}
%
Note that the inverse of this operator is involved in \eqref{feshbach-applied-to-h}. Let $z\in D(2,Cg)$ for some constant~$C>0$. Since $H_{ph}$ is self-adjoint, one can apply the spectral theorem to obtain
\begin{equation}\label{inverse-term-exists}		
\big\|(e^{-\theta}H_{ph}-z)^{-1}\big\|=\sup_{r\geq0}\frac{1}{|e^{-\theta}r-z|}\leq\frac{1}{\delta}.
\end{equation}
Here, the constant $\delta$ is the result of a geometric argument. The setting is sketched in the following image.

\begin{figure}[H]	
\centering
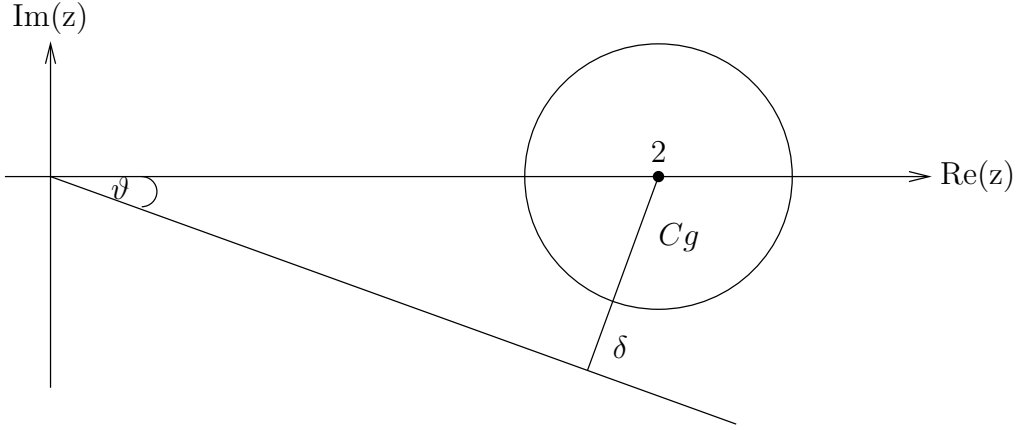
\caption{The basis for the geometric estimate}\label{fig-geometric-estimate-line-circle}
\end{figure}

Since the radius of the disk is proportional to the coupling constant, one can always achieve $\delta>0$ by restricting $g$. The distance can be calculated explicitly, yielding
\begin{equation}\label{value-for-delta}		
\sin(\vartheta)=\frac{Cg+\delta}{2}\ \Leftrightarrow\ \delta=2\sin(\vartheta)-Cg.
\end{equation}
For the following estimates, we include the leading order in $\vartheta$ into the calculation and rewrite the bound as $\vartheta d$ for some suitable constant $d>0$. It now follows that
\begin{equation}\label{calc-application-feshbach-8}		
\big\|\bigl(\overline{\P}\bigl(H_g(\theta)-z\big)\overline{\P}\bigr)^{-1}\overline{\P}\big\|=\big\|P_{\downarrow}\otimes\bigl(e^{-\theta}H_{ph}-z\bigr)^{-1}\big\|<\infty.
\end{equation}
The next term to consider is
\begin{equation}\label{calc-application-feshbach-9}		
\overline{\P}\bigl(\overline{\P}\bigl(H_g(\theta)-z\bigr)\overline{\P})^{-1}\overline{\P}\bigl(H_g(\theta)-z\bigr)\P=P_{\downarrow}\otimes(e^{-\theta}H_{ph}-z)^{-1}(g\sigma_1P_{\uparrow}\otimes\Phi_{\theta}).
\end{equation}
The norm is estimated by
\begin{align}		
&g\cdot \big\|\sigma_1P_{\uparrow}\otimes(e^{-\theta}H_{ph}-z)^{-1}\Phi_{\theta})\big\|\nonumber\\
&\leq g\cdot \big\|(e^{-\theta}H_{ph}-z)^{-1}\bigl(H_{ph}+1\bigr)\big\|\cdot\big\|\bigl(H_{ph}+1\bigr)^{-\frac{1}{2}}\big\|\cdot\big\|\bigl(H_{ph}+1\bigr)^{-\frac{1}{2}}\Phi_{\theta}\big\|.\label{calc-application-feshbach-10}
\end{align}
Note that the operator $(e^{-\theta}H_{ph}-z)$ is not self-adjoint and, therefore, it does not possess a square root. For this reason, the identity is inserted  such that Lemma~\ref{standard-estimate} can be applied. Rewriting the first term of~\eqref{calc-application-feshbach-10} similar to~\cite[Lemma~6.17]{BachBallesterosPizzo2017}, it can be estimated by
\begin{align}		
\Big\|\frac{H_{ph}+1}{e^{-\theta}H_{ph}-z}\Big\|&=\Big\|e^{\theta}+\frac{e^{\theta}z+1}{e^{-\theta}H_{ph}-z}\Big\|\leq1+\sup_{r\geq0}\frac{|z|+1}{|e^{-\theta}r-z|}\nonumber\\
&\leq1+\frac{2+Cg+1}{\vartheta d},\label{calc-application-feshbach-11}
\end{align}
using the spectral theorem for $H_{ph}$ and~\eqref{inverse-term-exists}. The remaining terms in~\eqref{calc-application-feshbach-10} can be estimated directly using the spectral theorem for $H_{ph}$ and Lemma \ref{standard-estimate}, respectively. Therefore,
\begin{equation}\label{calc-application-feshbach-13}		
\|\overline{\P}\bigl(\overline{\P}\bigl(H_g(\theta)-z\bigr)\overline{\P}\bigr)^{-1}\overline{\P}\bigl(H_g(\theta)-z\bigr)\P\|<\infty.
\end{equation}
Applying the same line of argument to the operator
\begin{equation}\label{calc-application-feshbach-14}	
\P\bigl(H_g(\theta)-z\bigr)\overline{\P}\bigl(\overline{\P}\bigl(H_g(\theta)-z\bigr)\overline{\P}\bigr)^{-1}\overline{\P}
=g(P_{\uparrow}\sigma_1\otimes\Phi_{\theta})\bigl(P_{\downarrow}\otimes(e^{-\theta}H_{ph}-z)^{-1}\bigr)
\end{equation}
also yields a finite bound for its norm. The last expression to consider is the second term on the right side of \eqref{feshbach-applied-to-h}. It is given by
\begin{equation}\label{calc-application-feshbach-15}		
\P\bigl(H_g(\theta)-z\bigr)\overline{\P}\bigl(\overline{\P}\bigl(H_g(\theta)-z\bigr)\overline{\P}\bigr)^{-1}\overline{\P}\bigl(H_g(\theta)-z\bigr)\P=g^2P_{\uparrow}\otimes\bigl(\Phi_{\theta}(e^{-\theta}H_{ph}-z)^{-1}\Phi_{\theta}\bigr).
\end{equation}
Here, we rewrite the norm as
\begin{align}	
&g^2\cdot\|P_{\uparrow}\otimes\bigl(\Phi_{\theta}\bigl(e^{-\theta}H_{ph}-z\bigr)^{-1}\Phi_{\theta}\bigr)\|\nonumber\\
&\leq g^2\cdot\|\Phi_{\theta}\bigl(H_{ph}+1\bigr)^{-\frac{1}{2}}\big\|\cdot\big\|\bigl(H_{ph}+1\bigr)\bigl(e^{-\theta}H_{ph}-z\bigr)^{-1}\big\|\cdot\big\|(H_{ph}+1)^{-\frac{1}{2}}\Phi_{\theta}\big\|,\label{calc-application-feshbach-16}
\end{align}
which is bounded due to the estimates already derived. Notice in particular that both terms $(H_{ph}+1)^{-\frac{1}{2}}$ are needed here, whereas only one was explicitly required for the previous bounds. This implies that the operator defined in \eqref{feshbach-applied-to-h} can be written as
%
\begin{equation}\label{feshbach-map-result}		
F_{\P}\bigl(H_g(\theta)-z\bigr)=P_{\uparrow}\otimes \bigl(2-z+e^{-\theta}H_{ph}-g^2\Phi_{\theta}(e^{-\theta}H_{ph}-z)^{-1}\Phi_{\theta}\bigr)
\end{equation}
and that it is indeed well-defined for $z\in D(2,Cg)$. The effective Hamiltonian acting on~$\cF$ is then given by
\begin{equation}\label{def-h-eff}		
H_{\mathrm{eff}}=2-z+e^{-\theta}H_{ph}-g^2\Phi_{\theta}(e^{-\theta}H_{ph}-z)^{-1}\Phi_{\theta}.
\end{equation}
Lastly, we rewrite the operator. The form aimed at is
\begin{equation}\label{desired-form}		
H_{\mathrm{eff}}=F_0-g^2W=F_0-g^2\sum_{m+n=2}W_{m,n},
\end{equation}
where the operators $W_{m,n}$ are given by
\begin{equation}\label{def-w-m-n}		
W_{m,n}=\int a^*(k_1)\dots a^*(k_m)\ w_{m,n}\bigl[z,H_{ph},k_1,\dots,k_m,\tilde{k}_1,\dots, \tilde{k}_n\bigr]\ a(\tilde{k}_1)\dots a(\tilde{k}_n)\ \dx^{m}k\ \dx^{n}\tilde{k}.
\end{equation}
First, writing out $\Phi_{\theta}(e^{-\theta}H_{ph}-z)^{-1}\Phi_{\theta}$ leads to
\begin{equation}\label{calc-ordering-1}		
\int\bigl(G_{\theta}(k_1)a^*(k_1)+\overline{G_{\overline{\theta}}(k_1)}a(k_1)\bigr)(e^{-\theta}H_{ph}-z)^{-1}\bigl(G_{\theta}(k_2)a^*(k_2)+\overline{G_{\overline{\theta}}(k_2)}a(k_2)\bigr)\ \dx k_1\dx k_2 .
\end{equation}
Expanding the integrand now yields three terms corresponding to the $W_{m,n}$ respectively. Next, we use the canonical commutation relations~\eqref{ccr} and the pull-through formulae for the field Hamiltonian (see e.g.~\cite[Lemma~52]{HaslerHerbst2011}) to move the operators to the right positions. Since the other terms may be treated similarly, the calculation is only carried out for the term corresponding to $W_{1,1}$. The first part obtained from expanding the integrand in~\eqref{calc-ordering-1} is given by
\begin{equation}\label{calc-ordering-2}		
\int a^*(k_1)\ G_{\theta}(k_1)\ \overline{G_{\overline{\theta}}(k_2)}\ (e^{-\theta}H_{ph}-z)^{-1}\ a(k_2)\ \dx k_1\dx k_2,
\end{equation}
which is already of the desired form. For the second part, moving the operator $a^*(k_2)$ from the right to the left and then swapping $a(k_1)$ and $(e^{-\theta}H_{ph}-z)^{-1}$ leads to
\begin{align}		
&\int\overline{G_{\overline{\theta}}(k_1)}a(k_1)(e^{-\theta}H_{ph}-z)^{-1}G_{\theta}(k_2)a^*(k_2)\ \dx k_1\dx k_2\nonumber\\
&=\int a^*(k_2) G_{\theta}(k_2)\overline{G_{\overline{\theta}}(k_1)}\bigl(e^{-\theta}H_{ph}+e^{-\theta}\omega(k_1)
+e^{-\theta}\omega(k_2)-z\bigr)^{-1}a(k_1)\ \dx k_1\dx k_2\nonumber\\
&\ +\int G_{\theta}(k)\overline{G_{\overline{\theta}}(k)}\bigl(e^{-\theta}H_{ph}+e^{-\theta}\omega(k)-z\bigr)^{-1}\ \dx k,\label{calc-ordering-3}
\end{align}
where the additional term results from applying the relation ${[a^*(k_2),a(k_1)]=\delta(k_2-k_1)\I_{\cF}}$ and explicitly solving one of the integrals. Since it does not include any creation nor annihilation operators, it is denoted as $W_{0,0}$ and is added to $F_0$. This yields the representation
\begin{equation}\label{def-h-eff-2}		
H_{\mathrm{eff}}=:H_{g_0,g}(z)=H_{g_0,0}(z)-g^2(W_{2,0}+W_{1,1}+W_{0,2}),
\end{equation}
where the term $H_{g_0,0}(z)$ corresponds to $F_0$ given by
\begin{equation}\label{def-f-0}		
H_{g_0,0}:=F_0=2-z+e^{-\theta}H_{ph}-g_0^2W_{0,0}
\end{equation}
and the operators $W_{2,0}$, $W_{1,1}$ and $W_{0,2}$ are given in the desired ordering by
\begin{align}		
W_{2,0}&=\int a^*(k_1)a^*(k_2)\Biggl[\frac{G_{\theta}(k_1)G_{\theta}(k_2)}{e^{-\theta}H_{ph}+e^{-\theta}\omega(k_2)}\Biggr]\ \dx k_1\dx k_2,\label{def-w-2-0}\\
W_{1,1}&=\int a^*(k) \Biggl[G_{\theta}(k)\overline{G_{\overline{\theta}}(\tilde{k})}\Bigl(\frac{1}{e^{-\theta}H_{ph}+e^{-\theta}\omega(k)+e^{-\theta}\omega(\tilde{k})-z}\nonumber\\
&\quad\quad\quad\quad\quad+\frac{1}{e^{-\theta}H_{ph}-z}\Bigr)\Biggr]a(k)\ \dx k\dx \tilde{k},\label{def-w-1-1}\\
W_{0,2}&=\int\Biggl[\frac{\overline{G_{\overline{\theta}}(\tilde{k}_1)}\ \overline{G_{\overline{\theta}}(\tilde{k}_2)}}{e^{-\theta}H_{ph}+e^{-\theta}\omega(\tilde{k}_1)}\Biggr]a(\tilde{k}_1)a(\tilde{k}_2)\ \dx\tilde{k}_1\dx\tilde{k}_2,\label{def-w-0-2}
\end{align}
with the corresponding $w_{m,n}$ highlighted by the large square brackets. Note that an additional parameter $g_0$ was introduced with $F_0$ in~\eqref{def-f-0}. This allows to separate the interaction terms including creation and annihilation operators from the term $W_{0,0}$ when needed. However, if not stated otherwise, $g_0=g$ is assumed. 
%
\begin{remark}
Note that the key symmetry identified in~\cite{BachBallesterosKoenenbergMenrath2017} also applies here. Consider the operator $S=\sigma_3\otimes(-1)^{\cN_{ph}}$ with $\sigma_3$ denoting the third Pauli matrix and $\cN_{ph}$ denoting the photon number operator on $\cF$. Similar to~\cite{BachBallesterosKoenenbergMenrath2017}, one finds that $[S,H_g(\theta)]=0$ using that the matrix occuring in the coupling function is off-diagonal. Denoting~${S_{ph}=(-1)^{\cN_{ph}}}$, one also observes~$[S_{ph},H_{g_0,g}]=0$. Note that, as the off-digonal form of the matrix was used multiple times identifying the individual terms of~$H_{g_0,g}$, the symmetry enters here implicitly in the form of the specific Feshbach-Schur map considered.
\end{remark}

\clearpage

\section{Existence of Resonances}\label{sect-resonances}
In this section, the main result of this paper is shown. We recall it once more:

\begingroup
\def\thetheorem{\ref*{the-main-result}}
\begin{theorem}[Existence of Resonances] \label{main-result} 
Let $\vartheta$ be small enough to satisfy the assumptions of Lemma~\ref{analytic-continuation}, i.e.,~\eqref{final-restriction-theta} holds, and choose the parameters $\rho_0,\gamma,g>0$ sufficiently small such that~\eqref{final-restriction-rho-0},~\eqref{final-restriction-gamma} and~\eqref{final-restriction-g} are satisfied. Then there exists an eigenvalue~${E_{\mathrm{res}}\in D(2,Cg)}$ of $H_g(\theta)$.
\end{theorem}
\addtocounter{theorem}{-1}
\endgroup

Since the proof is based on Pizzo's method, constructing a sequence of infrared-regularized Hamiltonians as well as their respective eigenvalues and the projections to the corresponding eigenspaces constitutes the core of the proof. They are established using an inductive argument similar to~\cite{BachBallesterosPizzo2017}. The respective parts are carried out in Sections~\ref{sect-ind-basis},~\ref{sect-ind-hypothesis} and~\ref{sect-ind-step}. The proof of Theorem~\ref{main-result} is then carried out in Section~\ref{sect-main}.

\subsection{Induction Basis: Eigenvalue of $H_g^{(0)}$}\label{sect-ind-basis}
In this section, we consider an initial fixed infrared cutoff $\rho_0$ and show that the corresponding operator~$H_g^{(0)}$ posesses a single, non-degenerate eigenvalue $E^{(0)}$ in the complex vicinity of $2$ using an argument similar to~\cite[Section~5]{BachBallesterosPizzo2017}. Since the analyticity required by Lemma~\ref{analytic-continuation} does not allow for a sharp ultraviolet cutoff in the coupling function, the operator~$H_g^{(0)}$ already includes a part of the interaction. The resolvent of the operator, therefore, is constructed by Neumann series expansion, which is the content of Lemma~\ref{invertibility-0}. Note that Lemma~\ref{properties} stays valid when the infrared cutoff is included. We then consider the projection~$P^{(0)}$ and show that it is rank-one by comparing it to the projection $P^{(-1)}$ mapping to the eigenspace of the unperturbed operator. This implies the existence of a single, simple eigenvalue, thus establishing the induction basis. The section is concluded with some general estimates for the norm of the corresponding eigenvector as well as the projection itself as a basis for further calculations. The following subsets of the complex plane are used in the analysis. Recalling that $\theta=i\vartheta$ here, define
\begin{align}		
\mathcal{A}&:=\Bigl(\bigl[2-\tfrac{\rho_0}{4},2+\tfrac{\rho_0}{4}\bigr]+i\bigl[-\tfrac{\sin(\vartheta)}{4}\rho_0,\tfrac{\sin(\vartheta)}{4}\rho_0\bigr]\Bigr)-e^{-i\vartheta}[0,\infty),\label{def-original-set}\\
\mathcal{A}_0&:=\mathcal{A}\setminus D\bigl(2,\tfrac{\sin(\vartheta)}{16}\rho_0\bigr).\label{def-a-0}
\end{align}
The set $\mathcal{A}$ describes a rectangle centered about $E^{(-1)}=2$ in the complex plane and the tail that it sweeps out when moving it continuously in the direction $-e^{-i\vartheta}$. For the set~$\mathcal{A}_0$, the eigenvalue $E^{(-1)}$ and a small neighborhood of it are excluded. The setup is sketched in the following image.
\clearpage
\begin{figure}[htb]	
\centering
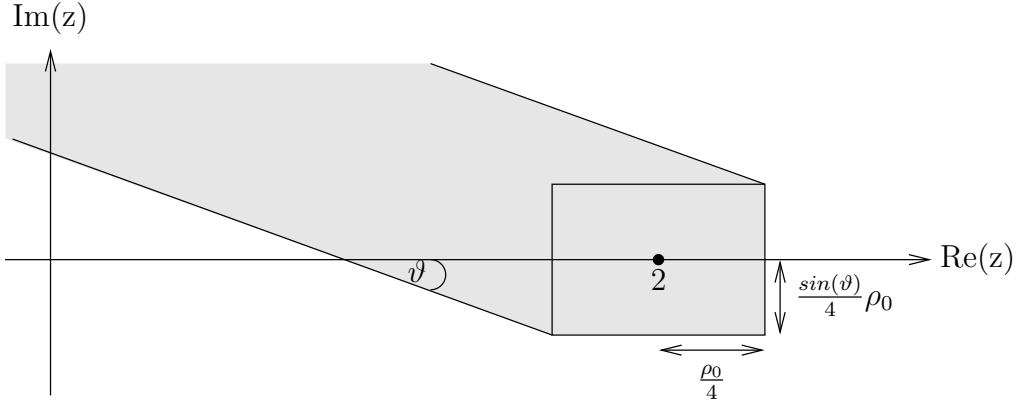
\caption{The set $\mathcal{A}$ in the complex plane}\label{fig-set-a-0}
\end{figure}

The next step is the construction of the resolvent. Note that the bound for the coupling constant assumed in the lemma explicitly depends on $\rho_0$.
\begin{lemma}\label{invertibility-0}		
Let $g$  be small enough such that $Cg<\frac{\sin(\vartheta)}{16}$ with the constant $C$ introduced in Lemma \ref{properties} and assume that
\begin{equation}\label{g-restriction-1}	
g<\frac{1}{6}\Biggl[\frac{8}{\sin(\vartheta)\sqrt{2\rho_0}}\Bigl(\Big\|\frac{G_{\theta}}{\sqrt{\omega}}\Big\|_{L^2}+\frac{1}{\sqrt{\rho_0}}\|G_{\theta}\|_{L^2}\Bigr)\Biggr]^{-1}.
\end{equation}
Then the resolvent of $H_g^{(0)}$ exists for all $z\in\mathcal{A}_0$ and obeys the bound
\begin{equation}\label{norm-resolvent-0}	
\big\|(H_g^{(0)}-z)^{-1}\big\|\leq\frac{6}{5}\ \frac{1}{|2-z|}.
\end{equation}
\end{lemma}
 
\begin{proof}		
First, the resolvent of $H_g^{(0)}$ is expanded, which yields
\begin{align}		
\bigl(H_g^{(0)}-z\bigr)^{-1}&=\sum_{j=0}^{\infty}\bigl(H_0^{(0)}-z\bigr)^{-1}\bigl[-gW_{\theta}^{(0)}\bigl(H_0^{(0)}-z\bigr)^{-1}\bigr]^j\label{calc-invertibility-0-1}
\end{align}
using a formal Neumann series expansion. Expanding the term by inserting the identity between the two operator involved yields
\begin{align}		
g\cdot\big\|W_{\theta}^{(0)}\bigl(H_0^{(0)}-z\bigr)^{-1}\big\|&=g\cdot\big\|W_{\theta}^{(0)}\bigl(H_{ph}^{(0)}+\rho_0\bigr)^{-\frac{1}{2}}\big\|\cdot\Big\|\frac{\bigl(H_{ph}^{(0)}+\rho_0\bigr)^{\frac{1}{2}}}{H_0^{(0)}-z}\Big\|,\label{calc-invertibility-0-2}
\end{align}
where the commuting operators are once more written as fractions. The first term may be estimated using Lemma \ref{standard-estimate}, giving the bound
\begin{equation}\label{calc-invertibility-0-3}		
\big\|W_{\theta}^{(0)}(H_{ph}^{(0)}+\rho_0)^{-\frac{1}{2}}\big\|\leq2\Bigl(\Big\|\frac{G_{\theta}}{\sqrt{\omega}}\Big\|_{L^2}+\rho_0^{-\frac{1}{2}}\|G_{\theta}\|_{L^2}\Bigr).
\end{equation}
For the second term note that $H_0^{(0)}$ is normal and one may, therefore, write
\begin{align}	
\Big\|\frac{(H_{ph}^{(0)}+\rho_0)^{\frac{1}{2}}}{H_0^{(0)}-z}\Big\|
&\leq\sup_{r\geq\rho_0}\big|(r+\rho_0)^{\frac{1}{2}}\bigl(H_{at}+e^{-i\vartheta}r-z\bigr)^{-1}\big|\nonumber\\
&=\max\Big\{\sup_{r\geq\rho_0}\frac{(r+\rho_0)^{\frac{1}{2}}}{|2+e^{-i\vartheta}r-z|},\ \sup_{r\geq\rho_0}\frac{(r+\rho_0)^{\frac{1}{2}}}{|e^{-i\vartheta}r-z|}\Big\}.\label{calc-invertibility-0-4}
\end{align}

To estimate this term, explicit bounds for $|2+e^{-i\vartheta}r-z|$ and $|e^{-i\vartheta}r-z|$ in terms of $r$ are needed. Keeping the picture of $\mathcal{A}_0$ in the complex plane in mind, note that the terms in question are given by the distance between some $z$ in the set and some $r$ located in the ray $e^{-i\vartheta}[\rho_0,\infty)$ or $2+e^{-i\vartheta}[\rho_0,\infty)$ respectively. We denote the right and the left bottom corner of the rectangle by $z_1$ and~$z_2$, i.e., we consider the complex  numbers given by
\begin{equation}\label{def-z-1-z-2}		
z_1:=2+\tfrac{\rho_0}{4}-i\tfrac{\sin(\vartheta)}{4},\quad z_2:=2-\tfrac{\rho_0}{4}-i\tfrac{\sin(\vartheta)}{4}.
\end{equation}
Before the estimates are considered, the setup is sketched in the following image.

\begin{figure}[htb]	
\centering
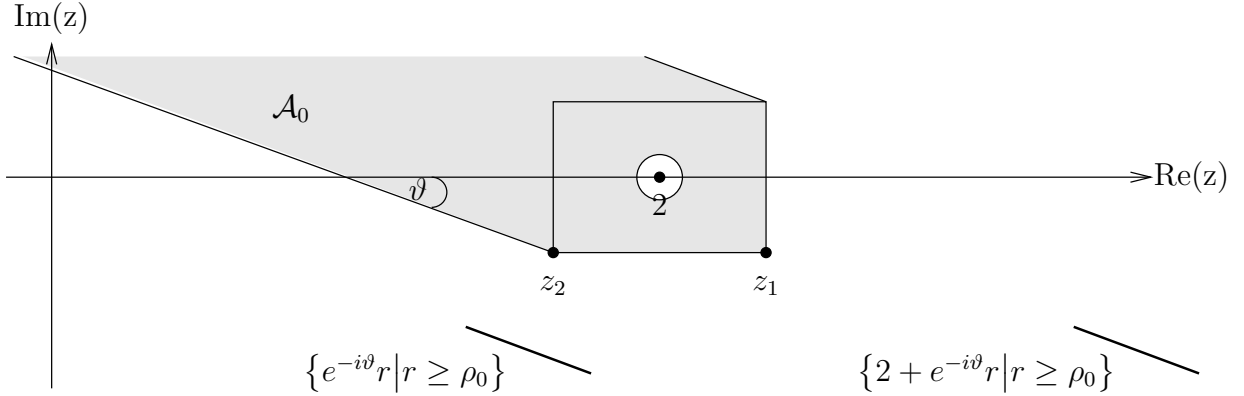
\caption{The basis for the geometric argument}\label{fig-geometry-estimate-1}
\end{figure}

For $z\in\mathcal{A}_0$ and $r\geq\rho_0$, direct calculation yields
\begin{align}		
|2+e^{-i\vartheta}r-z|&\geq|2+e^{-i\vartheta}r-z_1|\geq\frac{\sin(\vartheta)}{4}(r+\rho_0).\label{calc-invertibility-0-5}
\end{align}
The estimate for the second term appearing in the maximum in~\eqref{calc-invertibility-0-4} is derived similarly, implying
\begin{equation}\label{calc-invertibility-0-6}		
|e^{-i\vartheta}r-z|\geq|-ir\sin(\vartheta)-\mathrm{Im}(z_2)|\geq\frac{\sin(\vartheta)}{4}(r+\rho_0).
\end{equation}
This leads to the bound
\begin{equation}\label{calc-invertibility-0-7}
\max\Big\{\sup_{r\geq\rho_0}\frac{(r+\rho_0)^{\frac{1}{2}}}{|2+e^{-i\vartheta}r-z|},\ \sup_{r\geq\rho_0}\frac{(r+\rho_0)^{\frac{1}{2}}}{|e^{-i\vartheta}r-z|}\Big\}\leq\frac{4}{\sin(\vartheta)\sqrt{2\rho_0}}.
\end{equation}
One may, therefore, assemble a bound for~\eqref{calc-invertibility-0-2} from~\eqref{calc-invertibility-0-3} and~\eqref{calc-invertibility-0-7}, which yields a finite multiple of the coupling constant. Since $g$ is assumed to satisfy \eqref{g-restriction-1}, this bound is smaller than $\frac{1}{6}$, thus implying the convergence of the Neumann series. The explicit formula given in the lemma now follows by calculating the value of the geometric series and applying the spectral theorem to the normal operator $H_0^{(0)}$.
\end{proof}

With the resolvent of $H_g^{(0)}$ well-defined, we define the first projection
\begin{equation}\label{def-p0}		
P^{(0)}:=\frac{-1}{2\pi i}\int_{\gamma^{(0)}_{-1}}\frac{\dx z}{H_g^{(0)}-z},
\end{equation}
where $\gamma^{(0)}_{-1}$ is the circle with radius $r^{(0)}:=\frac{\sin(\vartheta)}{8}\rho_0$ and center $E^{(-1)}=2$.  Considering~$\Psi^{(0)}:=P^{(0)}\Psi^{(-1)}$, where~$\Psi^{(-1)}$ is the eigenvector corresponding to the eigenvalue~$E^{(-1)}$ of $H_0^{(0)}$, now yields a vector in the eigenspace of $H_g^{(0)}$. Note that
\begin{equation}\label{def-difference-p-0-p-1}		
\big\|\Psi^{(0)}-\Psi^{(-1)}\big\|=\big\|\bigl(P^{(0)}-P^{(-1)}\bigr)\Psi^{(-1)}\big\|\leq\big\|\bigl(P^{(0)}-P^{(-1)}\bigr)\big\|\cdot\big\|\Psi^{(-1)}\big\|.
\end{equation}
By the the second resolvent identity, the bounds from the previous lemma and the spectral theorem, the difference of the projections can be estimated as
\begin{align}		
\big\|\bigl(P^{(0)}-P^{(-1)}\bigr)\big\|&=\Big\|\frac{-1}{2\pi i}\int_{\gamma^{(0)}_{-1}}\Bigl(\bigl(H_g^{(0)}-z\bigr)^{-1}-\bigl(H_0^{(0)}-z\bigr)^{-1}\Bigr)\dx z\Big\|\nonumber\\
&=\frac{1}{2\pi}\ \Big\|\int_{\gamma^{(0)}_{-1}}\Bigl(\bigl(H_g^{(0)}-z\bigr)^{-1}\bigl(-gW_{\theta}^{(0)}\bigr)\bigl(H_0^{(0)}-z\bigr)^{-1}\Bigr)\dx z\Big\|\nonumber\\
&\leq\frac{1}{2\pi}\ 2\pi\ \frac{\sin(\vartheta)}{8}\rho_0\sup_{z\in\gamma^{(0)}_{-1}}\Bigl\{\big\|(H_g^{(0)}-z)^{-1}\big\|\cdot\big\|-gW_{\theta}^{(0)}\bigl(H_0^{(0)}-z\bigr)^{-1}\big\|\Bigr\}\nonumber\\
&\leq\frac{1}{6}\Bigl( \frac{1}{1-\frac{1}{6}}\Bigr) \frac{\sin(\vartheta)}{8}\rho_0\sup_{z\in\gamma^{(0)}_{-1}}\Bigl\{\big\|(H_0^{(0)}-z)^{-1}\big\|\Bigr\}=\frac{1}{5}.\label{calc-difference-p-0-p-1}
\end{align}
Note that the power of $\rho_0$ exactly cancels out, which is a consequence of considering the critical coupling function. If the model was slightly more regular, i.e., the coupling function behaved like~${|k|^{-1/2+\mu}}$ as $k\rightarrow0$ for some $\mu>0$, one would gain a positive power of $\rho_0$ here. The above calculation implies, in particular, that the distance between the projections is strictly less than one, thus allowing to identify $P^{(0)}$ as a rank-one projection. Applying~\eqref{calc-difference-p-0-p-1} to the eigenvectors yields
\begin{equation}\label{bounds-p-0}		
1-\frac{1}{5}\leq\big\|\Psi^{(0)}\big\|\leq1+\frac{1}{5},
\end{equation}
i.e., $\Psi^{(0)}$ is not zero. Together, these results establish that there is a unique eigenvalue~$E_g^{(0)}$ of~$H_g^{(0)}$ corresponding to the eigenvector $\Psi^{(0)}$ which is located in the interior of~$\gamma^{(0)}_{-1}$. In fact, one can localize it to $D(2,\frac{\sin(\vartheta)}{16}\rho_0)$ since the rest of the interior of the curve is included in the resolvent set of $H_g^{(0)}$. Starting from the eigenvalue equation
\begin{equation}\label{eigenvalue-equation-0}
{H_g^{(0)}\Psi^{(0)}=E_g^{(0)}\Psi^{(0)}},
\end{equation}
the value of $E_g^{(0)}$ can explicitly be written as
\begin{equation}\label{explicit-eigenvalue-0}		
E_g^{(0)}=\frac{\big\langle\Psi^{(-1)}\big|H_g^{(0)}\Psi^{(0)}\big\rangle}{\big\langle\Psi^{(-1)}\big|\Psi^{(0)}\big\rangle}.
\end{equation}
This completes the induction basis. The section is concluded with another estimate which is key to the following inductive argument as it allows to drop the additional restriction of having a minimal distance between~$z$ and the energy considered.
\begin{lemma}\label{resolvent-norm-0}	
Let $z\in\mathcal{A}$ and assume that the parameters satisfy the assumptions of Lemma~\ref{invertibility-0}. Then the following norm bound holds
\begin{equation}\label{formula-resolvent-norm-0}
\big\|\bigl(H_g^{(0)}-z\bigr)^{-1}\overline{P^{(0)}}\big\|\leq 43\ \frac{1}{\frac{\sin(\vartheta)}{2}\rho_0+|z-E_g^{(0)}|}.
\end{equation}
\end{lemma}

\begin{proof}
In the case $|z-2|>\frac{\sin(\vartheta)}{2}\rho_0$, the norm in question can directly be calculated using Lemma~\ref{invertibility-0} and the bound obtained for the norm of the projection~\eqref{bounds-p-0}. This yields
\begin{align}
\big\|\bigl(H_g^{(0)}-z\bigr)^{-1}\overline{P^{(0)}}\big\|&\leq \frac{6}{5}\ \frac{1}{\frac{\sin(\vartheta)}{16}\rho_0}\cdot \frac{11}{5}\leq \frac{66}{25}\ \frac{2}{\frac{\sin(\vartheta)}{16}\rho_0+\frac{1}{2}(|z-2|+|2-E_g^{(0)}|)}\nonumber\\
&\leq 43\ \frac{1}{\frac{\sin(\vartheta)}{2}\rho_0+|z-E_g^{(0)}|},\label{calc-resolvent-norm-0-1}
\end{align}
which is exactly~\eqref{formula-resolvent-norm-0}. The values for $z$ left to consider lie in the closed circle of radius~$\frac{\sin(\vartheta)}{16}\rho_0$ around~$2$. As~$\overline{P^{(0)}}$ maps to the complement of the eigenspace corresponding to $E_g^{(0)}$ and the eigenvalue is the only spectral point in $\mathcal{A}$, the complex-valued function
\begin{equation}\label{calc-resolvent-norm-0-2}
D\bigl(2,\tfrac{\sin(\vartheta)}{16}\rho_0\bigr)\mapsto a(z):=\big\langle\psi,\bigl(H_g^{(0)}-z\bigr)^{-1}\overline{P^{(0)}}\phi\big\rangle\in\C
\end{equation}
is analytic for fixed $\psi,\phi\in\cH^{(0)}$. One may, therefore, apply the maximum modulus principle and find that the maximal value of $|a(z)|$ is attained for some~$z$ satisfying~$|z|=\frac{\sin(\vartheta)}{16}\rho_0$. Applying the Cauchy-Schwarz inequality yields
\begin{equation}
\max_{|z|\leq\frac{\sin(\vartheta)}{16}\rho_0}|a(z)|\leq \|\psi\|\cdot \frac{6}{5}\ \frac{1}{\frac{\sin(\vartheta)}{16}\rho_0}\cdot \frac{11}{5}\ \|\phi\|,\label{calc-resolvent-norm-0-3}
\end{equation}
as Lemma~\ref{resolvent-norm-1} now is applicable and calculating the norm of the operator becomes similar to the first case. In particular, a similar calculation to~\eqref{calc-resolvent-norm-0-1} now yields~\eqref{formula-resolvent-norm-0} for~$|z-2|\leq\frac{\sin(\vartheta)}{2}\rho_0$, which implies the claim of the lemma.
\end{proof}
\clearpage

\subsection{Induction Hypothesis}\label{sect-ind-hypothesis}
Fix some $n\in\N_0$. For the induction hypothesis assume the existence of the following objects for any index $m\leq n$. First, suppose that there are simple eigenvalues $E_g^{(m)}$ of~$H_g^{(m)}$ that satisfy
\begin{equation}\label{ind-hyp-energies}		
|E_g^{(m)}-E_g^{(m-1)}|\leq 9\sqrt{\pi}C_f\cdot g\rho_{m-1}.
\end{equation}
Here, it is $C_f:=\max_{|k|\leq\rho_0}|f(k)|$ with $f$ introduced in~\eqref{coupling-function-general}. Next, assume that there are subsets $\mathcal{A}_m$ of the complex plane given by
\begin{equation}\label{ind-hyp-sets}		
\mathcal{A}_m:=\mathcal{A}\setminus\Bigl\{E_g^{(m)}+r-is\Big|r\in\R,s\in\bigl(\tfrac{\sin(\vartheta)}{4}\rho_{m},\infty\bigr)\Bigr\}
\end{equation}
with $E_g^{(m)}$ being the only spectral point of $H_g^{(m)}$ in the interior of $\mathcal{A}_m$ such that the operator $(H_g^{(m)}-z)$ is invertible for $z\in\mathcal{A}_m\setminus E_g^{(m)}$. The set $\mathcal{A}$ is taken from \eqref{def-original-set} and the sides of the rectangle parallel to the imaginary axis are cut such that the distance between $E_g^{(m)}$ and the lower edge matches the scale of $\rho_{m}$.

\begin{figure}[htb]	
\centering
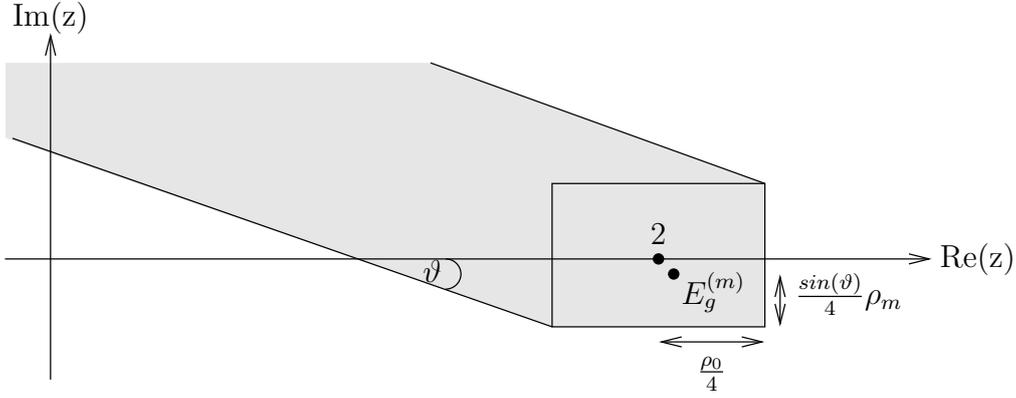
\caption{The set $\mathcal{A}_m$ in the complex plane}\label{fig-set-a-n}
\end{figure}

Further, assume the existence of a projection $P^{(m)}$ for $m\leq n$ given by
\begin{equation}\label{ind-hyp-projections-1}	
P^{(m)}:=\frac{-1}{2\pi i}\int_{\gamma^{(m)}_m}\frac{1}{H_g^{(m)}-z}\ \dx z,
\end{equation}
where $\gamma^{(m)}_m$ is the curve describing the circle with radius $r^{(m)}:=\frac{\sin(\vartheta)}{8}\rho_m$ around $E_g^{(m)}$ in the complex plane. The projections are rank-one and to satisfy the following estimate
\begin{equation}\label{ind-hyp-projections-2}	
\big\|P^{(m)}-P^{(m-1)}\big\|\leq \frac{100}{\sin(\vartheta)}\bigl(\rho_{m-1}\bigr)^{\frac{1}{4}}.
\end{equation}
Finally, suppose that for $z\in\mathcal{A}_m$ the following norm bound holds true
\begin{equation}\label{ind-hyp-resolvent-norm}	
\big\|\bigl(H_g^{(m)}-z\bigr)^{-1}\overline{P^{(m)}}\big\|\leq \frac{C_m}{\frac{\sin(\vartheta)}{2}\rho_m+|z-E_g^{(m)}|},
\end{equation}
where $C_0=43$ can be taken from Lemma~\ref{resolvent-norm-0} and the constant for $m\geq1$ is given by
\begin{equation}\label{ind-hyp-constants}
C_m:=43\cdot\Bigl(\frac{1120}{\sin(\vartheta)}+\frac{600}{\vartheta d\cdot\sin(\vartheta)}\Bigr)^m.
\end{equation}

\clearpage

\subsection{Induction Step: Construction of the Next Eigenvalue}\label{sect-ind-step}
The task for the induction step is the construction of the eigenvalue $E_g^{(n+1)}$ of $H_g^{(n+1)}$, the set $\mathcal{A}_{n+1}$ and the projection $P^{(n+1)}$ as well as establishing \eqref{ind-hyp-energies},\eqref{ind-hyp-projections-2}, and \eqref{ind-hyp-resolvent-norm} for~${m=n+1}$. As a key tool we recall the Feshbach-Schur map introduced in Section~\ref{sect-application-feshbach} and its isospectrality. The approach taken here is similar to the one used in~\cite{BachBallesterosPizzo2017}. The calculations, however, are a lot more explicit due to the form of the spin-boson Hamiltonian. Throughout this section, choose $g$ such that the assumptions of Lemma~\ref{invertibility-0} are satisfied. First, the invertibility of $H_g^{(n+1)}$ is considered and the resolvent is constructed using Neumann series expansion. Since $H_g^{(n)}$ and $H_g^{(n+1)}$ do not act on the same domain,~$\tH_g^{(n)}$ is considered as an intermediate step. Further, define

\begin{equation}\label{def-a-tilde}	
\widetilde{\mathcal{A}}_n:=\mathcal{A}\setminus\Bigl\{E_g^{(n)}+r-is\Big|r\in\R,s\in\bigl(\tfrac{\sin(\vartheta)}{2}\rho_{n+1},\infty\bigr)\Bigr\},
\end{equation}

which is a subset of $\mathcal{A}_n$ if $\gamma<\frac{1}{2}$. Here, the sides of the rectangle parallel to the imaginary axis have been cut such that the distance between $E_g^{(n)}$ and the lower edge matches the scale of $\rho_{n+1}$. It can be visualized as follows:

\begin{figure}[htb]	
\centering
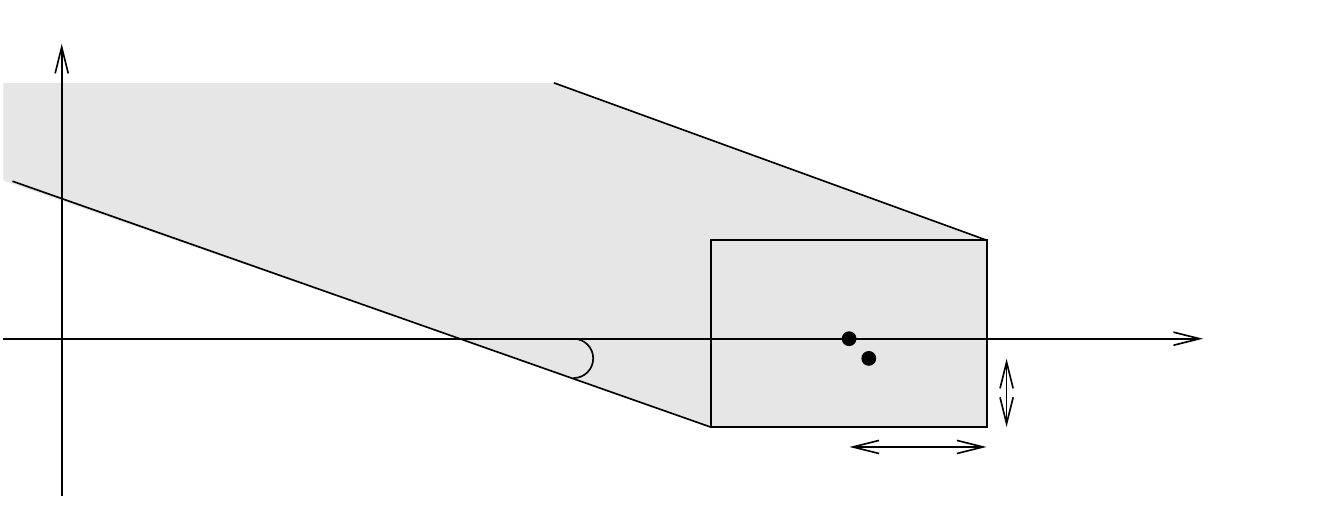
\caption{The set $\widetilde{\mathcal{A}}_n$ in the complex plane}\label{fig-set-a-n-tilde}
\end{figure}

For the analysis, choose the parameters $\rho_0$ and $\gamma$ from~\eqref{def-rhos} such that
\begin{equation}\label{parameter-restrictions-1}
\rho_0<\Bigl(\frac{\sin(\vartheta)}{100}\Bigr)^4,\quad \gamma<\Bigl(\frac{1}{5}\Bigr)^4=\frac{1}{625}.
\end{equation}
This implies that
\begin{equation}
\frac{100}{\sin(\vartheta)}\ \bigl(\rho_m\bigr)^{\frac{1}{4}}\leq\frac{1}{5}\cdot \Bigl(\frac{1}{5}\Bigr)^m.
\end{equation}
which, together with \eqref{ind-hyp-projections-2} yields the estimate
\begin{equation}\label{norm-projections}		
1-\frac{1}{5}-\frac{1}{5}\sum_{j=0}^{m-1}\Bigl(\frac{1}{5}\Bigr)^j\leq\big\|P^{(m)}\big\|\leq1+\frac{1}{5}+\frac{1}{5}\sum_{j=0}^{m-1}\Bigl(\frac{1}{5}\Bigr)^j.
\end{equation}
when the bounds for the parameters are inserted. Since the sums in \eqref{norm-projections} may be estimated by a geometric series, $P^{(m)}$ can be bounded in norm by $\frac{3}{2}$ uniformly in $m$.

\subsubsection*{Construction of the Resolvent}
We start by establishing a similar bound to~\eqref{ind-hyp-resolvent-norm} for~$(\tH_g^{(n)}-z)$ with $z\in\widetilde{\mathcal{A}}_n$.
%
\begin{lemma}\label{invertibility-tilde}	
Let $z\in\widetilde{\mathcal{A}}_n$. Then the following norm bound holds true
\begin{equation}\label{formula-invertibility-tilde}
\big\|\bigl(\tH_g^{(n)}-z\bigr)^{-1}\overline{\tP^{(n)}}\big\|\leq \frac{4}{\sin(\vartheta)}\ \frac{C_n}{\frac{\sin(\vartheta)}{2}\rho_{n+1}+|z-E_g^{(n)}|}.
\end{equation}
\end{lemma}

\begin{proof}		
Note that $E_g^{(n)}$ is the only eigenvalue of~$H_g^{(n)}$ in $\mathcal{A}_n$ and, therefore, in~$\tilde{\mathcal{A}}_n$. Due to the form of~$\tH_g^{(n)}$, the value~$E_g^{(n)}$ is also the only eigenvalue of~$\tH_g^{(n)}$ in~$\tilde{\mathcal{A}}_n$ and one may write
\begin{equation}\label{calc-invertibility-tilde-1}		
\tP^{(n)}=P^{(n)}\otimes P_{\Omega^{(n,n+1)}}
\end{equation}
with $P_{\Omega^{(n,n+1)}}$ denoting the projection to the Fock vacuum in $\cF^{(n,n+1)}$. Rewriting
\begin{align}		
\overline{\tP^{(n)}}&=\overline{P^{(n)}}\otimes\I^{(n,n+1)}+P^{(n)}\otimes P_{\Omega^{(n,n+1)}}^{\perp}.\label{calc-invertibility-tilde-2}
\end{align}
and inserting the result into the term to be estimated now yields
\begin{align}		
\bigl(\tH_g^{(n)}-z\bigr)^{-1}\overline{\tP^{(n)}}
&=\bigl(H_g^{(n)}\otimes\I^{(n,n+1)}+e^{\theta}\I^{(n)}\otimes H_{ph}^{(n,n+1)}-z\bigr)^{-1}\overline{P^{(n)}}\otimes\I^{(n,n+1)}\nonumber\\
&\quad+\bigl(E_g^{(n)}\otimes\I^{(n,n+1)}+e^{\theta}\I^{(n)}\otimes H_{ph}^{(n,n+1)}-z\bigr)^{-1}P^{(n)}\otimes P_{\Omega^{(n,n+1)}}^{\perp}.\label{calc-invertibility-tilde-3}
\end{align}
The norm can now be calculated by applying the spectral theorem to $H_{ph}^{(n,n+1)}$, allowing to reduce the operators acting on $\cH^{(n+1)}$ to operators acting on $\cH^{(n)}$ due to the different Hilbert spaces involved. This leads to
\begin{align}		
&\big\|\bigl(H_g^{(n)}\otimes\I^{(n,n+1)}+e^{\theta}\I^{(n)}\otimes H_{ph}^{(n,n+1)}-z\bigr)^{-1}\overline{P^{(n)}}\otimes\I^{(n,n+1)}\nonumber\\
&\quad+\bigl(E_g^{(n)}\otimes\I^{(n,n+1)}+e^{\theta}\I^{(n)}\otimes H_{ph}^{(n,n+1)}-z\bigr)^{-1}P^{(n)}\otimes P_{\Omega^{(n,n+1)}}^{\perp}\big\|\nonumber\\
&=\sup_{r\in\sigma\bigl(H_{ph}^{(n,n+1)}\bigr)}\big\|\bigl(H_g^{(n)}+e^{\theta}r-z\bigr)^{-1}\overline{P^{(n)}}\big\|+\sup_{r\geq\rho_{n+1}}\big\|\bigl(E_g^{(n)}+e^{-\theta}r-z\bigr)^{-1}P^{(n)}\big\|\label{calc-invertibility-tilde-4}
\end{align}
With $\|(H_g^{(n)}-z)^{-1}\overline{P^{(n)}}\|$ estimated by~\eqref{ind-hyp-resolvent-norm} in the induction hypothesis and Equation~\eqref{norm-projections} allowing to estimate the norm of the projection, it follows that
\begin{align}		
\sup_{r\in\sigma\bigl(H_{ph}^{(n,n+1)}\bigr)}\big\|\bigl(H_g^{(n)}+e^{\theta}r-z\bigr)^{-1}\overline{P^{(n)}}\big\|&\leq\sup_{r\in\sigma\bigl(H_{ph}^{(n,n+1)}\bigr)}\ \frac{C_n}{\frac{\sin(\vartheta)}{2}\rho_n+|z-(E_g^{(n)}+e^{-\theta}r)|},\nonumber\\
\sup_{r\geq\rho_{n+1}}\big\|\bigl(E_g^{(n)}+e^{-\theta}r-z\bigr)^{-1}P^{(n)}\big\|&\leq\sup_{r\geq\rho_{n+1}}\frac{3}{2}\ \big|\bigl(E_g^{(n)}+e^{-\theta}r-z\bigr)^{-1}\big|.\label{calc-invertibility-tilde-5}
\end{align}

Finally, inserting~\eqref{calc-invertibility-tilde-5} into~\eqref{calc-invertibility-tilde-4}, one can rewrite the term to obtain the form given in~\eqref{formula-invertibility-tilde}. Recall that $\rho_{n+1}<\rho_n$, and note that the definition of $\widetilde{\mathcal{A}}_n$ given in~\eqref{def-a-tilde} leads to the estimate
\begin{equation}\label{calc-invertibility-tilde-6}		
\big|E_g^{(n)}+e^{-\theta}r-z\big|\geq\frac{\sin(\vartheta)}{2}\rho_{n+1}\quad \forall z\in\widetilde{\mathcal{A}}_n,r\geq\rho_{n+1},
\end{equation}
which is also visible from Figure~\ref{fig-set-a-n-tilde}. We use another geometric estimate to put $|E_g^{(n)}-z|$ and $|E_g^{(n)}-(z-e^{-\theta}r)|$ into relation.
%
Choosing some~$\delta>0$, let $l_{n,0}$  and $l_{n,\delta}$ denote the sets
\begin{equation}\label{def-sets-l-n-delta}
l_{n,0}:=\big\{E_g^{(n)}-te^{-\theta}\big|t\in\R\big\},\quad l_{n,\delta}:=\big\{E_g^{(n)}+\delta e^{i\frac{\pi}{2}-\theta}-te^{-\theta}\big|t\in\R\big\},
\end{equation}
respectively. Here, $l_{n,0}$ describes a line going through the point $E_g^{(n)}$ in the complex plane with slope $e^{-\theta}$ and $l_{n,\delta}$ denotes its upward parallel translation by the distance $\delta$. This construction can now be used to split the set $\widetilde{\mathcal{A}}_n$ into two disjoint subsets
\begin{equation}\label{def-subsets-a-n-tilde}
\widetilde{\mathcal{A}}_n':=\Bigl(\bigcup_{\delta\geq0}l_{n,\delta}\Bigr)\cap\big\{z\in\widetilde{\mathcal{A}}_n\big|\mathrm{Re}(z)\geq\mathrm{Re}\bigl(E_g^{(n)}\bigr)\bigr\},\quad \widetilde{\mathcal{A}}_n'':=\widetilde{\mathcal{A}}_n\setminus\widetilde{\mathcal{A}}_n',
\end{equation} 
which are visualized in the following image.

\begin{figure}[htb]	
\centering
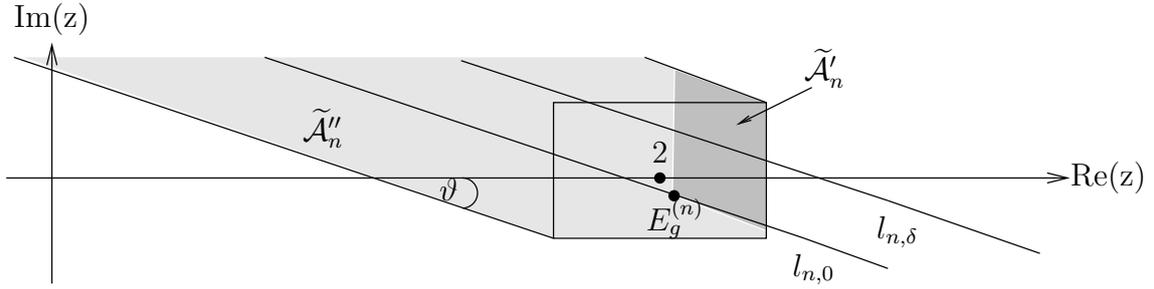
\caption{The sets defined in~\eqref{def-sets-l-n-delta} and~\eqref{def-subsets-a-n-tilde}}\label{fig-geometry-estimate-2}
\end{figure}

Note that the choice of the subsets implies that
\begin{equation}\label{calc-invertibility-tilde-7}
|E_g^{(n)}-(z-e^{-\theta}r)|\geq|E_g^{(n)}-z|,\quad \forall r\geq\rho_{n+1},z\in\widetilde{\mathcal{A}}_n'',
\end{equation}
which yields the desired relation between the two terms in the case that $z$ is taken from~$\widetilde{\mathcal{A}}_n''$ and, therefore, motivates the decomposition. When considering $\widetilde{\mathcal{A}}_n'$,  let $z_1\in l_{n,\delta}$ denote the point that satisfies $\mathrm{Im}(z_1)=\mathrm{Im}(E_g^{(n)})-\frac{\sin(\vartheta)}{4}\rho_{n+1}$, i.e., the intersection of $l_{n,\delta}$ and the line along the lower edge of the rectangle included in $\widetilde{\mathcal{A}}_n$, and $z_2\in l_{n\delta}$ denote the point that satisfies~$\mathrm{Im}(z_2)=\mathrm{Im}(E_g^{(n)})$, i.e. the intersection of $l_{n,\delta}$ and a parallel of the real axis that goes through the point $E_g^{(n)}$. Further note $z_3\in l_{n,\delta}$ to be the point with the shortest distance to $E_g^{(n)}$. They can be added into the sketch as follows:

\begin{figure}[htb]	
\centering
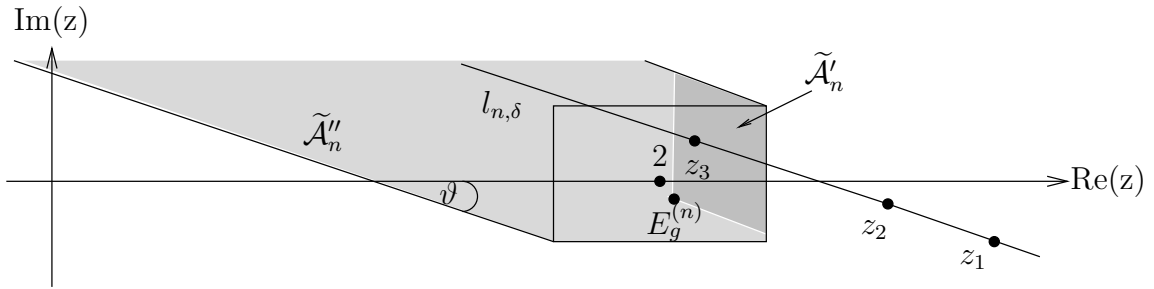
\caption{The basis for the geometric argument}\label{fig-geometry-estimate-3}
\end{figure}

The distance between $E_g^{(n)}$ and the remainig values of $z$ can now be estimated by
\begin{align}
\sup_{z\in\widetilde{\mathcal{A}}_n'}|z-E_g^{(n)}|^2&\leq|z_1-E_g^{(n)}|^2=\delta^2+|z_3-z_1|^2=\delta^2+\bigl(|z_3-z_2|+|z_2-z_1|\bigr)^2\nonumber\\
&=\delta^2+\Bigl(\frac{\delta}{\tan(\vartheta)}+\frac{\rho_{n+1}}{2}\Bigr)^2.\label{calc-invertibility-tilde-8}
\end{align}
Recalling~\eqref{calc-invertibility-tilde-6}, one can now choose~$\delta=\frac{\rho_{n+1}}{2}\sin(\vartheta)$, which leads to
\begin{align}
\frac{|E_g^{(n)}-z|}{|E_g^{(n)}-(z-e^{-\theta}r)|}&\leq\Biggl(\frac{\bigl(\frac{\rho_{n+1}}{2}\sin(\vartheta)\bigr)^2+\bigl(\frac{\rho_{n+1}\sin(\vartheta)}{2\tan(\vartheta)}+\frac{\rho_{n+1}}{2}\bigr)^2}{\bigl(\frac{\sin(\vartheta)}{2}\rho_{n+1}\bigr)^2}\Biggr)^{\frac{1}{2}}\nonumber\\
&=\Bigl(1+\Bigl(\frac{\cos(\vartheta)}{\sin(\vartheta)}+\frac{1}{\sin(\vartheta)}\Bigr)^2\Bigr)^{\frac{1}{2}}\leq\frac{2}{\sin(\vartheta)}.\label{calc-invertibility-tilde-9}
\end{align}
This can be rearranged to yield
\begin{equation}\label{calc-invertibility-tilde-10}
\frac{1}{|E_g^{(n)}-(z-e^{-\theta}r)|}\leq\frac{2}{\sin(\vartheta)}\ \frac{1}{|E_g^{(n)}-z|}.
\end{equation}
Together with~\eqref{calc-invertibility-tilde-6} it implies that
\begin{equation}\label{calc-invertibility-tilde-11}
\frac{1}{|E_g^{(n)}-(z-e^{-\theta}r)|}\leq\frac{4}{\sin(\vartheta)}\ \frac{1}{\frac{\sin(\vartheta)}{2}\rho_{n+1}+|z-E_g^{(n)}|}\quad \forall r\geq\rho_{n+1},z\in\widetilde{\mathcal{A}}_n,
\end{equation}
which can directly be inserted for for the second term in~\eqref{calc-invertibility-tilde-5}. Since $\frac{1}{\sin(\vartheta)},C_n\geq1$ by assumption, the bound given in the claim of the lemma follows.
\end{proof}

Note that~\eqref{ind-hyp-resolvent-norm} is applied for both $z$ and $(z-e^{-\theta}r)$, which motivates the shape of the sets given in~\eqref{ind-hyp-sets}. With this result in place, the invertibility of~$H_g^{(n+1)}$ is considered next.
%
Building on the results of Section~\ref{sect-application-feshbach}, the calculations can be transferred to the infrared-regularized Hamiltonians as the infrared cutoff does not affect the part of $H_g$ acting on~$\C^2$. Note, however, that the component of the projection $\P$ defined in~\eqref{def-projection-p} acting on the Fock space needs to be adjusted to match the space considered. To avoid confusion with the projections mapping to the eigenspaces, any index corresponding to the underlying space is omitted and the projection is still denoted as $\P$. Considering the Feshbach-Schur map of $H_g^{(n)}$ now yields
\begin{equation}\label{def-feshbach-n}
H_{g_0,g}^{(n)}(z)=2-z+e^{-\theta}H_{ph}^{(n)}-g_0^2\int\frac{G_{\theta}^{(n)}(k)\overline{G_{\overline{\theta}}^{(n)}(k)}}{e^{-\theta}H_{ph}^{(n)}+e^{-\theta}\omega(k)-z}\ \dx k-g^2\bigl(W_{2,0}^{(n)}+W_{1,1}^{(n)}+W_{0,2}^{(n)}\bigr),
\end{equation}
with the operators appearing in the integrals defining $W_{2,0}^{(n)}$, $W_{1,1}^{(n)}$, and $W_{0,2}^{(n)}$, respectively, involving the same infrared cutoffs as the integral for $W_{0,0}^{(n)}$ explicitly given in~\eqref{def-feshbach-n}. Recall that~$G_{\theta}^{(n)}$ denotes~$\1_{\R^3\setminus B_n}G_{\theta}$ such that the characteristic function may be separated from the coupling function when needed. The operator $H_{g_0,g}^{(n)}(z)$ is studied in detail in the following lemma. Note that the results may be derived for any $m<n$ in the same way.

\begin{lemma}\label{feshbach-estimates-cutoff}	
Let $z\in\tilde{\mathcal{A}}_n\setminus\bigl\{E_g^{(n)}\big\}$ with ${\big|z-E_g^{(n)}\big|\geq\frac{\sin(\vartheta)}{16}\rho_{n+1}}$. Then  the operator $H_{g_0,g}^{(n)}(z)$ is well-defined and it is
\begin{align}		
&\big\|g\Phi_{\theta}^{(n)}\bigl(e^{-\theta}H_{ph}^{(n)}-z\bigr)^{-1}\big\|\leq g\cdot C_{FS}\Bigl(\Big\|\frac{G}{\sqrt{\omega}}\Big\|_{L^2}+\|G\|_{L^2}\Bigr)\label{formula-feshbach-estimates-cutoff-1}\\
&\big\|\bigl(e^{-\theta}H_{ph}^{(n)}-z\bigr)^{-1}g\Phi_{\theta}^{(n)}\big\|\leq g\cdot C_{FS}\Bigl(\Big\|\frac{G}{\sqrt{\omega}}\Big\|_{L^2}+\|G\|_{L^2}\Bigr),\label{formula-feshbach-estimates-cutoff-2}
\end{align}
where $\Phi_{\theta}^{(n)}=a^*(G_{\theta}^{(n)})+a(G_{\bar{\theta}}^{(n)})$. The constant can be chosen as
\begin{equation}\label{def-c-fs}
C_{FS}=2+7(\vartheta d)^{-1}.
\end{equation}
\end{lemma}

\begin{proof}		
When considering $H_{g_0,g}^{(n)}(z)$, one can choose the same approach as used in Section~\ref{sect-application-feshbach}. As a first step, the existence of the term $\overline{\P}(\overline{\P}(H_g^{(n)}-z)\overline{\P})^{-1}$ is shown. Its norm can be estimated by
\begin{equation}\label{calc-feshbach-estimates-cutoff-1}	
\big\|\overline{\P}\bigl(\overline{\P}\bigl(H_g^{(n)}-z\bigr)\overline{\P}\bigr)^{-1}\big\|=\big\|\bigl(e^{-\theta}H_{ph}^{(n)}-z\bigr)^{-1}\big\|=\sup_{r\geq\rho_n}\frac{1}{|e^{-\theta}r-z|}\leq\frac{1}{\delta},
\end{equation}
using another geometric estimate. Note, however, that $\delta$ is not the same as in Section~\ref{sect-application-feshbach}. Here, the constant is estimated by calculating the distance between two parallel lines instead. The setup is sketched in the following image.\\

\begin{figure}[htb]	
\centering
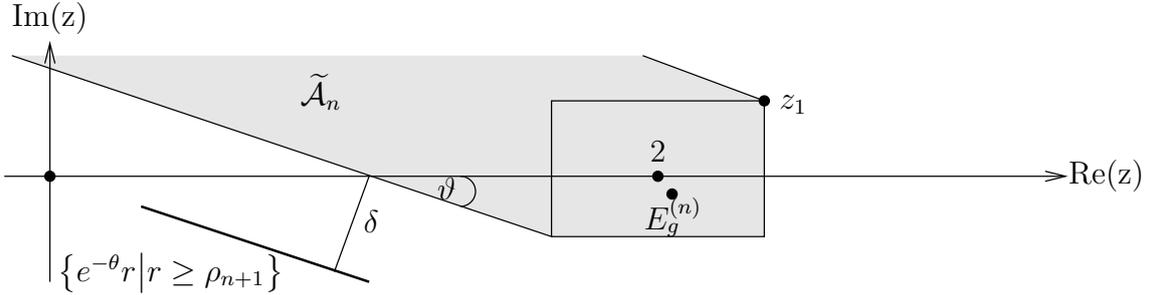
\caption{The basis for the geometric estimate}\label{fig-geometry-estimate-4}
\end{figure}

The distance between the set~$\{e^{-\theta}r|r\geq\rho_{n+1}\}$ and the lower edge of~$\widetilde{\mathcal{A}}_n$ can now be explicitly stated as
\begin{equation}\label{calc-feshbach-estimates-cutoff-2}	
\delta=\sin(\vartheta)\Bigl[2-\frac{\rho_0}{4}-\cos(\vartheta)\Bigl(\mathrm{Im}\bigl(E_g^{(n)}\bigr)+\frac{\sin(\vartheta)}{2}\rho_{n+1}\Bigr)\Bigr].
\end{equation}
Note the two lines are parallel and that $\delta$, therefore, allows for an estimate that is uniform in both $z$ and $r$. As done in Section~\ref{sect-application-feshbach}, the leading order in $\vartheta$ is explicitly written out and the bound is replaced by $\vartheta d$ for a suitable constant $d$. Recall that $\rho_n$ is defined on an exponential scale and that the energy differences, by the induction hypothesis, are of the same magnitude. One can, therefore, estimate $\mathrm{Im}(E_g^{(n)})$ by a convergent geometric series and give a constant~$d$ that is uniform in $m\leq n$. The next terms to be considered are the ones given in~\eqref{formula-feshbach-estimates-cutoff-1} and~\eqref{formula-feshbach-estimates-cutoff-2}, respectively. As the proofs are similar, only the first one is carried out explicitly here. First, expanding the term noted in~\eqref{formula-feshbach-estimates-cutoff-1} leads to
\begin{align}		
&\big\|\P \bigl(H_g^{(n)}-z\bigr)\overline{\P}\bigl(\overline{\P}\bigl(H_g^{(n)}-z\bigr)\overline{\P}\bigr)^{-1}\overline{\P}\big\|=g\cdot\big\|\sigma_1 P_{\uparrow}\Phi_{\theta}^{(n)}\bigl(e^{-\theta}H_{ph}^{(n)}-z)^{-1}\big\|\nonumber\\
&\leq g\cdot\big\|\Phi_{\theta}^{(n)}\bigl(H_{ph}^{(n)}+1\bigr)^{-\frac{1}{2}}\big\|\cdot\big\|\bigl(H_{ph}^{(n)}+1\bigr)^{-\frac{1}{2}}\big\|\cdot\big\|\bigl(e^{-\theta}H_{ph}^{(n)}-z)^{-1}\bigl(H_{ph}^{(n)}+1\bigr)\big\|.\label{calc-feshbach-estimates-cutoff-3}
\end{align}
The three factors on the right side of~\eqref{calc-feshbach-estimates-cutoff-3} are now considered separately. First, Lemma~\ref{standard-estimate} yields the bound
\begin{equation}\label{calc-feshbach-estimates-cutoff-4}	
\big\|\Phi_{\theta}^{(n)}\bigl(H_{ph}^{(n)}+1\bigr)^{-\frac{1}{2}}\big\|\leq2\Bigl(\Big\|\frac{G_{\theta}}{\sqrt{\omega}}\Big\|_{L^2}+\|G_{\theta}\|_{L^2}\Bigr).
\end{equation}
Recalling that $H_{ph}^{(n)}$ is self-adjoint, the second term in~\eqref{calc-feshbach-estimates-cutoff-3} is estimated using the spectral theorem. It follows that
\begin{equation}\label{calc-feshbach-estimates-cutoff-5}	
\big\|\bigl(H_{ph}^{(n)}+1\bigr)^{-\frac{1}{2}}\big\|=\sup_{r\geq\rho_n}\frac{1}{\sqrt{r+1}}\leq1.
\end{equation}
The remaining factor in~\eqref{calc-feshbach-estimates-cutoff-3} can be rewritten just as done in Section~\ref{sect-application-feshbach}, yielding
\begin{align}		
\Big\|\frac{H_{ph}^{(n)}+1}{e^{-\theta}H_{ph}^{(n)}-z}\Big\|&=\Big\|e^{\theta}+\frac{e^{\theta}z+1}{e^{-\theta}H_{ph}^{(n)}-z}\Big\|\leq1+\sup_{r\geq\rho_n}\frac{|z|+1}{|e^{-\theta}r-z|}\nonumber\\
&\leq1+\sup_{r\geq\rho_n}\frac{|z|}{|e^{-\theta}r-z|}+\frac{1}{\vartheta d}.\label{calc-feshbach-estimates-cutoff-6}
\end{align}
Note, however, that $|z|$ is not bounded due to the definition of the set $\widetilde{\mathcal{A}}_n$. The tool to bound the quotient is, once more, a geometric estimate. Recall Figure~\ref{fig-set-a-n-tilde}. When considering~${z\in\widetilde{\mathcal{A}}_n}$ with $\mathrm{Re}(z)\leq0$, note that
\begin{equation}\label{calc-feshbach-estimates-cutoff-7}	
|z|\leq|e^{-\theta}r-z|,\quad \forall r\geq\rho_n,
\end{equation}
as subtracting $e^{-\theta}$ from these values of $z$ increases the distance to the origin since
\begin{align}\label{calc-real-imaginary-part}
\mathrm{Re}(z-e^{-\theta}r)&=\mathrm{Re}(z)-r\cos(\vartheta)<\mathrm{Re}(z),\nonumber\\
\mathrm{Im}(z-e^{-\theta}r)&=\mathrm{Im}(z)+r\sin(\vartheta)>\mathrm{Im}(z)
\end{align}
i.e., the absolute value of both the real and the imaginary part increases for the values of~$z$ considered. This implies~\eqref{calc-feshbach-estimates-cutoff-7} and the quotient in the middle of the right side of~\eqref{calc-feshbach-estimates-cutoff-6} can be estimated by 1. Conversely, if $\mathrm{Re}(z)>0$, note that the point yielding the maximal absolute value of $|z|$ is the top right corner of the rectangle marked as $z_1$ in Figure~\ref{fig-geometry-estimate-4}. It can be estimated by
\begin{equation}\label{calc-feshbach-estimates-cutoff-8}	
|z_1|=\sqrt{\Bigr(2+\frac{\rho_0}{4}\Bigr)^2+\Bigl(\frac{\sin(\vartheta)}{4}\rho_0\Bigr)^2}\leq\sqrt{\Bigr(\frac{9}{4}\Bigr)^2+\Bigl(\frac{1}{8}\Bigr)^2}\leq\frac{5}{2}
\end{equation}
using the restrictions on the parameters, i.e., $\vartheta<\frac{\pi}{6}$ and $\rho_0<1$. Inserting the result into~\eqref{calc-feshbach-estimates-cutoff-6} now implies that
\begin{equation}\label{calc-feshbach-estimates-cutoff-9}	
1+\sup_{r\geq\rho_n}\frac{|z|}{|e^{-\theta}r-z|}+\frac{1}{\vartheta d}\leq1+\frac{7}{2}\ \frac{1}{\vartheta d},
\end{equation}
as the denominator of the quotient can be estimated geometrically when $|z|$ is bounded and $\frac{5}{2}(\vartheta d)^{-1}\geq1$ by the assumptions for the parameters. This yields the value of $C_{FS}$ given in the claim of the lemma. Finally, consider
\begin{align}		
&\big\|\P\bigl(H_g^{(n)}-z\bigr)\overline{\P}\bigl(\overline{\P}\bigl(H_g^{(n)}-z\bigr)\overline{\P}\bigr)^{-1}\overline{\P}\bigl(H_g^{(n)}-z\bigr)\P\big\|\nonumber\\
&=g^2\cdot\big\|\Phi_{\theta}^{(n)}\bigl(e^{-\theta}H_{ph}^{(n)}-z\bigr)^{-1}\Phi_{\theta}^{(n)}\big\|\nonumber\\
&\leq g^2\cdot\big\|\Phi_{\theta}^{(n)}\bigl(H_{ph}^{(n)}+1\bigr)^{-\frac{1}{2}}\big\|\cdot\big\|\bigl(H_{ph}^{(n)}+1\bigr)^{-\frac{1}{2}}\bigl(e^{-\theta}H_{ph}^{(n)}-z\bigr)\bigl(H_{ph}^{(n)}+1\bigr)^{-\frac{1}{2}}\big\|\nonumber\\
&\quad\quad\cdot\big\|\bigl(H_{ph}^{(n)}+1\bigr)^{-\frac{1}{2}}\Phi_{\theta}^{(n)}\big\|\label{calc-feshbach-estimates-cutoff-10}
\end{align}
which can be treated using the same estimates as before and also yields a finite bound, thus implying that $H_{g_0,g}^{(n)}(z)$ is well-defined for the values of $z$ considered.
\end{proof}

Note that the same calculations may also be applied to $\tH_g^{(n)}$, leading to an operator of a similar form acting on $\cF^{(n+1)}$ which involves $H_{ph}^{(n+1)}$ instead of $H_{ph}^{(n)}$, but still includes the~${n^{th}}$ infrared cutoff in the terms involving the coupling function. In particular, the bounds given in Lemma~\ref{feshbach-estimates-cutoff} also apply to $\tH_g^{(n)}$ due to the estimates still holding if $H_{ph}^ {(n)}$ is replaced by $H_{ph}^{(n+1)}$. For the following computations, the operator $F_{\P}(\tH_g^{(n)})$ is denoted as $\tH_{g_0,g}^{(n)}$. It is given by
\begin{equation}\label{def-feshbach-n-tilde}
\tH_{g_0,g}^{(n)}(z)=2-z+e^{-\theta}H_{ph}^{(n+1)}-g_0^2\int\frac{G_{\theta}^{(n)}(k)\overline{G_{\overline{\theta}}^{(n)}(k)}}{e^{-\theta}H_{ph}^{(n+1)}+e^{-\theta}\omega(k)-z}\ \dx k-g^2\bigl(\tW_{2,0}^{(n)}+\tW_{1,1}^{(n)}+\tW_{0,2}^{(n)}\bigr).
\end{equation}
Here, $\tW_{m_1,m_2}^{(n)}$ reflects the notation of $\tH_{g_0,g}^{(n)}$ and marks that the operators involve $H_{ph}^{(n+1)}$ and, therefore, act on $\cF^{(n+1)}$, but that the interaction term only includes the ${n^{th}}$ infrared cutoff. Finally, the resolvent of $H_g^{(n+1)}$ is considered. By isospectrality and the above considerations, the desired result is obtained by studying the invertibility of~$H_{g_0,g}^{(n+1)}(z)$. For convenience, the norm bounds needed for the convergence of the Neumann series are split into two separate lemmas which are then assembled in the proof of Theorem~\ref{invertibility-n+1}. We prepare the calculations with another lemma.
%
\begin{lemma}\label{feshbach-norm}		
Let $H_{g_0,g}^{(n)}(z)$ and $\tH_{g_0,g}^{(n)}(z)$ be as defined in~\eqref{def-feshbach-n} and~\eqref{def-feshbach-n-tilde}, respectively, and consider $z\in\tilde{\mathcal{A}}_n\setminus\bigl\{E_g^{(n)}\big\}$ with ${\big|z-E_g^{(n)}\big|\geq\frac{\sin(\vartheta)}{16}\rho_{n+1}}$. Then the following norm bounds hold true
\begin{align}
&\big\|H_{g_0,g}^{(n)}(z)^{-1}\big\|\leq25\ \frac{C_n}{\frac{\sin(\vartheta)}{2}\rho_{n+1}+|z-E_g^{(n)}|},\label{formula-feshbach-norm-n}\\
&\big\|\tH_{g_0,g}^{(n)}(z)^{-1}\big\|\leq \frac{28}{\sin(\vartheta)}\ \frac{C_n}{\frac{\sin(\vartheta)}{2}\rho_{n+1}+|z-E_g^{(n)}|}.\label{formula-feshbach-norm-n-tilde}
\end{align}
\end{lemma}

\begin{proof}		
From the isospectrality of the Feshbach map it follows that
\begin{equation}
H_{g_0,g}^{(n)}(z)^{-1}=\P\bigl(H_g^{(n)}-z\bigr)^{-1}\P,\quad \tH_{g_0,g}^{(n)}(z)^{-1}=\P\bigl(\tH_g^{(n)}-z\bigr)^{-1}\P,
\end{equation}
where the projection $\P$ is of the form~\eqref{def-projection-p} with the identity operator involved adjusted to match the space considered. To avoid confusion with the projections to the eigenspaces, the infrared cutoff is not explicitly included here. With~\eqref{def-projection-p} implying~$\|\P\|=1$ for both cases considered, the bounds~\eqref{formula-feshbach-norm-n} and~\eqref{formula-feshbach-norm-n-tilde} can be derived by estimating the norm of the respective resolvent. Since the proofs are similar, the calculations are only carried out for the first bound. As $z$ is restricted such that it respects a minimal distance to $E_g^{(n)}$, one may calculate
\begin{align}
\big\|(H_g^{(n)}-z)^{-1}\big\|&=\big\|(H_g^{(n)}-z)^{-1}(\overline{P^{(n)}}+P^{(n)})\big\|\nonumber\\
&\leq\big\|(H_g^{(n)}-z)^{-1}\overline{P^{(n)}}\big\|+\big\|(H_g^{(n)}-z)^{-1}\tP^{(n)}\big\|\nonumber\\
&\leq \frac{C_n}{\frac{\sin(\vartheta)}{2}\rho_{n+1}+|z-E_g^{(n)}|}+\Big\|\frac{P_n}{E_g^{(n)}-z}\Big\|,
\end{align}
using~\eqref{ind-hyp-resolvent-norm} from the induction hypothesis and the fact that~$E_g^{(n)}$ is an eigenvalue of $H_g^{(n)}$. Note that $E_g^{(n)}$ is also an eigenvalue of $\tH_g^{(n)}$. When considering $\tH_g^{(n)}$, Lemma~\ref{invertibility-tilde} is used instead of~\eqref{ind-hyp-resolvent-norm}. The result can be rewritten and estimated further, leading to
\begin{align}
&\frac{C_n}{\frac{\sin(\vartheta)}{2}\rho_{n+1}+|z-E_g^{(n)}|}+\Big\|\frac{P_n}{E_g^{(n)}-z}\Big\|\nonumber\\
&\leq \frac{C_n}{\frac{\sin(\vartheta)}{2}\rho_{n}+|z-E_g^{(n)}|}+\frac{3}{2}\ \frac{2}{|E_g^{(n)}-z|+\frac{\sin(\vartheta)}{16}\rho_{n+1}}\nonumber\\
&\leq\Bigl(C_n+24\Bigr)\frac{1}{\frac{\sin(\vartheta)}{2}\rho_{n+1}+|z-E_g^{(n)}|},
\end{align}
which yields the desired bounds as both $C_n$ and $\frac{1}{\sin(\vartheta)}$ are greater than one.
\end{proof}

With this estimate in place, one may now begin to establish the bounds needed for the Neumann series expansion.
%
\begin{lemma}\label{neumann-estimate-1}		
Let $W_{m_1,m_2}^{(n,n+1)}$ with $m_1,m_2\in\{0,1,2\}$ such that $m_1+m_2=2$ denote the interaction terms involving the coupling function restricted to~$\rho_{n+1}\leq|k|\leq\rho_n$ and~$H_{ph}^{(n,n+1)}=\I^{(n)}\otimes H_{ph}^{(n,n+1)}$ denote the part of the photon field Hamiltonian acting on $\cF^{(n,n+1)}$ respectively. Then for $\varepsilon>0$ and $\phi\in\cF^{(n+1)}$, the following norm bounds hold
\begin{align}
&\big\|\tW_{m_1,m_2}^{(n,n+1)}\phi\big\|\leq\frac{C_{m_1,m_2}}{\vartheta d}\int_{\{\rho_{n+1}\leq|k|\leq\rho_n\}}\Bigl(\frac{1}{\omega(k)}+\frac{1}{\varepsilon}\Bigr)\big|G_{\theta}(k)\big|^2\ \dx^3 k\ \big\|\bigl(H_{ph}^{(n,n+1)}+\varepsilon\bigr)\phi\big\|\label{formula-neumann-estimate-1-separate}\\
&\big\|\bigl(\tW_{2,0}^{(n,n+1)}+\tW_{1,1}^{(n,n+1)}+\tW_{0,2}^{(n,n+1)}\bigr) \bigl(H_{ph}^{(n,n+1)}+\varepsilon\bigr)^{-1}\big\|\nonumber\\
&\quad\leq\frac{C}{\vartheta d}\int_{\{\rho_{n+1}\leq|k|\leq\rho_n\}}\Bigl(\frac{1}{\omega(k)}+\frac{1}{\varepsilon}\Bigr)\big|G_{\theta}(k)\big|^2\ \dx^3 k,\label{formula-neumann-estimate-1-assembled}
\end{align}
where the constants $C_{m_1,m_2},C>0$ neither depend on $n$ nor $\theta$. In particular, one may chose the constant in~\eqref{formula-neumann-estimate-1-assembled} as $C=10$.
\end{lemma}

\begin{proof}		
To derive the second inequality, the estimates for $W_{2,0}^{(n,n+1)}$, $W_{1,1}^{(n,n+1)}$ and~$W_{0,2}^{(n,n+1)}$ need to be developed first. First, we consider $W_{0,2}^{(n,n+1)}$. 
%
By continuity and the Cauchy-Schwarz inequality, the term may be written as
\begin{align}		
\big\|W_{0,2}^{(n,n+1)}\phi\big\|^2
&=\int\big\langle w_{0,2}^{(n,n+1)}a(k_1)a(k_2)\phi\big| w_{0,2}^{(n,n+1)}a(k_3)a(k_4)\phi\big\rangle\ \dx k_1\dx k_2\dx k_3\dx k_4\nonumber\\
&\leq\int\Bigl(\big\|w_{0,2}^{(n,n+1)}a(k_1)a(k_2)\phi\big\|\cdot\big\|w_{0,2}^{(n,n+1)}a(k_3)a(k_4)\phi\big\|\Bigr)\ \dx k_1\dx k_2\dx k_3\dx k_4\nonumber\\
&\leq\Bigr(\int\big\|w_{0,2}^{(n,n+1)}\big\|\cdot\big\| a(k_1)a(k_2)\phi\big\|\ \dx k_1\dx k_2\Bigl)^2,\label{calc-neumann-estimate-1-2}
\end{align}
with $\dx k_j$ denoting $\dx^3 k_j$ for readability. Since $w_{0,2}$ is a function of $H_{ph}^{(n+1)}$, the spectral theorem and an argument similar to the one used in the proof of Lemma \ref{standard-estimate} lead to
\begin{align}		
&\Bigr(\int\sup_{r\geq\rho_{n+1}}\big|w_{0,2}^{(n,n+1)}\big|\cdot\big\| a(k_1)a(k_2)\phi\big\|\ \dx k_1\dx k_2\Bigl)^2\label{calc-neumann-estimate-1-3}\\
&\leq\Bigr(\int\sup_{r\geq\rho_{n+1}}\big|w_{0,2}^{(n,n+1)}(r)\big|^2\frac{\dx k_1\dx k_2}{\omega(k_1)\omega(k_2)}\Bigr)\Bigl(\int\omega^{(n,n+1)}(k_1)\omega^{(n,n+1)}(k_2)\big\|a(k_1)a(k_2)\phi\big\|\ \dx k_1\dx k_2\Bigl),\nonumber
\end{align}
with $\omega^{(n,n+1)}$ denoting $\1_{B_n\setminus B_{n+1}}\omega$. The second term can now be identified as~$\|H_{ph}^{(n,n+1)}\phi\|^2$. Inserting the definition of $w_{0,2}$ then yields for the first term
\begin{align}		
\int\sup_{r\geq\rho_{n+1}}\big|w_{0,2}^{(n,n+1)}(r)\big|^2\frac{\dx k_1\dx k_2}{\omega(k_1)\omega(k_2)}
&\leq\sup_{s\geq\rho_{n+1}}\frac{1}{|e^{-\theta}s-z|^2}\Bigl(\int\frac{\big|G_{\bar{\theta}}^{(n,n+1)}(k)\big|^2}{\omega(k)}\dx k\Bigr)^2.\label{calc-neumann-estimate-1-4}
\end{align}
Here, the term $|e^{-\theta}s-z|$ can once more be estimated geometrically by the bound~$\vartheta d$ using the same argument as in Lemma~\ref{feshbach-estimates-cutoff}. Reassembling the two terms now leads to
\begin{align}		
\big\|W_{0,2}^{(n,n+1)}\phi\big\|&\leq\frac{1}{\vartheta d}\int\frac{1}{\omega(k)}\big|G_{\bar{\theta}}^{(n,n+1)}(k)\big|^2\ \dx k\ \big\|H_{ph}^{(n,n+1)}\phi\big\|\nonumber\\
&\leq\frac{1}{\vartheta d}\int\Bigl(\frac{1}{\omega(k)}+\frac{1}{\varepsilon}\Bigr)\big|G_{\bar{\theta}}^{(n,n+1)}(k)\big|^2\ \dx k\ \big\|\bigl(H_{ph}^{(n,n+1)}+\varepsilon\bigr)\phi\big\|,\label{calc-neumann-estimate-1-5}
\end{align}
with the last inequality holding for any $\varepsilon>0$, yielding the desired norm bound and the constant $C_{0,2}=1$.
%
The next term to be considered is $W_{1,1}^{(n,n+1)}$. We rewrite the term as
\begin{align}		
\big\|W_{1,1}^{(n,n+1)}\phi\big\|^2
&=\int\big\langle a^*(k_1)w_{1,1}^{(n,n+1)}a(k_2)\big| a^*(k_3)w_{1,1}^{(n,n+1)}a(k_4)\big\rangle\ \dx k_1\dx k_2\dx k_3\dx k_4\nonumber\\
&=\int\big\langle\phi\big|a^*(k_2)\bigl(w_{1,1}^{(n,n+1)}\bigr)^*a(k_1)a^*(k_3)w_{1,1}^{(n,n+1)}a(k_4)\big\rangle\ \dx k_1\dx k_2\dx k_3\dx k_4\nonumber\\
&=\int\big\langle\phi\big|a^*(k_3)a^*(k_2)\bigl(w_{1,1}^{(n,n+1)}\bigr)^*w_{1,1}^{(n,n+1)}a(k_1)a(k_4)\big\rangle\ \dx k_1\dx k_2\dx k_3\dx k_4\nonumber\\
&\quad+\int\big\langle\phi\big|a^*(k_2)\bigl(w_{1,1}^{(n,n+1)}\bigr)^*w_{1,1}^{(n,n+1)}a(k_4)\big\rangle\ \dx k_1\dx k_2\dx k_4\label{calc-neumann-estimate-1-6}
\end{align}
Note that it is $w_{1,1}^{(n,n+1)}=w_{1,1}^{(n,n+1)}[H_{ph}^{(n+1)}+\omega(k_1)]$ after applying the pull-through formulae in the last step. Similarly, it is $(w_{1,1}^{(n,n+1)})^*=(w_{1,1}^{(n,n+1)}[H_{ph}^{(n+1)}+\omega(k_3)])^*$. Let $I_1$ and $I_2$ denote the first and second summand of the above result respectively. First, it is
\begin{align}		
I_1&=\int\big\langle\phi\big|a^*(k_3)a^*(k_2)\bigl(w_{1,1}^{(n,n+1)}\bigr)^*w_{1,1}^{(n,n+1)}a(k_1)a(k_4)\big\rangle\ \dx k_1\dx k_2\dx k_3\dx k_4\nonumber\\
&=\int\big\langle w_{1,1}^{(n,n+1)}a(k_2)a(k_3)\phi\big|w_{1,1}^{(n,n+1)}a(k_1)a(k_4)\big\rangle\ \dx k_1\dx k_2\dx k_3\dx k_4,\label{calc-neumann-estimate-1-7}
\end{align}
which can be treated similarly to $W_{0,2}^{(n,n+1)}$, yielding
\begin{align}		
I_1&\leq\frac{4}{(\vartheta d)^2}\int\frac{\big|\1^{n,n+1}(k_1)G_{\bar{\theta}}(k_1)\1^{n,n+1}(k_2)G_{\theta}(k_2)\big|^2}{\omega(k_1)\omega(k_2)}\ \dx k_1\dx k_2\ \big\|H_{ph}^{(n,n+1)}\phi\big\|^2\nonumber\\
&\leq\frac{4}{(\vartheta d)^2}\Bigl(\int_{\{\rho_{n+1}\leq|k|\leq\rho_n\}}\Bigl(\frac{1}{\omega(k)}+\frac{1}{\varepsilon}\Bigr)\big|G_{\theta}(k)\big|^2\dx k\Bigr)^2\ \big\|\bigl(H_{ph}^{(n,n+1)}+\varepsilon\bigr)\phi\big\|^2.\label{calc-neumann-estimate-1-9}
\end{align}
The next term to be considered is $I_2$. The argument can be applied similarly, however, note that $w_{1,1}^{(n,n+1)}=w_{1,1}^{(n,n+1)}[H_{ph}^{(n+1)},k_1,k_2,k_1,k_4]$, i.e. both terms involving $\overline{G_{\bar{\theta}}}$ depend on~$k_1$. Applying the Cauchy-Schwarz inequality yields
\begin{align}		
I_2&=\int\big\langle\phi\big|a^*(k_2)\bigl(w_{1,1}^{(n,n+1)}\bigr)^*w_{1,1}^{(n,n+1)}a(k_4)\big\rangle\ \dx k_1\dx k_2\dx k_4\nonumber\\
&\leq\int\big\| w_{1,1}^{(n,n+1)}a(k_2)\phi\big\|\cdot\big\|w_{1,1}^{(n,n+1)}a(k_4)\big\|\ \dx k_1\dx k_2\dx k_4.\label{calc-neumann-estimate-1-10}
\end{align}
Inserting the definition of $w_{1,1}$ and applying a similar argument as previously leads to
\begin{align}		
I_2&\leq\frac{4}{(\vartheta d)^2}\Bigl(\int_{\{\rho_{n+1}\leq|k|\leq\rho_n\}}\frac{\big|G_{\theta}(k)\big|^2}{\varepsilon}\dx k\Bigr)\Bigl(\int_{\{\rho_{n+1}\leq|k|\leq\rho_n\}}\frac{\big|G_{\theta}(k)\big|^2}{\omega(k)}\ \dx k\Bigr)\ \big\|\bigl(H_{ph}^{(n,n+1)}+\varepsilon\bigr)\phi\big\|^2\nonumber\\
&\leq\frac{4}{(\vartheta d)^2}\Bigl(\int_{\{\rho_{n+1}\leq|k|\leq\rho_n\}}\Bigl(\frac{1}{\omega(k)}+\frac{1}{\varepsilon}\Bigr)\big|G_{\theta}(k)\big|^2\ \dx^3 k\Bigr)^2\big\|\bigl(H_{ph}^{(n,n+1)}+\varepsilon\bigr)\phi\big\|^2.\label{calc-neumann-estimate-1-16}
\end{align}
The final estimate for $W_{1,1}^{(n,n+1)}$ now follows from
\begin{equation}\label{calc-neumann-estimate-1-17}	
\big\|W_{1,1}^{(n,n+1)}\big\|\leq\sqrt{I_1+I_2}\leq\sqrt{I_1}+\sqrt{I_2}
\end{equation}
and adding the different bounds established above, yielding $C_{1,1}=4$. 
%
The last term to be considered is~$W_{2,0}^{(n,n+1)}$. Taking the same approach as previously, one arrives at
\begin{align}		
\big\|W_{2,0}^{(n,n+1)}\phi\big\|^2
&=\int\big\langle a^*(k_1)a^*(k_2)w_{2,0}^{(n,n+1)}\big| a^*(k_3)a^*(k_4)w_{2,0}^{(n,n+1)}\big\rangle\ \dx k_1\dx k_2\dx k_3\dx k_4\label{calc-neumann-estimate-1-18}\\
&=\int\big\langle\phi\big|\bigl(w_{2,0}^{(n,n+1)}\bigr)^*a(k_2)a(k_1)a^*(k_3)a^*(k_4)w_{2,0}^{(n,n+1)}\phi\big\rangle\ \dx k_1\dx k_2\dx k_3\dx k_4.\nonumber
\end{align}
Applying the canonical commutation relations to rearrange the operators leads to
\begin{align}		
a(k_2)a(k_1)a^*(k_3)a^*(k_4)&=a^*(k_3)a^*(k_4)a(k_2)a(k_1)\nonumber\\
&\quad+\delta(k_2-k_4)a^*(k_3)a(k_1)+\delta(k_3-k_1)a^*(k_4)a(k_2)\nonumber\\
&\quad+\delta(k_2-k_3)a^*(k_4)a(k_1)+\delta(k_1-k_4)a^*(k_3)a(k_2)\nonumber\\
&\quad+\delta(k_1-k_3)\delta(k_2-k_4)+\delta(k_1-k_4)\delta(k_2-k_3).\label{calc-neumann-estimate-1-19}
\end{align}
The result can now be inserted back into the integral, yielding a total of seven terms. The first of them, including two creation and annihilation operators each, can be treated similarly to the one occurring for $W_{0,2}^{(n,n+1)}$ and in $I_1$. The following four terms, each including only one creation and one annihilation operator, yield similar calculations as done for $I_2$. Considering one of the remaining terms leads to the estimate
\begin{align}		
I_3&:=\int\big\langle\phi\big|\bigl(w_{2,0}^{(n,n+1)}\bigr)^*\delta(k_1-k_3)\delta(k_2-k_4)w_{2,0}^{(n,n+1)}\phi\big\rangle\ \dx k_1\dx k_2\dx k_3\dx k_4\nonumber\\
&\leq\frac{1}{(\vartheta d)^2}\Bigl(\int_{\{\rho_{n+1}\leq|k|\leq\rho_n\}}\big|G_{\theta}(k)\big|^2\ \dx k\Bigr)^2\ \|\phi\|^2.\label{calc-neumann-estimate-1-20}
\end{align}
The second term can be treated similarly and yields the same bound. Taking the square root and expanding the term now leads to
\begin{align}		
\sqrt{I_3}&\leq\frac{1}{\vartheta d}\Bigl(\int_{\{\rho_{n+1}\leq|k|\leq\rho_n\}}\big|G_{\theta}(k)\big|^2\ \dx k\Bigr)\ \|\phi\|\nonumber\\
&\leq\frac{1}{\vartheta d}\Bigl(\int_{\{\rho_{n+1}\leq|k|\leq\rho_n\}}\frac{1}{\varepsilon}\big|G_{\theta}(k)\big|^2\ \dx k\Bigr)\ \|\varepsilon\phi\|\nonumber\\
&\leq\frac{1}{\vartheta d}\Bigl(\int_{\{\rho_{n+1}\leq|k|\leq\rho_n\}}\Bigl(\frac{1}{\omega}+\frac{1}{\varepsilon}\Bigr)\big|G_{\theta}(k)\big|^2\ \dx k\Bigr)\ \big\|\bigl(H_{ph}^{(n+1)}+\varepsilon)\phi\big\|.\label{calc-neumann-estimate-1-21}
\end{align}
The bound for $\|W_{2,0}^{(n,n+1)}\phi\|$ now follows by adding the bounds for the different terms, leading to $C_{2,0}=5$ and completing the proof of (i). Expanding
\begin{equation}\label{calc-neumann-estimate-1-22}	
\big\|\bigl(\tW_{2,0}^{(n,n+1)}+\tW_{1,1}^{(n,n+1)}+\tW_{0,2}^{(n,n+1)}\bigr)\big\|\leq\big\|\tW_{2,0}^{(n,n+1)}\phi\big\|+\big\|\tW_{1,1}^{(n,n+1)}\phi\big\|+\big\|\tW_{0,2}^{(n,n+1)}\big\|
\end{equation}
and choosing $\phi=\bigl(H_{ph}^{(n,n+1)}+\varepsilon\bigr)^{-1}\psi$ as done for the proof of Lemma \ref{standard-estimate} now yields~\eqref{formula-neumann-estimate-1-assembled} with $C=10$ by adding the constants.
\end{proof}
%
\begin{lemma}\label{neumann-estimate-2}	
Let $z\in\widetilde{\mathcal{A}}_n$ and $|z-E_g^{(n)}|\geq\frac{\sin(\vartheta)}{16}\rho_{n+1}$. Then the norm bound
\begin{equation}\label{formula-neumann-estimate-2}		
\big\|\bigl(H_{ph}^{(n,n+1)}+\rho_{n+1}\bigr)\tH_{g_0,g}^{(n)}(z)^{-1}\big\|\leq\frac{100\ C_n}{\sin(\vartheta)}
\end{equation}
holds. Here, the constant $C_n$ is the same as in~\eqref{ind-hyp-resolvent-norm}.
\end{lemma}

\begin{proof}		
Recall that $\tH_{g_0,g}^{(n)}(z)$ is given by~\eqref{def-feshbach-n-tilde} and note that
\begin{equation}\label{calc-neumann-estimate-2-1}		
H_{ph}^{(n+1)}=H_{ph}^{(n)}\otimes\I^{(n,n+1)}+\I^{(n)}\otimes H_{ph}^{(n,n+1)}.
\end{equation}
The identity operators are omitted in the following calculations, i.e., \eqref{calc-neumann-estimate-2-1} is written as~$H_{ph}^{(n+1)}=H_{ph}^{(n)}+H_{ph}^{(n,n+1)}$. Observe that applying the spectral theorem for~$H_{ph}^{(n,n+1)}$ to the term to be estimated leads to
\begin{align}		
&\big\|\bigl(H_{ph}^{(n,n+1)}+\rho_{n+1}\bigr)\tH_{g_0,g}^{(n)}(z)^{-1}\big\|\nonumber\\
&=\big\|\bigl(H_{ph}^{(n,n+1)}+\rho_{n+1}\bigr)\bigl[2-z+e^{-\theta}\bigl(H_{ph}^{(n)}+H_{ph}^{(n,n+1)}\bigr)-g_0^2W_{0,0}\bigl(H_{ph}^{(n)}+H_{ph}^{(n,n+1)}\bigr)\nonumber\\
&\quad-g^2\bigl(\tW_{2,0}^{(n)}\bigl(H_{ph}^{(n)}+H_{ph}^{(n,n+1)}\bigr) +\tW_{1,1}^{(n)}\bigl(H_{ph}^{(n)}+H_{ph}^{(n,n+1)}\bigr)+\tW_{0,2}^{(n)}\bigl(H_{ph}^{(n)}+H_{ph}^{(n,n+1)}\bigr)\bigr)\bigr]^{-1}\big\|\nonumber\\
&=\sup_{r\in\sigma\bigl(H_{ph}^{(n,n+1)}\bigr)}\big\|(r+\rho_{n+1})\bigl[2-z+e^{-\theta}\bigl(H_{ph}^{(n)}+r\bigr)-g_0^2W_{0,0}\bigl(H_{ph}^{(n)}+r\bigr)\nonumber\\
&\quad-g^2\bigl(\tW_{2,0}^{(n)}\bigl(H_{ph}^{(n)}+r\bigr) +\tW_{1,1}^{(n)}\bigl(H_{ph}^{(n)}+r\bigr)+\tW_{0,2}^{(n)}\bigl(H_{ph}^{(n)}+r\bigr)\bigr)\bigr]^{-1}\big\|\nonumber\\
&=\sup_{r\in\sigma\bigl(H_{ph}^{(n,n+1)}\bigr)}\big\|(r+\rho_{n+1})H_{g_0,g}^{(n)}(z-e^{-\theta}r)\big\|.\label{calc-neumann-estimate-2-2}
\end{align}
The result can directly be estimated using Lemma \ref{feshbach-norm}. By the definition of the operator, it is $\sigma\bigl(H_{ph}^{(n,n+1)}\bigr)\subset\{0\}\cup[\rho_{n+1},\infty)$ such that the term can be split up, leading to
\begin{align}		
&\sup_{r\in\sigma\bigl(H_{ph}^{(n,n+1)}\bigr)}\big\|(r+\rho_{n+1})H_{g_0,g}^{(n)}(z-e^{-\theta}r)^{-1}\big\|\nonumber\\
&=\max\Bigl\{\rho_{n+1}\big\|H_{g_0,g}^{(n)}(z)^{-1}\big\|,\sup_{r\geq\rho_{n+1}}(r+\rho_{n+1})\big\|H_{g_0,g}^{(n)}(z-e^{-\theta}r)\big\|\Bigr\}\nonumber\\
&\leq\max\Bigl\{\frac{45\ C_n}{\sin(\vartheta)},\ \sup_{r\geq\rho_{n+1}}(r+\rho_{n+1})\frac{25\ C_n}{\frac{\sin(\vartheta)}{2}\rho_{n+1}+|E_g^{(n)}-z+e^{-\theta}r|}\Bigr\}\label{calc-neumann-estimate-2-3}
\end{align}
using the restrictions on $z$ from the assumptions of the lemma. For the second entry, a geometric argument is used to obtain the desired bound. Recall the shape of $\widetilde{\mathcal{A}}_n$ from Figure~\ref{fig-set-a-n-tilde} and note that the term can be interpreted as the distance between the line~$\{e^{-\theta}r|r\geq\rho_{n+1}\}$ and the set consisting of the possible values for $z-E_g^{(n)}$. By the definition of the set $\widetilde{\mathcal{A}}_n$, it is
\begin{align}\label{calc-neumann-estimate-2-5}		
\mathrm{Im}\bigl(z-E_g^{(n)}\bigr)&\geq\mathrm{Im}\bigl(E_g^{(n)}\bigr)-\frac{\sin(\vartheta)}{2}\rho_{n+1}-\mathrm{Im}\bigl(E_g^{(n)}\bigr)=-\frac{\sin(\vartheta)}{2}\rho_{n+1},
\end{align}
which suggests a similar approach to Lemma~\ref{invertibility-0} in the induction hypothesis. Choose a value~${z_0\in\widetilde{\mathcal{A}}}$ such that the minimum in~\eqref{calc-neumann-estimate-2-5} is attained, i.e., $\mathrm{Im}(z_0)=-\frac{\sin(\vartheta)}{2}\rho_{n+1}$. Estimating the distance to the line $\{e^{-\theta}r|r\geq\rho_{n+1}\}$ now leads to
\begin{align}		
|z-E_g^{(n)}-e^{-\theta}r|&\geq\Big|-\frac{\sin(\vartheta)}{2}\rho_{n+1}+r\sin(\vartheta)\Big|
\geq\frac{\sin(\vartheta)}{4}(r+\rho_{n+1})\label{calc-neumann-estimate-2-6}
\end{align}
by only considering the imaginary part. This implies that
\begin{equation}\label{calc-neumann-estimate-2-7}		
(r+\rho_{n+1})\cdot\frac{25\ C_n}{\frac{\sin(\vartheta)}{2}\rho_{n+1}+\frac{\sin(\vartheta)}{4}(\rho_{n+1}+r)}\leq\frac{100\ C_n}{\sin(\vartheta)}
\end{equation}
holds uniformly in $r\geq\rho_{n+1}$ when the estimate is applied to the second term in~\eqref{calc-neumann-estimate-2-3}. The bound can now be assembled to yield
\begin{equation}\label{calc-neumann-estimate-2-8}		
\big\|\bigl(H_{ph}^{(n,n+1)}+\rho_{n+1}\bigr)\tH_{g_0,g}^{(n)}(z)^{-1}\big\|\leq\max\Bigl\{\frac{45\ C_n}{\sin(\vartheta)},\ \frac{100\ C_n}{\sin(\vartheta)}\Bigr\}=\frac{100\ C_n}{\sin(\vartheta)},
\end{equation}
which is exactly~\eqref{formula-neumann-estimate-2} and, therefore, completes the proof.
\end{proof}

As the next step, the existence of $H_{g_0,g}^{(n+1)}(z)^{-1}$ is considered using a Neumann series expansion. Here, the results of Lemmas~\ref{neumann-estimate-1} and~\ref{neumann-estimate-2} can be assembled to show convergence and, therefore, establish the existence of the resolvent of $H_g^{(n+1)}$ by the isospectrality of the Feshbach-Schur map.
%
\begin{theorem}\label{invertibility-n+1}		
Let $z\in\widetilde{\mathcal{A}}_n\setminus\bigl\{E_g^{(n)}\bigr\}$ with $\big|z-E_g^{(n)}\big|\geq\frac{\sin(\vartheta)}{16}\rho_{n+1}$ and choose the parameters $g$, $\rho_0$ and $\gamma$ such that

\begin{equation}\label{parameter-restrictions-2}		
\rho_0\leq\Bigl(\frac{1}{43}\Bigr)^4,\quad \gamma<\Bigl(\frac{1120}{\sin(\vartheta)}+\frac{600}{\vartheta d\cdot\sin(\vartheta)}\Bigr)^4,\quad  g^2\cdot\frac{4100C_f^2\pi}{\sin(\vartheta)\gamma\vartheta d}<\frac{1}{2},
\end{equation}
which implies the relation $\rho_n^{\frac{1}{4}}C_n<1$. Then the norm bound
\begin{equation}\label{formula-invertibility-n+1}	
\big\|H_{g_0,g}^{(n+1)}(z)^{-1}\big\|\leq2\ \big\|\tH_{g_0,g}^{(n)}(z)^{-1}\big\|
\end{equation}
holds and $z$ lies in the resolvent set of $H_g^{(n+1)}$.
\end{theorem}

\begin{proof}		
With the invertibility of $\tH_{g_0,g}^{(n)}(z)$ already established,  we use Neumann series expansion to write
\begin{align}		
H_{g_0,g}^{(n+1)}(z)^{-1}
&=\sum_{j=1}^{\infty}\tH_{g_0,g}^{(n)}(z)^{-1}\Bigl[\Bigl(g_0^2\tW_{0,0}^{(n,n+1)}+g^2\bigl(\tW_{2,0}^{(n,n+1)}+\tW_{1,1}^{(n,n+1)}\nonumber\\
&\quad\quad\quad\quad\quad\quad\quad\quad+\tW_{0,2}^{(n,n+1)}\bigr)\Bigr)\tH_{g_0,g}^{(n)}(z)^{-1}\Bigr]^j\label{calc-invertibility-n+1-1}
\end{align}
with the series converging when the norm of the operator in square brackets is less than one. Expanding part of the term yields
\begin{align}		
&g^2\cdot\big\|\bigl(\tW_{2,0}^{(n,n+1)}+\tW_{1,1}^{(n,n+1)}+\tW_{0,2}^{(n,n+1)}\bigr)\tH_{g_0,g}^{(n)}(z)^{-1}\big\|\\
&\leq g^2\cdot\big\|\bigl(\tW_{2,0}^{(n,n+1)}+\tW_{1,1}^{(n,n+1)}+\tW_{0,2}^{(n,n+1)}\bigr)\label{calc-invertibility-n+1-2} \bigl(H_{ph}^{(n,n+1)}+\varepsilon\bigr)^{-1}\big\|\cdot\big\|\bigl(H_{ph}^{(n,n+1)}+\varepsilon\bigr) \tH_{g_0,g}^{(n)}(z)^{-1}\big\|\nonumber
\end{align}
Here, Lemmas \ref{neumann-estimate-1} and \ref{neumann-estimate-2} can be applied. Note that the choice  $\varepsilon=\rho_{n+1}$ yields the best estimate for the first term such that it can be estimated by
\begin{align}
\frac{10}{\vartheta d}\int_{\{\rho_{n+1}\leq|k|\leq\rho_n\}}\Bigl(\frac{1}{\omega(k)}+\frac{1}{\varepsilon}\Bigr)\big|G_{\theta}(k)\big|^2\ \dx^3 k\leq\frac{20}{\vartheta d\rho_{n+1}}\int_{\{\rho_{n+1}\leq|k|\leq\rho_n\}}\big|G_{\theta}(k)\big|^2\ \dx^3 k.
\end{align}
Next, the definition of $G_{\theta}$ is inserted and the integrand is estimated once more. The integral can then be calculated using spherical coordinates, leading to
\begin{align}		
\int_{\{\rho_{n+1}\leq|k|\leq\rho_n\}}\big|G_{\theta}(k)\big|^2\ \dx^3 k
&=2\pi\ C_f^2\ (\rho_n^2-\rho_{n+1}^2)\leq2\pi\ C_f^2\ \rho_n^2.\label{calc-invertibility-n+1-3}
\end{align}
Here, $C_f$ denotes the maximal value of $|f(k)|$ for $k\in\overline{B(0,\rho_0)}$. In particular, the estimate for $f$ is uniform in $n$ and, therefore, yields a constant that is independent of the infrared cutoff. The constants can now be assembled, implying
\begin{align}		
g^2\cdot \big\|\bigl(\tW_{2,0}^{(n,n+1)}+\tW_{1,1}^{(n,n+1)}+\tW_{0,2}^{(n,n+1)}\bigr)\tH_{g_0,g}^{(n)}(z)^{-1}\big\|
&\leq g^2\cdot\frac{4000C_f^2\pi}{\sin(\vartheta)\gamma\vartheta d}\cdot C_n\cdot\rho_n.\label{calc-invertibility-n+1-4}
\end{align}
The last term to be estimated is the one involving $W_{0,0}^{(n,n+1)}$. It can be expanded in the same way. Here, the bound follows from Lemma~\ref{neumann-estimate-2}, the spectral theorem as well as the constants and intermediate results from the previous calculation, leading to
\begin{align}		
\big\|W_{0,0}^{(n,n+1)}\bigl(H_{ph}^{(n,n+1)}+\rho_{n+1}\bigr)^{-1}\big\|
&\leq\int\sup_{s\geq\rho_{n+1}}\frac{\1_{B_n\setminus B_{n+1}}(k)\big|G_{\theta}(k)G_{\overline{\theta}(k)}\big|}{|e^{-\theta}s-z|}\ \dx k\cdot\sup_{r\geq\rho_{n+1}}\frac{1}{r+\rho_{n+1}}\nonumber\\
&\leq\frac{4\pi C_f^2}{\vartheta d}\ \frac{1}{2}\rho_n^2\cdot \frac{1}{2\rho_{n+1}}=\frac{\pi C_f^2}{\gamma\vartheta d}\cdot\rho_n.\label{calc-invertibility-n+1-6}
\end{align}
Collecting all the constants now yields
\begin{align}		
&g^2\cdot \big\|\bigl(W_{0,0}^{(n,n+1)}+\tW_{2,0}^{(n,n+1)}+\tW_{1,1}^{(n,n+1)}+\tW_{0,2}^{(n,n+1)}\bigr)\tH_{g_0,g}^{(n)}(z)^{-1}\big\|\nonumber\\
&\leq g^2\cdot\frac{4100\pi C_f^2}{\sin(\vartheta)\gamma\vartheta d}\cdot C_n\ \rho_n<\frac{1}{2}\cdot\rho_n^{\frac{3}{4}}<\frac{1}{2}.\label{calc-invertibility-n+1-7}
\end{align}
by choosing $g_0=g$ and inserting~\eqref{parameter-restrictions-2}. In total, the calculations establish an exponential decay in the form of a positive power of $\rho_n$ that allows to balance out the exponential growth of the constant $C_n$ for increasing $n$. This implies the convergence of the Neumann series as well as the bound
\begin{equation}\label{calc-invertibility-n+1-8}	
\big\|H_{g_0,g}^{(n+1)}(z)^{-1}\big\|\leq2\ \big\|\tH_{g_0,g}^{(n)}(z)^{-1}\big\|
\end{equation}
by applying the above estimate and calculating the value of the convergent geometric series. The existence of $(H_g^{(n+1)}-z)^{-1}$ now follows from the isospectrality of the Feshbach-Schur map.
\end{proof}

Theorem~\ref{invertibility-n+1} is crucial to the induction as it allows for a single restriction on the coupling constant for any $n\in\N$, thus avoiding an adjustment that depends on the infrared cutoff. As $(H_g^{(n+1)}-z)^{-1}$ is well-defined for $z\in\gamma^{(n+1)}_n$, the existence of the projection $P^{(n+1)}$ follows.

\subsubsection*{Estimating the Difference of the Projections}
As the difference of $\tP^{(n)}$ and $P^{(n+1)}$ can be stated as an operator-valued integral, we start by estimating the integrand.
%
\begin{theorem}\label{estimate-difference-resolvents}	
Let $z\in\widetilde{\mathcal{A}}_n\setminus\bigl\{E_g^{(n)}\bigr\}$ with $\big|z-E_g^{(n)}\big|\geq\frac{\sin(\vartheta)}{16}\rho_{n+1}$ and assume that~\eqref{parameter-restrictions-2} holds. Further, choose the coupling constant $g$ such that
\begin{equation}\label{g-restriction-2}
g<\Bigl[C_{FS}\Bigl(\Big\|\frac{G}{\sqrt{\omega}}\Big\|_{L^2}+\|G\|_{L^2}\Bigr)\Bigr]^{-1}.
\end{equation}
Then the following norm bound holds true
\begin{equation}\label{formula-estimate-difference-resolvents}		
\big\|\bigl(H_g^{(n+1)}-z\bigr)^{-1}-\bigl(H_g^{(n)}-z\bigr)^{-1}\big\|\leq g\cdot \frac{600\sqrt{\pi}}{\gamma\bigl(\sin(\vartheta)\bigr)^2}\ C_{FS}\ C_f\cdot\rho_n^{-\frac{3}{4}} +\frac{400}{\gamma\bigl(\sin(\vartheta)\bigr)^2}\cdot \rho_n^{-\frac{1}{2}}.
\end{equation}
\end{theorem}

\begin{proof}		
The existence of the operator is ensured by Theorem \ref{invertibility-n+1}. With $H_{g_0,g}^{(n+1)}(z)$ and its norm already calculated, one can now use the isospectrality of the Feshbach-Schur map to transfer the results to the actual resolvent, leading to
\begin{align}		
&\bigl(H_g^{(n+1)}-z\bigr)^{-1}-\bigl(H_g^{(n)}-z\bigr)^{-1}\nonumber\\
&=\bigl[\P-\overline{\P}\bigl(e^{-\theta}H_{ph}^{(n+1)}-z\bigr)^{-1}g\sigma_1P_{\uparrow}\otimes\Phi_{\theta}^{(n+1)}\bigr]H_{g_0,g}^{(n+1)}(z)^{-1}\nonumber\\
&\quad\ \bigl[\P-gP_{\uparrow}\sigma_1\otimes\Phi_{\theta}^{(n+1)}\bigl(e^{-\theta}H_{ph}^{(n+1)}-z\bigr)^{-1}\overline{\P}\bigr]\nonumber\\
&\quad-\bigl[\P-\overline{\P}\bigl(e^{-\theta}H_{ph}^{(n+1)}-z\bigr)^{-1}g\sigma_1P_{\uparrow}\otimes\Phi_{\theta}^{(n)}\bigr]\tH_{g_0,g}^{(n)}(z)^{-1}\nonumber\\
&\quad\ \bigl[\P-gP_{\uparrow}\sigma_1\otimes\Phi_{\theta}^{(n)}\bigl(e^{-\theta}H_{ph}^{(n+1)}-z\bigr)^{-1}\overline{\P}\bigr].\label{calc-estimate-difference-resolvents-1}
\end{align}
Note that the summand appearing at the end of the respective terms cancels out due to it equalling $(e^{-\theta}H_{ph}^{(n+1)}-z)^{-1}$ for both operators. Note that the first terms that result from expanding the brackets can be written as
\begin{equation}\label{calc-estimate-difference-resolvents-2}	
\P H_{g_0,g}^{(n+1)}(z)^{-1}\P-\P\tH_{g_0,g}^{(n)}(z)^{-1}\P=\P\bigl[H_{g_0,g}^{(n+1)}(z)^{-1}-\tH_{g_0,g}^{(n)}(z)^{-1}\bigr]\P.
\end{equation}
where the operator in its middle is given by the difference of the inverses of the Feshbach-Schur map of $H_g^{(n+1)}$ and $\tH_g^{(n)}$ respectively. This form is also aimed for when treating the other terms. We rewrite
\begin{align}		
&-\P H_{g_0,g}^{(n+1)}(z)^{-1}gP_{\uparrow}\sigma_1\otimes\Phi_{\theta}^{(n+1)}\bigl(e^{-\theta}H_{ph}^{(n+1)}-z\bigr)^{-1}\overline{\P}\nonumber\\
&\quad+\P\tH_{g_0,g}^{(n)}(z)^{-1}gP_{\uparrow}\sigma_1\otimes\Phi_{\theta}^{(n)}\bigl(e^{-\theta}H_{ph}^{(n+1)}-z\bigr)^{-1}\overline{\P}\nonumber\\
&=-\P\bigl[H_{g_0,g}^{(n+1)}(z)^{-1}-\tH_{g_0,g}^{(n)}(z)^{-1}\bigr]gP_{\uparrow}\sigma_1\otimes\Phi_{\theta}^{(n+1)}\bigl(e^{-\theta}H_{ph}^{(n+1)}-z\bigr)^{-1}\overline{\P}\nonumber\\
&\quad-\P\tH_{g_0,g}^{(n)}(z)^{-1}gP_{\uparrow}\sigma_1\otimes\bigl[\Phi_{\theta}^{(n+1)}-\Phi_{\theta}^{(n)}\bigr]\bigl(e^{-\theta}H_{ph}^{(n+1)}-z\bigr)^{-1}\overline{\P}.\label{calc-estimate-difference-resolvents-3}
\end{align}
Similarly, it follows that
\begin{align}		
&-\overline{\P}\bigl(e^{-\theta}H_{ph}^{(n+1)}-z\bigr)^{-1}g\sigma_1P_{\uparrow}\otimes\Phi_{\theta}^{(n+1)}H_{g_0,g}^{(n+1)}(z)^{-1}\P\nonumber\\
&\quad+\overline{\P}\bigl(e^{-\theta}H_{ph}^{(n+1)}-z\bigr)^{-1}g\sigma_1P_{\uparrow}\otimes\Phi_{\theta}^{(n)}\tH_{g_0,g}^{(n)}(z)^{-1}\P\nonumber\\
&=-\overline{\P}\bigl(e^{-\theta}H_{ph}^{(n+1)}-z\bigr)^{-1}g\sigma_1P_{\uparrow}\otimes\Phi_{\theta}^{(n+1)}\bigl[H_{g_0,g}^{(n+1)}(z)^{-1}-\tH_{g_0,g}^{(n)}(z)^{-1}\bigr]\P\nonumber\\
&\quad-\overline{\P}\bigl(e^{-\theta}H_{ph}^{(n+1)}-z\bigr)^{-1}g\sigma_1P_{\uparrow}\otimes\bigl[\Phi_{\theta}^{(n+1)}-\Phi_{\theta}^{(n)}\bigr]H_{g_0,g}^{(n+1)}(z)^{-1}\P.\label{calc-estimate-difference-resolvents-4}
\end{align}
The remaining terms can be treated in the same way, leading to
\begin{align}		
&\overline{\P}\bigl(e^{-\theta}H_{ph}^{(n+1)}-z\bigr)^{-1}g\sigma_1P_{\uparrow}\otimes\Phi_{\theta}^{(n+1)}H_{g_0,g}^{(n+1)}(z)^{-1}gP_{\uparrow}\sigma_1\otimes\Phi_{\theta}^{(n+1)}\bigl(e^{-\theta}H_{ph}^{(n+1)}-z\bigr)^{-1}\overline{\P}\nonumber\\
&\quad-\overline{\P}\bigl(e^{-\theta}H_{ph}^{(n+1)}-z\bigr)^{-1}g\sigma_1P_{\uparrow}\otimes\Phi_{\theta}^{(n)}\tH_{g_0,g}^{(n)}(z)^{-1}gP_{\uparrow}\sigma_1\otimes\Phi_{\theta}^{(n)}\bigl(e^{-\theta}H_{ph}^{(n+1)}-z\bigr)^{-1}\overline{\P}\nonumber\\
&=\overline{\P}\bigl(e^{-\theta}H_{ph}^{(n+1)}-z\bigr)^{-1}g\sigma_1P_{\uparrow}\otimes\Phi_{\theta}^{(n)}\bigl[H_{g_0,g}^{(n+1)}(z)^{-1}-\tH_{g_0,g}^{(n)}(z)^{-1}\bigr]\nonumber\\
&\quad\ gP_{\uparrow}\sigma_1\otimes\Phi_{\theta}^{(n+1)}\bigl(e^{-\theta}H_{ph}^{(n+1)}-z\bigr)^{-1}\overline{\P}\nonumber\\
&\quad+\overline{\P}\bigl(e^{-\theta}H_{ph}^{(n+1)}-z\bigr)^{-1}g\sigma_1P_{\uparrow}\otimes\bigl[\Phi_{\theta}^{(n+1)}-\Phi_{\theta}^{(n)}\bigr]H_{g_0,g}^{(n+1)}(z)^{-1}\nonumber\\
&\quad\ gP_{\uparrow}\sigma_1\otimes\Phi_{\theta}^{(n+1)}\bigl(e^{-\theta}H_{ph}^{(n+1)}-z\bigr)^{-1}\overline{\P}\nonumber\\
&\quad-\overline{\P}\bigl(e^{-\theta}H_{ph}^{(n+1)}-z\bigr)^{-1}g\sigma_1P_{\uparrow}\otimes\Phi_{\theta}^{(n)}\tH_{g_0,g}^{(n)}(z)^{-1}\nonumber\\
&\quad\ gP_{\uparrow}\sigma_1\otimes\bigl[\Phi_{\theta}^{(n+1)}-\Phi_{\theta}^{(n)}\bigr]\bigl(e^{-\theta}H_{ph}^{(n+1)}-z\bigr)^{-1}\overline{\P}.\label{calc-estimate-difference-resolvents-5}
\end{align}
The first term of each result can now be regrouped to yield
\begin{align}		
&\bigl[\P-\overline{\P}\bigl(e^{-\theta}H_{ph}^{(n+1)}-z\bigr)^{-1}g\sigma_1P_{\uparrow}\otimes\Phi_{\theta}^{(n)}\bigr]\bigl[H_{g_0,g}^{(n+1)}(z)^{-1}-\tH_{g_0,g}^{(n)}(z)^{-1}\bigr)\bigr]\nonumber\\
&\quad\cdot \bigl[\P-gP_{\uparrow}\sigma_1\otimes\Phi_{\theta}^{(n+1)}\bigl(e^{-\theta}H_{ph}^{(n+1)}-z\bigr)^{-1}\overline{\P}\bigr],\label{calc-estimate-difference-resolvents-6}
\end{align}
which is the form that has been aimed for. Note that all the remaining terms now, too, involve differences. The estimate for the operators including $[\Phi_{\theta}^{(n+1)}-\Phi_{\theta}^{(n)}]$ is explicitly carried out for one of them here and can be done similarly for the three remaining terms of the same form. First, we split up the term as
\begin{align}
&\big\|\P\tH_{g_0,g}^{(n)}(z)^{-1}gP_{\uparrow}\sigma_1\otimes\bigl[\Phi_{\theta}^{(n+1)}-\Phi_{\theta}^{(n)}\bigr] \bigl(e^{-\theta}H_{ph}^{(n+1)}-z\bigr)^{-1}\overline{\P}\big\|\nonumber\\
&\leq g\cdot \big\|\tH_{g_0,g}^{(n)}(z)^{-1}\big\|\cdot\big\|[\Phi_{\theta}^{(n+1)}-\Phi_{\theta}^{(n)}\bigr]\bigl(e^{-\theta}H_{ph}^{(n+1)}-z\bigr)^{-1}\big\|,\label{calc-estimate-difference-resolvents-7}
\end{align}
recalling that it is $\|\P\|=\|\overline{\P}\|=1$ and $\|\sigma_1\|=1$. The individual factors can now be estimated by Lemma~\ref{feshbach-norm} and similar calculations to Lemma~\ref{feshbach-estimates-cutoff}. Note, however, that the function considered here is $\1_{B_n\setminus B_{n+1}}G_{\theta}$ and that the integrals, therefore, can be calculated similarly to Theorem~\ref{invertibility-n+1}, yielding a positive power of~$\rho_n$. For the terms involving~$\|H_{g_0,g}^{(n+1)}(z)^{-1}\|$, Theorem~\ref{invertibility-n+1} is used instead of Lemma~\ref{feshbach-norm}. The constants are now assembled to yield
\begin{align}		
&g\cdot \big\|\tH_{g_0,g}^{(n)}(z)^{-1}\big\|\cdot\big\|[\Phi_{\theta}^{(n+1)}-\Phi_{\theta}^{(n)}\bigr]\bigl(e^{-\theta}H_{ph}^{(n+1)}-z\bigr)^{-1}\big\|\nonumber\\
&<g\cdot \frac{100\sqrt{\pi}}{\gamma\bigl(\sin(\vartheta)\bigr)^2}\ C_{FS}\ C_f\cdot\rho_n^{-\frac{3}{4}},\label{calc-estimate-difference-resolvents-8}
\end{align}
when the restrictions on $z$ and the parameters are inserted. Note that Lemma~\ref{feshbach-estimates-cutoff} needs to be applied twice for the terms obtained from~\eqref{calc-estimate-difference-resolvents-5}, yielding an additional constant. Since these terms also include an additional power of $g$,~\eqref{g-restriction-2} allows estimating this contribution by one. The last term to be estimated is given by
\begin{align}		
&\big\|\bigl[\P-\overline{\P}\bigl(e^{-\theta}H_{ph}^{(n+1)}-z\bigr)^{-1}g\sigma_1P_{\uparrow}\otimes\Phi_{\theta}^{(n)}\bigr]\bigl[H_{g_0,g}^{(n+1)}(z)^{-1}-\tH_{g_0,g}^{(n)}(z)^{-1}\bigr]\nonumber\\
&\quad\cdot \bigl[\P-gP_{\uparrow}\sigma_1\otimes\Phi_{\theta}^{(n+1)}\bigl(e^{-\theta}H_{ph}^{(n+1)}-z\bigr)^{-1}\overline{\P}\bigr]\big\|\nonumber\\
&\leq\bigl(1+g\big\|\bigl(e^{-\theta}H_{ph}^{(n+1)}-z\bigr)^{-1}\Phi_{\theta}^{(n)}\big\|\bigr)\cdot\big\|\bigl(H_{g_0,g}^{(n+1)}(z)^{-1}-\tH_{g_0,g}^{(n)}(z)^{-1}\big\|\nonumber\\
&\quad\cdot (1+g\big\|\bigl(e^{-\theta}H_{ph}^{(n+1)}-z\bigr)^{-1}\Phi_{\theta}^{(n+1)}\big\|\bigr).\label{calc-estimate-difference-resolvents-9}
\end{align}
Here, the estimate for the first term follows directly from Lemma~\ref{feshbach-estimates-cutoff}, yielding
\begin{equation}\label{calc-estimate-difference-resolvents-10}	
g\cdot \big\|\bigl(e^{-\theta}H_{ph}^{(n+1)}-z\bigr)^{-1}\Phi_{\theta}^{(n)}\big\|\leq g\cdot C_{FS}\Bigl(\Big\|\frac{G}{\sqrt{\omega}}\Big\|_{L^2}+\|G\|_{L^2}\Bigr)<1
\end{equation}
together with the assumptions of the theorem. The same relation can be applied for the term involving $\Phi_{\theta}^{(n+1)}$, allowing to estimate both brackets in~\eqref{calc-estimate-difference-resolvents-9} by two respectively. The last term can be rewritten using the second resolvent identity since both operators appearing in the difference are of the form $(A-z)^{-1}$ for suitable operators $A$. This yields
\begin{align}		
\big\|H_{g_0,g}^{(n+1)}(z)^{-1}-\tH_{g_0,g}^{(n)}(z)^{-1}\big\|
\leq\big\|H_{g_0,g}^{(n+1)}(z)^{-1}\big\|\cdot\big\|\bigl[H_{g_0,g}^{(n+1)}(z)-\tH_{g_0,g}^{(n)}(z)\bigr]\tH_{g_0,g}^{(n)}(z)^{-1}\big\|,\label{calc-estimate-difference-resolvents-11}
\end{align}
which implies that the desired bound can be assembled from~\eqref{formula-invertibility-n+1},~\eqref{formula-feshbach-norm-n-tilde} and~\eqref{calc-invertibility-n+1-7}. This implies that
\begin{align}		
\big\|H_{g_0,g}^{(n+1)}(z)^{-1}-\tH_{g_0,g}^{(n)}(z)^{-1}\big\|<\frac{100}{\gamma\bigl(\sin(\vartheta)\bigr)^2}\cdot\rho_n^{-\frac{1}{2}},\label{calc-estimate-difference-resolvents-12}
\end{align}
which yields a bound for~\eqref{calc-estimate-difference-resolvents-9} when multiplied by four. The bounds can now be assembled to yield the constants given in the claim of the theorem.
\end{proof}

With the difference of the resolvents estimated, one can directly proceed to treat the projections defined from them.
%
\begin{corollary}[Establishing~\eqref{ind-hyp-projections-2} for the induction step]\label{estimate-difference-projections}
Under the assumptions of Theorem~\ref{estimate-difference-resolvents}, choose the coupling constant such that it additionally satisfies
\begin{equation}\label{g-restriction-3}	
g<\Bigl[\frac{3}{2}\sqrt{\pi}C_{FS}C_f\Bigr]^{-1}
\end{equation}
then the following norm bound holds
\begin{equation}\label{formula-estimate-difference-projections}	
\big\|P^{(n+1)}-\tP^{(n)}\big\|\leq \frac{100}{\sin(\vartheta)}\rho_n^{\frac{1}{4}}.
\end{equation}
Further, $P^{(n+1)}$ is a rank-one projection.
\end{corollary}

\begin{proof}
Recall that $\gamma^{(n+1)}_n$ denotes the curve along the circle with radius~$r^{(n+1)}=\frac{\sin(\vartheta)}{8}\rho_{n+1}$ centered at~$E_g^{(n)}$. It now follows that
\begin{align}
\big\|P^{(n+1)}-\tP^{(n)}\big\|=\frac{1}{2\pi}\Bigg|\Bigg|\int_{\gamma^{(n+1)}_n}\Biggl[\frac{1}{H_g^{(n+1)}-z}-\frac{1}{\tH_g^{(n)}-z}\Biggr]\dx z\Bigg|\Bigg|\nonumber\\
\leq\frac{1}{2\pi}\cdot2\pi r^{(n+1)}\big\|\bigl(H_g^{(n+1)}-z\bigr)^{-1}-\bigl(H_g^{(n)}-z\bigr)^{-1}\big\|,\label{calc-estimate-difference-projections}
\end{align}
which yields the desired estimate by applying Theorem~\ref{estimate-difference-projections}. This yields the bound
\begin{align}
r^{(n+1)}\cdot\big\|\bigl(H_g^{(n+1)}-z\bigr)^{-1}-\bigl(H_g^{(n)}-z\bigr)^{-1}\big\|
&\leq\frac{100}{\sin(\vartheta)}\rho_n^{\frac{1}{4}}
\end{align}
when~\eqref{g-restriction-3} is inserted. This implies, in particular, that~\eqref{norm-projections} holds for $m=n+1$. Using~\eqref{parameter-restrictions-1}, it now follows that
\begin{equation}
\big\|P^{(n+1)}-\tP^{(n)}\big\|\leq\frac{1}{5}\cdot \Bigl(\frac{1}{5}\Bigr)^n<1,
\end{equation}
which implies that $P^{(n+1)}$ is a rank-one projection.
\end{proof}

\subsubsection*{Difference of the Eigenvalues}
As a direct consequence of Corollary \ref{estimate-difference-projections}, the existence of an eigenvalue $E_g^{n+1}$ of $H_g^{(n+1)}$ in the complex vicinity of $E_g^{(n)}$is implied. Note that the distance of these two numbers is, by construction, smaller than $\frac{\sin(\vartheta)}{16}\rho_{n+1}$. The bound can be improved to also include the coupling constant. This is the content of the following theorem.
%
\begin{theorem}[Establishing~\eqref{ind-hyp-energies} for the induction step]\label{estimate-difference-energies}		
Let $g,\rho_0,\gamma>0$ such that the parameters satisfy the assumptions of Theorem~\ref{estimate-difference-resolvents}. Then the following bound holds true
\begin{equation}\label{formula-estimate-difference-energies}
\big|E_g^{(n+1)}-E_g^{(n)}\big|\leq 9\sqrt{\pi}C_f\cdot g\rho_n
\end{equation}
\end{theorem}

\begin{proof}		
Remembering that $P^{(n+1)}$ and $\tP^{(n)}$ project to the eigenspaces corresponding to~$E_g^{(n+1)}$ and~$E_g^{(n)}$, respectively, consider
\begin{align}		
P^{(n+1)}\bigl(H_g^{(n+1)}-\tH_g^{(n)}\bigr)\tP^{(n)}
&=\bigl(E_g^{(n+1)}-E_g^{(n)}\bigr)P^{(n+1)}\tP^{(n)},\label{calc-estimate-difference-energies-1}
\end{align}
which leads to an expression of the energy difference in terms of the operators already studied. Applying the norm and rearranging the terms leads to
\begin{align}		
\big|E_g^{(n+1)}-E_g^{(n)}\big|&\leq\frac{1}{\big\|P^{(n+1)}\tP^{(n)}\big\|}\big\|P^{(n+1)}g\Phi_{\theta}^{(n,n+1)}\tP^{(n)}\big\|.\label{calc-estimate-difference-energies-2}
\end{align}
Estimating the right-hand side now yields the desired bound. For the first factor note
\begin{equation}\label{calc-estimate-difference-energies-3}	
\big\|P^{(n+1)}\tP^{(n)}\big\|=\big\|\tP^{(n)}-\overline{P^{(n+1)}}\tP^{(n)}\big\|\geq\big\|\tP^{(n)}\big\|-\big\|\overline{P^{(n+1)}}\tP^{(n)}\big\|,
\end{equation}
where the second term may be rewritten as
\begin{align}		
\overline{P^{(n+1)}}\tP^{(n)}=\bigl(\overline{\tP^{(n)}}+\tP^{(n)}-P^{(n+1)}\bigr)\tP^{(n)}
=\bigl(\tP^{(n)}-P^{(n+1)}\bigr)\tP^{(n)}.\label{calc-estimate-difference-energies-4}
\end{align}
Using Corollary~\ref{estimate-difference-projections}, this leads to the estimate
\begin{equation}\label{calc-estimate-difference-energies-5}	
\big\|\overline{P^{(n+1)}}\tP^{(n)}\big\|\leq\big\|\tP^{(n)}-P^{(n+1)}\big\|\cdot\big\|\tP^{(n)}\big\|\leq \frac{1}{5}\ \Bigl(\frac{1}{5}\Bigr)^n\cdot\frac{3}{2}<1.
\end{equation}
Here,~\eqref{parameter-restrictions-1} was used to calculate an explicit bound. This yields
\begin{equation}\label{calc-estimate-difference-energies-6}	
\frac{1}{\big\|P^{(n+1)}\tP^{(n)}\big\|}\leq \frac{1}{\frac{3}{2}-1}=2
\end{equation}
The only term left to estimate is
\begin{align}		
\big\|P^{(n+1)}g\Phi_{\theta}^{(n,n+1)}\tP^{(n)}\big\|&=g\cdot\big\|P^{(n+1)}\sigma_1P^{(n)}\otimes a^*(G^{(n,n+1)}P_{\Omega^{(n,n+1)}}\big\|.\label{calc-estimate-difference-energies-7}
\end{align}
Here, both $P^{(n+1)}$ and $P^{(n)}$ can be bounded in norm by $\frac{3}{2}$ and it is $\|\sigma_1\|=1$. Applying the relation $a(G^{(n,n+1)})\Omega^{(n,n+1)}=0$ following from the definition of the annihilation operator now yields
\begin{equation}\label{calc-estimate-difference-energies-8}	
\big\|a^*(G^{(n,n+1)})P_{\Omega^{(n,n+1)}}\big\|\leq \sqrt{4\pi}C_f\cdot\rho_n
\end{equation}
recalling that~${C_f=\max_{z\in D(0,\rho_0)}|f(z)|}$. This implies the desired bound when the constants are assembled.
\end{proof}

\subsubsection*{Closing the Induction Step}
The only point left to consider is establishing \eqref{ind-hyp-resolvent-norm} for the induction step to close the argument. As a first step, the norm of the resolvent of $H_g^{(n+1)}$ is considered following an argument similar to~\cite[Theorem~6.29]{BachBallesterosPizzo2017}.
%
\begin{lemma}\label{resolvent-norm-1}	
Let $z\in\widetilde{\mathcal{A}}_n\setminus\bigl\{E_g^{(n)}\bigr\}$ with $\big|z-E_g^{(n)}\big|\geq\frac{\sin(\vartheta)}{16}\rho_{n+1}$ and assume that the parameters satisfy the assumptions of Theorem~\ref{estimate-difference-resolvents}. Then the following norm bound holds true
\begin{equation}\label{formula-resolvent-norm-1}
\big\|\bigl(H_g^{(n+1)}-z\bigr)^{-1}\big\|\leq \Bigl(\frac{224}{\sin(\vartheta)}+\frac{120}{\vartheta d\cdot\sin(\vartheta)}\Bigr)\frac{C_n}{\frac{\sin(\vartheta)}{2}\rho_{n+1}+|z-E_g^{(n)}|}.
\end{equation}
\end{lemma}

\begin{proof}	
By the isospectrality of the Feshbach map and Theorem~\ref{invertibility-n+1}, the resolvent of~$H_g^{(n+1)}$ exists for the values of $z$ considered. Recall that it is
\begin{align}		
&\bigl(H_g^{(n+1)}-z\bigr)^{-1}=\bigl[\P-\overline{\P}\bigl(e^{-\theta}H_{ph}^{(n+1)}-z\bigr)^{-1}g\sigma_1P_{\uparrow}\otimes\Phi_{\theta}^{(n+1)}\bigr]H_{g_0,g}^{(n+1)}(z)^{-1}\nonumber\\
&\quad\cdot \bigl[\P-gP_{\uparrow}\sigma_1\otimes\Phi_{\theta}^{(n+1)}\bigl(e^{-\theta}H_{ph}^{(n+1)}-z\bigr)^{-1}\overline{\P}\bigr]+\bigl(e^{-\theta}H_{ph}^{(n+1)}-z\bigr)^{-1}.\label{calc-resolvent-norm-1-1}
\end{align}
Considering the norm now leads to
\begin{align}		
&\big\|\bigl(H_g^{(n+1)}-z\bigr)^{-1}\big\|\leq\bigl(1+g\big\|\bigl(e^{-\theta}H_{ph}^{(n+1)}-z\bigr)^{-1}\Phi_{\theta}^{(n+1)}\big\|\bigr)\cdot\big\|H_{g_0,g}^{(n+1)}(z)^{-1}\big\|\nonumber\\
&\quad\cdot\bigl(1+g\big\|\Phi_{\theta}^{(n+1)}\bigl(e^{-\theta}H_{ph}^{(n+1)}-z\bigr)^{-1}\big\|\bigr)+\big\|\bigl(e^{-\theta}H_{ph}^{(n+1)}-z\bigr)^{-1}\big\|.\label{calc-resolvent-norm-1-2}
\end{align}
Here, the first summand can be treated similarly to Theorem~\ref{estimate-difference-resolvents}, by applying Lemma~\ref{feshbach-estimates-cutoff} and the restrictions on $g$ for the terms in the brackets as well as Theorems~\ref{invertibility-n+1} and~\ref{feshbach-norm}. This leaves only the last term of~\eqref{calc-resolvent-norm-1-2} to consider. First, it follows from the spectral theorem that
\begin{equation}\label{calc-resolvent-norm-1-4}	
\big\|\bigl(e^{-\theta}H_{ph}^{(n+1)}-z\bigr)^{-1}\big\| =\sup_{r\in\sigma\bigl(H_{ph}^{(n+1)}\bigr)}\frac{1}{|e^{-\theta}r-z|}.
\end{equation}
Recall that $\sigma(H_{ph}^{(n+1)})\subseteq\{0\}\cup[\rho_{n+1},\infty)$. Once more, the term is interpreted geometrically to obtain the desired bound. As opposed to Lemma~\ref{feshbach-estimates-cutoff}, which only required a finite bound on the norm, the estimate here needs to yield a term that can be rewritten to the match the claim of the lemma. First, consider $z$ such that~$\mathrm{Re}(z)>0$. For this subeset of $\widetilde{\mathcal{A}}_n$, the absolute value of $z$ is bounded. It can be estimated by~$\frac{5}{2}$ (recall Figure~\ref{fig-geometry-estimate-4} and Lemma~\ref{feshbach-estimates-cutoff}). Bounding~$|e^{-\theta}r-z|$ by $\vartheta d$ again, it follows that
\begin{equation}\label{calc-resolvent-norm-1-5}	
\frac{|z|}{|e^{-\theta}r-z|}\leq\frac{5}{2}\ \frac{1}{\vartheta d},\quad \mathrm{Re}(z)>0,
\end{equation}
which can be rearranged to yield
\begin{equation}\label{calc-resolvent-norm-1-6}	
\frac{1}{|e^{-\theta}r-z|}\leq\frac{5}{2\vartheta d}\ \frac{1}{|z|}=\frac{5}{2\vartheta d}\ \frac{1}{|z-0|},\quad \mathrm{Re}(z)>0.
\end{equation}
Note that the estimate now explicitly includes the distance from $z$ to the origin. For the values of $z$ satisfying $\mathrm{Re}(z)\leq0$, the absolute value $|z|$ can not be bounded. Here, the same estimate as in Lemma~\ref{feshbach-estimates-cutoff} is used. Taking all values of $z$ permitted by the assumptions of the lemma into account, the bound
\begin{equation}\label{calc-resolvent-norm-1-7}		
\frac{1}{|e^{-\theta}r-z|}\leq\min\Bigl\{\frac{5}{2\vartheta d}\ \frac{1}{|z|},\ \frac{1}{\vartheta d}\Bigr\}\leq\frac{5}{2\vartheta d}\ \frac{1}{|z|}
\end{equation}
follows. The term introduced in~\eqref{calc-resolvent-norm-1-5} and~\eqref{calc-resolvent-norm-1-6} now yields the tool to rewrite the result in the desired form. Consider two points~$z_1$ and~$z_2$ in the complex plane and a real number~$\lambda>1$. Then the set of points~$z\in\C$ satisfying the equation~${\lambda|z-z_1|=|z-z_2|}$ is a circle with radius~${\frac{\lambda}{\lambda^2-1}|z_1-z_2|}$ centered at~${z_1+\frac{1}{\lambda^2-1}(z_1-z_2)}$. When the equality sign is replaced by an inequality sign and the~$z\in\C$ satisfying~${\lambda|z-z_1|\geq|z-z_2|}$ are considered, the solution set is given by the complex plane without the interior of the circle. Observe
\begin{equation}\label{calc-resolvent-norm-1-8}	
\dist\bigl(0,\widetilde{\mathcal{A}}_n\bigr)\geq\sin(\vartheta)
\end{equation}
which follows from~\eqref{calc-feshbach-estimates-cutoff-2} and the assumptions on the parameters. Next, consider $\lambda=\frac{3}{\sin(\vartheta)}$ and note that it is
\begin{align}		
\Bigl(\frac{\lambda}{\lambda^2-1}+\frac{1}{\lambda^2-1}\Bigr)\cdot\big|E_g^{(n)}\big|
&<\sin(\vartheta).\label{calc-resolvent-norm-1-9}
\end{align}
Since $E_g^{(n)}$ is included in the rectangle~${[2-\frac{\rho_0}{4},2+\frac{\rho_0}{4}]+i[-\frac{\sin(\vartheta)}{4}\rho_0,\frac{\sin(\vartheta)}{4}\rho_0]\subset\mathcal{A}}$, its absolute value was estimated by~$\frac{5}{2}$ (recall Figure~\ref{fig-geometry-estimate-4} and Lemma~\ref{feshbach-estimates-cutoff}). Considering~${z_1=0}$ and~${z_2=E_g^{(n)}}$ now yields that
\begin{equation}\label{calc-resolvent-norm-1-10}	
D\bigl({z_1+\tfrac{1}{\lambda^2-1}(z_1-z_2)},{\tfrac{\lambda}{\lambda^2-1}|z_1-z_2|}\bigr)\cap\widetilde{\mathcal{A}}_n=\emptyset.
\end{equation}
Recall that the complement of the open disk occuring in~\eqref{calc-resolvent-norm-1-9} yields exactly the values of $z$ such that the inequality
\begin{equation}\label{calc-resolvent-norm-1-11}	
\lambda|z-z_1|\geq|z-z_2|
\end{equation}
is satisfied. Inserting the values for $\lambda$, $z_1$ and $z_2$ considered, it follows that
\begin{equation}\label{calc-resolvent-norm-1-12}	
\frac{3}{\sin(\vartheta)}|z-0|\geq|z-E_n|,\quad z\in\widetilde{\mathcal{A}}_n.
\end{equation}
This can be rearranged to yield
\begin{equation}\label{calc-resolvent-norm-1-13}	
\frac{1}{|z-0|}\leq\frac{3}{\sin(\vartheta)}\ \frac{1}{|z-E_g^{(n)}|},
\end{equation}
which can be rewritten to the desired form using the restrictions on $z$. This leads to
\begin{equation}\label{calc-resolvent-norm-1-14}	
\frac{1}{|z-E_g^{(n)}|}\leq\frac{2}{\frac{\sin(\vartheta)}{16}\rho_{n+1}+|z-E_g^{(n)}|}\leq\frac{16}{\frac{\sin(\vartheta)}{2}\rho_{n+1}+|z-E_g^{(n)}|}.
\end{equation}
Inserting the results into~\eqref{calc-resolvent-norm-1-4} now yields
\begin{align}		
\big\|\bigl(e^{-\theta}H_{ph}^{(n+1)}-z\bigr)^{-1}\big\|&\leq\frac{120}{\vartheta d\cdot\sin(\vartheta)}\ \frac{1}{\frac{\sin(\vartheta)}{2}\rho_{n+1}+|z-E_g^{(n)}|},\label{calc-resolvent-norm-1-15}
\end{align}
recalling that the minimum in~\eqref{calc-resolvent-norm-1-7} was estimated against the term that can be rewritten. The desired estimate now follows by adding the bounds for the two terms in~\eqref{calc-resolvent-norm-1-2}.
\end{proof}

Building on Lemma~\ref{resolvent-norm-1}, one may now close the argument by estimating~$\|H_g^{(n+1)}\overline{P^{(n+1)}}\|$. The approach is similar to Lemma~\ref{resolvent-norm-0}. Note that considering the projection to the complement of the eigenspace allows dropping the restrictions on $z$.
%
\begin{theorem}[Establishing~\eqref{ind-hyp-resolvent-norm} for the induction step]\label{resolvent-norm-2}	
Let $z\in\mathcal{A}_{n+1}$ and assume that the parameters satisfy the assumptions of Theorem~\ref{estimate-difference-resolvents}. Further, choose the coupling constant small enough such that
\begin{equation}\label{g-restriction-4}	
g<\frac{\sin(\vartheta)}{72\sqrt{\pi}C_f}\cdot\gamma
\end{equation}
is satisfied. Then it is
\begin{equation}\label{formula-resolvent-norm-2}	
\big\|\bigl(H_g^{(n+1)}-z\bigr)^{-1}\overline{P^{(n+1)}}\big\|\leq \frac{C_{n+1}}{\frac{\sin(\vartheta)}{2}\rho_{n+1}+|z-E_g^{(n+1)}|},
\end{equation}
where the constant $C_{n+1}$ is given as in~\eqref{ind-hyp-resolvent-norm}.
\end{theorem}

\begin{proof}		
For $\big|z-E_g^{(n)}\big|>\frac{\sin(\vartheta)}{16}\rho_{n+1}$, the norm in question can directly be calculated using Lemma~\ref{resolvent-norm-1} and the bound obtained for the norm of the projections~\eqref{norm-projections} which, due to Corollary~\ref{estimate-difference-projections} also holds true for $P^{(n+1)}$. This leads to $\big\|\overline{P^{(n+1)}}\big\|\leq\frac{5}{2}$ and, therefore, yields
\begin{align}		
\big\|\bigl(H_g^{(n+1)}-z\bigr)^{-1}\overline{P^{(n+1)}}\big\|
&\leq\Bigl(\frac{560}{\sin(\vartheta)}+\frac{300}{\vartheta d\cdot\sin(\vartheta)}\Bigr)\frac{C_n}{\frac{\sin(\vartheta)}{2}\rho_{n+1}+|z-E_g^{(n)}|}.\label{calc-resolvent-norm-2-1}
\end{align}
The values for $z$ left to consider lie in the closed circle of radius $\frac{\sin(\vartheta)}{16}\rho_{n+1}$ around $E_g^{(n)}$. As $\overline{P^{(n+1)}}$ maps to the complement of the eigenspace corresponding to $E_g^{(n+1)}$ and the eigenvalue is the only spectral point in $\widetilde{\mathcal{A}}_n$, the function
\begin{equation}\label{calc-resolvent-norm-2-2}	
a(z):=\big\langle\psi,\bigl(H_g^{(n+1)}-z\bigr)^{-1}\overline{P^{(n+1)}}\phi\big\rangle
\end{equation}
is analytic when regarded as a function of $z$ with fixed $\psi,\phi\in\cH^{(n+1)}$. One can, therefore, apply the maximum modulus principle and find that the maximal value of $|a(z)|$ is to be found for some~$z$ satisfying~$|z|=\frac{\sin(\vartheta)}{16}\rho_{n+1}$. Applying the Cauchy-Schwarz inequality to the right side of~\eqref{calc-resolvent-norm-2-2} implies the estimate
\begin{align}		
\max_{|z|\leq\frac{\sin(\vartheta)}{16}\rho_{n+1}}|a(z)|&\leq \Bigl(\frac{560}{\sin(\vartheta)}+\frac{300}{\vartheta d\cdot\sin(\vartheta)}\Bigr)\frac{C_n}{\frac{\sin(\vartheta)}{2}\rho_{n+1}+|z-E_g^{(n)}|}\cdot\|\psi\|\cdot \|\phi\|,\label{calc-resolvent-norm-2-3}
\end{align}
as Lemma~\ref{resolvent-norm-1} now is applicable and calculating the norm of the operator becomes similar to the first case. The only point left to consider is rewriting the $|E_g^{(n)}-z|$ term to match the form required to complete the induction step. Note that~\eqref{g-restriction-4} implies that
\begin{equation}\label{calc-resolvent-norm-2-4}	
\big|E_g^{(n+1)}-E_g^{(n)}\big|\leq\frac{\sin(\vartheta)}{8}\rho_{n+1}<\frac{\sin(\vartheta)}{2}\rho_{n+1}.
\end{equation}
By~\eqref{ind-hyp-constants}, this implies that
\begin{align}		
\big\|\bigl(H_g^{(n+1)}-z\bigr)^{-1}\overline{P^{(n+1)}}\big\|
&\leq\frac{C_{n+1}}{\frac{\sin(\vartheta)}{2}\rho_{n+1}+|z-E_g^{(n+1)}|},\label{calc-resolvent-norm-2-5}
\end{align}
which establishes~\eqref{ind-hyp-resolvent-norm} for $m=n+1$. Note that it follows from~\eqref{calc-resolvent-norm-2-4} that~${\mathcal{A}_{n+1}\subset\widetilde{\mathcal{A}}_n}$.
\end{proof}
%
Theorem~\ref{resolvent-norm-2} closes the argument and, therefore, completes the induction step. The result of the inductive argument can be summarized as follows:
\begin{theorem}\label{induction-complete}		
Choose $\rho_0$, $\gamma$ and $g$ such that the assumptions of Corollary~\ref{estimate-difference-projections}, Theorem~\ref{estimate-difference-energies} and Theorem~\ref{resolvent-norm-2} are satisfied. Then, for every $m\in\N_0$, there exists an eigenvalue of $H_g^{(m)}$ which is the only spectral point of the operator in the set $\mathcal{A}_m$ and a rank-one projection to the corresponding eigenspace given by \eqref{ind-hyp-projections-1}. Further, the estimates~\eqref{ind-hyp-energies},~\eqref{ind-hyp-resolvent-norm}, and~\eqref{ind-hyp-projections-2} are valid for every $m\in\N_0$.
\end{theorem}

\subsubsection*{Discussion of the Bounds for the Parameters}
To conclude the section, the restrictions on the parameters needed for the inductive argument are collected to provide an overview of the assumptions necessary for considering the infrared limit. First, recall Section~\ref{sect-dilation-analyticity}, in particular Lemma~\ref {analytic-continuation}. Here, imposing a restriction on the parameter $\vartheta$ is necessary to ensure the analytic continuation of the Hamiltonian. The bound is given by
\begin{equation}\label{final-restriction-theta}	
\vartheta<\frac{\pi}{6}.
\end{equation}
Next, consider the infrared cutoff scale and recall that $\rho_n=\rho_0\gamma^n$. By definition, it is~${\rho_0,\gamma\in(0,1)}$. At the beginnig of the induction step, however, one needs to choose~${\gamma<\frac{1}{2}}$ to ensure that~$\widetilde{\mathcal{A}}_n\subset\mathcal{A}_n$. The main restrictions on the parameters stem from~\eqref{parameter-restrictions-1}, which allows to bound the projections~$P^{(n)}$ by $\frac{3}{2}$ in norm, and~\eqref{parameter-restrictions-2}, which sets up the relation between~$\rho_n$ and the constants~$C_n$ and, therefore, is crucial to estimating the norm of the difference of the projections. Note that the restrictions can easily be satisfied simultaneously by choosing the smaller value. This leads to
\begin{align}	
\rho_0&<\min\Bigl\{\Bigl(\frac{\sin(\vartheta)}{100}\Bigr)^4,\ \Bigl(\frac{1}{43}\Bigr)^4\Bigr\},\label{final-restriction-rho-0}\\
\gamma&<\min\Bigl\{\Bigl(\frac{1}{5}\Bigr)^4,\ \Bigl(\frac{1120}{\sin(\vartheta)}+\frac{600}{\vartheta d\cdot\sin(\vartheta)}\Bigr)^4\Bigr\}.\label{final-restriction-gamma}
\end{align}
In particular, it is $\rho_0=\rho_0(\vartheta)$ and $\gamma=\gamma(\vartheta)$, i.e., both parameters depend on $\vartheta$ and can only be chosen after the first parameter has been fixed. Finally, consider the bounds on the coupling constant. Here, the first bound~\eqref{g-restriction-1} originated form the induction basis and ensures the convergence of the first Neumann series expansion. Also, recall that another restriction on $g$ is included in~\eqref{parameter-restrictions-2}, which is necessary to establish the convergence of the Neumann series expansion in the induction step. In~\eqref{g-restriction-2} and~\eqref{g-restriction-3}, the coupling constant is adjusted again to compensate the constants occuring in connection with considering the Feshbach-Schur map. This is another crucial ingredient to establishing the bound on the difference of the projections, but is also used in obtaining a bound on~$\|(H_g^{(n+1)}-z)^{-1}\|$ and, therefore, necessary to close the argument. At this point, it is also required to consider~\eqref{g-restriction-4}, since it ensures that $\mathcal{A}_{n+1}\subset\widetilde{\mathcal{A}}_n$, which is required to complete the induction step. Note that, in total, five restrictions are imposed on $g$. They can, again, be satisfied simultaneously by considering the smallest value, yielding
\begin{align}
g&<\min\Biggl\{\frac{1}{6}\Bigl[\frac{8}{\sin(\vartheta)\sqrt{2\rho_0}}\Bigl(\Big\|\frac{G_{\theta}}{\sqrt{\omega}}\Big\|_{L^2}+\frac{1}{\sqrt{\rho_0}}\|G_{\theta}\|_{L^2}\Bigr)\Bigr]^{-1},\ \frac{1}{4}\sqrt{\frac{\sin(\vartheta)\gamma\vartheta d}{4100C_f^2\pi}},\nonumber\\
&\quad\quad\quad\quad\Bigr[C_{FS}\Bigl(\Big\|\frac{G_{\theta}}{\sqrt{\omega}}\Big\|_{L^2}+\|G_{\theta}\|_{L^2}\Bigr)\Bigr]^{-1},\ \Bigl[\frac{3}{2}\sqrt{\pi}C_{FS}C_f\Bigr]^{-1},\ \frac{\sin(\vartheta)\gamma}{72\sqrt{\pi}C_f}\Biggr\}.\label{final-restriction-g}
\end{align}
Note that $g=g(\vartheta,\rho_0,\gamma)$such that the bound for the coupling constant can be calculated from the other parameters considered. Further, observe that the restrictions on $g$ also include many of the characteristics of the model studied, which are reflected in the norm bounds of $f$, $G_{\theta}$ and $\frac{G_{\theta}}{\sqrt{\omega}}$, but also in the terms obtained from the geometric arguments, e.g., the constant $C_{FS}$.

\subsection{Proof of Theorem~\ref{main-result}: Existence of Resonances for the Non-Regularized Hamiltonian}\label{sect-main}
Recall that the aim of this chapter is to prove Theorem \ref{main-result}, which states that $H_g(\theta)$ posesses a resonance in the complex vicinity of $2$. In the following section, the tools previously collected are assembled to establish this result. Recall the sequence of projections given by

\begin{equation}\label{projections-on-h}		
P^{(n)}_{\infty}:=P^{(n)}\otimes P_{\Omega^{(n,\infty)}}
\end{equation}
where $P_{\Omega^{(n,\infty)}}$ is the projection to the Fock vacuum in $\cF[\L^2(B_n)]$ and note that they are all defined on $\cH$. The following result is a direct consequence of the inductive argument completed in the previous sections.

\begin{corollary}\label{convergences}		
Let $m\in\N_0$ and $E_g^{(m)},P^{(m)}$ be as constructed in Sections~\ref{sect-ind-basis}, \ref{sect-ind-hypothesis} and~\ref{sect-ind-step} satisfying the bounds~\eqref{ind-hyp-energies} and ~\eqref{ind-hyp-projections-2} respectively. Then the sequence of energies~$(E_g^{(n)})_n$ converges to a complex number $E_{\mathrm{res}}$ and the sequence of projections $(P^{(n)}_{\infty})$ defined in~\eqref{projections-on-h} converges to a rank-one projection which is denoted by $P_{\mathrm{res}}$.
\end{corollary}

\begin{proof}		
From the exponential bound given in \eqref{ind-hyp-energies} it follows directly that $E_g^{(n)}$ is a Cauchy sequence, thus implying the convergence due to the completeness of $\C$. The same argument may be applied to the sequence of projections using the exponential bound given in~\eqref{ind-hyp-projections-2}, as the inequality involves the part of the operators that does not consist of projections to a vacuum vector. With the convergence established, one can find an index $n_0\in\N_0$ such that $\|P^{(n_0)}_{\infty}-P_{\mathrm{res}}\|<1$. As $P^{(n)}$ and, therefore, $P^{(n)}_{\infty}$, too, is rank-one for all $n\in\N_0$, one can conclude that $P_{\mathrm{res}}$ is also rank-one.
\end{proof}

With these results established, the next step is to show that the limit $P_{\mathrm{res}}$ of the projections maps to an eigenspace of $H_g(\theta)$. Note that the corresponding eigenvalue is given by the limit of the energies $E_{\mathrm{res}}$ and that the proof only uses the closedness of the Hamiltonian on top of the convergences established in the previous corollary. The argument used is based on \cite[Theorem~3.6]{BachBallesterosKoenenbergMenrath2017}.
%
\begin{theorem}\label{limit-is-eigenvalue}		
Let $E_{\mathrm{res}}$ and $P_{\mathrm{res}}$ as established in Corollary~\ref{convergences} and suppose that the parameters satisfy~\eqref{final-restriction-theta},~\eqref{final-restriction-rho-0},~\eqref{final-restriction-gamma} as well as~\eqref{final-restriction-g}. Then the range of $P_{\mathrm{res}}$ is contained in the domain of $H_g(\theta)$ and $E_{\mathrm{res}}$ is an eigenvalue of $H_g(\theta)$ satisfying
\begin{equation}\label{formula-limit-is-eigenvalue}	
H_g(\theta)P_{\mathrm{res}}=E_{\mathrm{res}}P_{\mathrm{res}}.
\end{equation}
\end{theorem}

\begin{proof}		
First introduce the following functions
\begin{equation}\label{def-cutoff-function-counterparts}	
\omega^{(n,\infty)}=\I_{B_n}\omega,\quad G^{(n,\infty)}=\I_{B_n}G_{\theta},
\end{equation}
which are defined on the ball with radius $\rho_n$ centered at the origin. Similar to the operators including the infrared cutoff, one can introduce the operators
\begin{equation}\label{def-cutoff-operators-counterparts}	
H_{ph}^{(n,\infty)}:=H_{ph}\bigl(\omega^{(n,\infty)}\bigr),\quad W_{\theta}^{(n,\infty)}:=\sigma_1\otimes\Bigl(a\bigl(G_{\bar{\theta}}^{(n,\infty)}\bigr)+a^*\bigl(G_{\theta}^{(n,\infty)}\bigr)\Bigr),
\end{equation}
which act on $\cF[\L^2(B_n)]$. Since there are isomorphisms such that
\begin{equation}\label{isomorphisms-counterparts}	
\cF[\L^2(\R^3\setminus B_n)]\otimes\cF[\L^2(B_n)]\cong\cF[\L^2(\R^3)],\quad\cH^{(n)}\otimes\cF[\L^2(B_n)]\cong\cH,
\end{equation}
one may interpret them as the counterparts to the infrared-regularized operators. Recalling that~$P_n^{\infty}=P^{(n)}\otimes P_{\Omega^{(n,\infty)}}$, it follows that
\begin{equation}\label{calc-limit-is-eigenvalue-1}	
\bigl(\I_{\cF^{(n)}}\otimes H_{ph}^{(n,\infty)}\bigr)P_{\infty}^{(n)}=0,
\end{equation}
as $P_{\infty}^{(n)}$ maps the part taken from $\cF[\L^2(B_n)]$ of any vector $\psi\in\cH$ to the Fock vacuum and applying $H_{ph}^{(n,\infty)}$ maps the resulting multiple of the vacuum vector to zero. It also follows
\begin{equation}\label{calc-limit-is-eigenvalue-2}	
\big\|W_{\theta}^{(n,\infty)}P_{\infty}^{(n)}\big\|\leq\frac{3}{2}\ \big\|G_{\theta}^{(n,\infty)}\big\|_{\L^2},
\end{equation}
by using that $a(G_{\theta}^{(n,\infty)})\Omega^{(n,\infty)}=0$ and estimating the remaining term. As the support of~$G_{\theta}^{(n,\infty)}$ gets smaller when $n$ increases, the right side of the inequality tends to zero when considering $n\rightarrow\infty$. Recall that
\begin{equation}\label{calc-limit-is-eigenvalue-3}	
\lim_{n\rightarrow\infty}E_g^{(n)}=E_{\mathrm{res}},\quad \lim_{n\rightarrow\infty}P_{\infty}^{(n)}=P_{\mathrm{res}}
\end{equation}
by Corollary~\ref{convergences}, such that the results can be put together to yield
\begin{equation}\label{calc-limit-is-eigenvalue-4}	
\lim_{n\rightarrow\infty}E_g^{(n)}P_{\infty}^n=E_{\mathrm{res}}P_{\mathrm{res}}.
\end{equation}
Using~\eqref{calc-limit-is-eigenvalue-1} and~\eqref{calc-limit-is-eigenvalue-2}, on the other hand, it follows that
\begin{equation}\label{calc-limit-is-eigenvalue-5}	
\lim_{n\rightarrow\infty}H_g(\theta)P_{\infty}^{(n)}=E_{\mathrm{res}}P_{\mathrm{res}},
\end{equation}
as the components of $H_g(\theta)$ have been studied seperately above. Recall that the limit of the sequence $(P_{\infty}^{(n)})_n$ is a rank-one projection from Corollary~\ref{convergences} and note the vectors constructed through~${\Psi^{(n)}=P_{\infty}^{(n)}\Psi^{(-1)}}$ with $\Psi^{(-1)}$ taken from~\eqref{def-psi-0}. Since they are nonzero by construction, one finds that $\Psi:=P_{\mathrm{res}}\Psi^{(-1)}\neq0$. This allows to identify ${\lim_{n\rightarrow\infty}\Psi^{(n)}=\Psi}$ using the definition of the sequence and the convergence of the projections. This leads to
\begin{equation}\label{calc-limit-is-eigenvalue-6}	
E_{\mathrm{res}}\Psi=\lim_{n\rightarrow\infty}H_g(\theta)\Psi^{(n)},
\end{equation}
which is similar to~\eqref{calc-limit-is-eigenvalue-5} with the operators on both sides being applied to the vector~$\Psi^{(-1)}$. As $H_g(\theta)$ is a closed operator following the argument of Lemma \ref{properties} (iii), the graph set of $H_g(\theta)$ is closed. Therefore, as the graph contains $(\Psi^{(n)},H_g(\theta)\Psi^{(n)})$ for all~$n$, it has to contain the limit $(\Psi,H_g(\theta)\Psi)$ as well which implies that $\Psi$ is an element of~$\dom(H_g(\theta))$. This yields
\begin{equation}\label{calc-limit-is-eigenvalue-7}		
H_g(\theta)\Psi=E_{\mathrm{res}}\Psi.
\end{equation}
Using that $P_{\mathrm{res}}$ is rank-one, one can conclude that
\begin{equation}\label{calc-limit-is-eigenvalue-8}		
\mathrm{ran}(P_{\mathrm{res}})\subseteq \dom(H_g(\theta))
\end{equation}
which, by linearity, establishes the claim of the theorem.
\end{proof}
%
With the existence of the eigenvalue established and it being identified as the limit of the eigenvalues of the infrared-regularized Hamiltonians, one can use~\eqref{ind-hyp-energies} to localize~$E_{\mathrm{res}}$ in the complex vicinity of $E^{(-1)}=2$. Due to the estimate additionally including the coupling constant, it is, in fact,~$E_{\mathrm{res}}\in D(2,Cg)$ for a suitable constant $C>0$, thus establishing the claim of Theorem~\ref{main-result}.

\section*{Acknowledgments}
The main result of this paper was developed as part of a master thesis at TU Braunschweig and I am very grateful to Prof. Dr. Volker Bach for supervising my work on this topic. I would also like to thank his work group, in particular Lars Menrath, for the interesting discussions on this and related subjects.

\renewcommand*{\bibname}{References}
\addcontentsline{toc}{chapter}{References}
\bibliographystyle{plain}

\end{document}